\documentclass[aps,nofootinbib,amsmath,prd,onecolumn,showpacs,superscriptaddress,groupedaddress,notitlepage]{revtex4-1}

\usepackage{hyperref}
\usepackage[american]{babel}
\usepackage{amsfonts}
\usepackage{amsmath}
\usepackage{amssymb}
\usepackage{amsthm}
\usepackage{epsfig}
\usepackage{latexsym}
\usepackage{paralist}
\usepackage{fancyhdr}
\usepackage{graphicx}
\usepackage[vcentermath]{youngtab}
\usepackage{young}
\usepackage{ytableau}
\usepackage{etex}
\usepackage{braket}
\usepackage{float}
\usepackage{slashed}
\usepackage{xcolor}
\usepackage{soul}
\usepackage{dcolumn} 
\usepackage{physics}
\usepackage{bm}
\usepackage{nicefrac}
\usepackage{tikz}
\usetikzlibrary{arrows,positioning,fit,shapes,calc}
\usepackage{caption}
\usepackage{subcaption}
\usepackage{relsize}
\usepackage[export]{adjustbox}

\newcommand{\D}{{\rm{D}}}
\newcommand{\Sp}{{\rm{S}}}

\newcommand{\B}{{\cal B}}

\newcommand{\T}{{\cal T}}

\newcommand{\simpF}{\Delta^{(4)}}
\newcommand{\simpT}{\Delta^{(2)}}

\newcommand*\xbar[1]{%
  \hbox{%
    \vbox{%
      \hrule height 0.5pt % The actual bar
      \kern0.5ex%         % Distance between bar and symbol
      \hbox{%
        \kern-0.1em%      % Shortening on the left side
        \ensuremath{#1}%
        \kern-0.1em%      % Shortening on the right side
      }%
    }%
  }%
}

\newtheorem{definition}{Def.}
\newtheorem{theorem}{Theorem}
\newtheorem{construction}{Construction}
\newtheorem{corollary}{Corollary}
\newtheorem{proposition}{Proposition}

\def\bea{\begin{eqnarray}} 
\def\eea{\end{eqnarray}}
\def\be{\begin{equation}} 
\def\ee{\end{equation}} 
\def\ba{\begin{array}}
\def\ea{\end{array}}

\def\be{\begin{equation}}
\def\ee{\end{equation}}
\def\bea{\begin{eqnarray}}
\def\eea{\end{eqnarray}}

\begin{document}
%%%%%%%%%%%%%%%%%%%%%%%%%%%%%%%%%%%%%%%%%%%%%%%%%%%

\title{\bf Trisections in colored tensor models}

\author{Riccardo Martini}
\email{riccardo.martini@oist.jp}
\affiliation{
Okinawa Institute of Science and Technology Graduate University, 1919-1, Tancha, Onna,
Kunigami District, Okinawa 904-0495, Japan}

\author{Reiko Toriumi}
\email{reiko.toriumi@oist.jp}
\affiliation{
Okinawa Institute of Science and Technology Graduate University, 1919-1, Tancha, Onna,
Kunigami District, Okinawa 904-0495, Japan}

%%%%%%%%%%%%%%%
\begin{abstract}
%%%%%%%%%%%%%%%
%
We give a procedure to construct (quasi-)trisection diagrams for closed (pseudo-)manifolds generated by colored tensor models without restrictions on the number of simplices in the triangulation, therefore generalizing previous works in the context of crystallizations and PL-manifolds. We further speculate on generalization of similar constructions for a class of pseudo-manifolds generated by simplicial colored tensor models.
\end{abstract}


\maketitle

\tableofcontents

\renewcommand{\thefootnote}{\arabic{footnote}}
\setcounter{footnote}{0}

%%%%%%%%%%%%%%%%%%%%%%%%%%%%%%%%%%%%%%%%%%%%%%%%%%%
\section{Introduction} \label{sec:intro}
%%%%%%%%%%%%%%%%%%%%%%%%%%%%%%%%%%%%%%%%%%%%%%%%%%%

One of the most remarkable results in theoretical physics lies in random matrix models \cite{DiFrancesco:1993cyw}
whose critical limit by 't Hooft's topological expansion \cite{tHooft:1973alw} provides a universal random geometry as a Brownian map \cite{legall}, which is proven \cite{miller2015} equivalent to Liouville continuum gravity (quantum gravity with dilaton field in two-dimensions) \cite{liouville}.
Upon introduction of nontrivial dynamics, matrix models
can be shown to be mathematically rich.
The theories based on Kontsevich-type matrix models \cite{kontsevich} can be reformulated as a non-commutative quantum field theory \cite{grossewulk}, namely, the Grosse-Wulkenhaar model. 
The Grosse-Wulkanhaar model
is an  appealing quantum field theory with mathematical rigor, 
and exhibits properties like constructive renormalizability, 
asymptotic safety \cite{gwasympsafe}, integrability \cite{gwintegrable}, and Osterwalder-Schrader positivity \cite{gwos}.

Tensor models are higher rank analogues of such random matrix models, which therefore lend themselves well to be a candidate to produce even more remarkable results  for higher dimensional random geometry and quantum gravity \cite{tensortrack, Gurau:2016cjo, Guraubook}.
Colored tensor models \cite{Gurau:2011xp} in particular, are shown to represent fluctuating piecewise-linear (PL) pseudomanifolds
via their perturbative expansion in  Feynman graphs encoding topological spaces \cite{Bandieri:1982}. Colored tensor models admit a $1/N$ expansion of the partition function
\cite{GuraulargeN} with a resummable leading order, given by melonic graphs \cite{Bonzom:2011zz},  exhibiting critical behavior and a continuum limit \cite{Bonzom:2011zz}.
Melonic graph amplitudes satisfy a Lie algebra encoded in the large $N$ limit of the Schwinger-Dyson equations for tensor models \cite{Gurau:2011tj}.
Nonperturbative aspects such as Borel summability  \cite{borel} and topological recursion \cite{tr} are also studied.

Tensor models also provide a very interesting platform to explore new types of quantum/statistical field theory, owing to their non-local interactions and their vast combinatorics. 
As with matrix models, the combinatorial nature of tensor models can be enriched by introducing differential operators such that the resulting theory contains nontrivial dynamics.
Consequently, the statistical model acquires a notion of scale and its $1/N$ expansion can be translated into a renormalization group flow of the theory.
A  series of analyses and results to understand the renormalization group flow can be found in the works  \cite{bengeloun, Carrozza:2012uv, Carrozza:2014rba, Carrozza:2016vsq}. Different methods have been developed to accommodate the non-local nature of tensor models coming from combinatorics, such as dimensional regularization \cite{BenGeloun:2014qat}  and $4-\epsilon$ expansion \cite {Carrozza:2014rya}.
Having a formulation of renormalization group flow, one can then search for non-trivial fixed points, e.g., via functional renormalization group \cite{frg} and check their stability via Ward-Takahashi identities \cite{ward}. 
Other nonperturbative studies include Polchinski equations \cite{Krajewski:2015clk}.
Moreover, in recent years, tensor models have found a new avenue of research in holography via the
large $N$ melonic limit, which is
shared with the Sachdev-Kitaev-Ye model \cite{witten}.
Indeed, tensor models are a conceptually and computationally powerful tool not only to address random geometric problems but also problems in holography \cite{syk}, non-local quantum and statistical field theories, artificial intelligence \cite{Lahoche:2020txd}, turbulence \cite{Dartois:2018kfy}, linguistics \cite{Ramgoolam:2019}, and condensed matter \cite{pspin},
and serve as a very rich playground for theoretical physicists and mathematicians alike.

In this present work, we focus on studying the topological information encoded in the graphs generated by rank-$4$ colored tensor models. 
Understanding and revealing topological information and structure of PL-manifolds generated by tensor models are important work in the context of random geometry and quantum gravity.
Of course, this present work is not the first one nor the only one to address the topological properties encoded in the PL-pseudomanifolds that colored tensor models represent.
In fact, there are precedent works examining topological spaces of tensor models \cite{Gurau:2009tw, Gurau:2010nd, Gurau:2011xp, Guraubook} e.g., homology and homotopy of the graphs have been presented.
For three-dimensions, therefore correspondingly for rank-$3$ colored tensor models, Heegaard splitting has been identified in \cite{Ryan:2011qm}.

However, this particular work of ours focuses on a novel concept, trisections in four-dimensional topology, which were recently introduced by Gay and Kirby in 2012 \cite{GayKirby}.
Trisections are a novel tool to describe $4$-manifolds by revealing the nested structure of lower-dimensional submanifolds. In particular, the trisection genus of a $4$-manifold is a topological invariant. In the context of discrete manifolds, the trisection of all standard simply connected PL 4-manifolds has been studied for example in \cite{bell2017}, 
and trisections in so-called crystallization graphs have been investigated in \cite{Casali:2019gem}.
In the former work \cite{bell2017}, they rely on Pachner moves to ensure that these submanifolds are handlebodies. 
However, in colored tensor models, we do not have the priveledge to perform Pachner moves, since they are not compatible with colors in rank-$4$ tensor models. 
In the latter work \cite{Casali:2019gem}, the study focused on crystallization graphs, which are very special graphs that ensure the connectivity of  each of the submanifolds. 
However, in tensor models we generate also graphs which are not crystallizations, and furthermore, in the continuum limit of tensor models, where we are interested in large volume and refined triangulations, we will not find crystallization graphs dominating. Hence, crystallizations have a limited applicability in tensor models.

We therefore, would like to address and formulate trisection in the colored tensor model setting in this work.

We organize our paper as follows.
In sec.~{\ref{sec:tensormodels}}, we review some key points related to colored tensor models, which our work is based on.
In particular, in sec.~\ref{sec:coloredtensormodels}, we review the construction of tensor models and the definition of their partition function.
In sec.~\ref{sec:topcolgraph}, we recall how Feynman graphs of colored tensor models can encode manifolds and what kind of topological information they store.
\\
In sec.~\ref{sec:heegaardsplitting}, we illustrate a few key concepts of three-dimensional topology  necessary to our work.
In sec.~\ref{sec:attachinghandles}, we explain how to describe manifolds via their handle decomposition and recall how, in the case of $3$-manifolds, it encodes their Heegaard splitting.
Section~\ref{sec:connected-sum} analyzes the behavior of Heegaard splittings under connected sum, which will be of great importance in the later part of the paper, while in sec.~\ref{sec:jackets-heeg} and \ref{sec:more-heeg-split}, we review two constructions of Heegaard surfaces that are known in the literature and  are based on combinatorial methods.
\\
Sec.~\ref{sec:trisections}, finally, is dedicated to the construction of trisections.
After introducing the concept of trisection for smooth $4$-manifolds, in sec.~\ref{sec:stab} we review a particular kind of move, known as stabilization, and highlight some features that stabilization shares with connected sum of trisections.
In sec.~\ref{sec:cutting-simplices}, we focus on how to partition the vertices of a $4$-simplex in three sets, which is the starting point of our construction of trisections.
In sec.~\ref{sec:split4bubbles}, we study the structure obtained in a PL-manifold via our combinatorial construction and point out what kind of problems are encountered for a generic graph of a four-dimensional manifold. From this point onward, our work departs from previous results studying trisections via triangulations of $4$-manifolds.
In sec.~\ref{sec:connect4bubbles}, finally, we show how the information about trisection can be extracted from the colored graph of reank-$4$ colored tensor models.
Sec.~\ref{sec:4dhandlebodies} and sec.~\ref{sec:central-surface} elaborate on the analyses of the result. In particular we prove that, indeed, we split the manifold under investigation into three four-dimensional handlebodies and we analyze the trisection diagram generated with our procedure.
Sec.~\ref{sec:pseudo-mfd} addresses relaxation of the hypothesis of  graphs dual to manifolds and illustrate in some cases that it is possible to draw a few topological conclusions for a wider class of graphs.
\\
Finally, in sec.~\ref{sec:conclusions}, we summarize our results and point out a few possible future directions which may benefit from our present work.

%%%%%%%%%%%%%%%%%%%%%%%%%%%%%%%%%%%%%%%%%
\section{Tensor models}
\label{sec:tensormodels}
%%%%%%%%%%%%%%%%%%%%%%%%%%%%%%%%%%%%%%%%%

%%%%%%%%%%%%%%%%%%%%%%%%%%%%%%%%%%%%%%%%%%%%%%%%%%%
\subsection{$(d+1)$-colored tensor models} 
\label{sec:coloredtensormodels}
%%%%%%%%%%%%%%%%%%%%%%%%%%%%%%%%%%%%%%%%%%%%%%%%%%%

In this section, we introduce tensor models, and in particular colored tensor models and some of their relevant objects which will be used later in order to contsruct trisections.

Tensor models are statistical theories of random tensors and can be thought as zero-dimensional field theories. Due to their low dimensionality, tensor models mostly encode combinatorial information and many of their properties can be directly imported to their higher dimensional copunterpart: tensor field theories. Colors are introduced via an extra index labeling the tensor themselves and we require the covariance of the theory to be diagonal with respect to the color indices. This last requirement will allow us to have a much greater control on the combinatorics encoded in the theory.

Besides the field content of the theory (e.g., rank of tensors and  amount of colors considered), a colored tensor model is defined in perturbation theory upon specifying a free covariance and an array of interactions (deformations around the free theory). In this paper, we restrict to a simplicial model (the meaning of this name will be clear soon). We therefore consider the multiplicative group of integers modulo $N$, $\mathbb{Z}_N$ and let $I$ be $\mathbb{Z}_N^{\times d}$ with elements ${\bf {n}} \in I$, ${\bf{n}}=\{n_1, \dots , n_d\}$ and $F(I)$ the space of complex functions on $I$. We give the following definition: 

\begin{definition}
\label{def:simplicial-tensor-model}
A $(d+1)$-colored tensor model of rank $d$ tensors is defined via a measure $d \nu$
\begin{equation}
d \nu = \prod_{i=0}^d d \mu_{C^i} (\phi^i, \bar{\phi}^i) e^{-S}, \quad S = \lambda \sum_{{\bf{n}}_i \in I} {\cal K}_{\bf{n}_0\cdots {\bf{n}}_d} \prod_{i=0}^{d} \phi^i_{{\bf{n}}_i} +{\bar \lambda} \sum_{{\bar {\bf{n}}}_i \in I} {\bar {\cal K}}_{{\bar {\bf{n}}}_0\cdots {\bar {\bf{n}}}_d} \prod_{i=0}^{d} \bar{\phi}^i_{{\bar {\bf{n}}}_i} 
\end{equation}
where
\begin{itemize}
\item $\phi^i : I \rightarrow {\mathbb C}$ are $d+1$ complex random fields;
\item $C^i$ : $F(I) \rightarrow F(I)$ are $d+1$ covariances; 
\item $ {\cal K}$, $\xbar{ {\cal K}}$ : $I^{\times (d+1)} \rightarrow \mathbb{C}$ are two vertex kernels.
\end{itemize}
\end{definition}

If $\cal K$ and $\bar{ {\cal K}}$ are such that every tensor has exactly one index ($n_i$) contracted with another tensor in the interaction, we call the model a \textit{simplicial} colored tensor model. Note that in the interaction term, every color index appears on the same footing, while the free measure factorizes in the product of single color measures. Thanks to this structure, the Feynman diagrams of a simplicial colored tensor model can be represented as {\textit colored graph}, i.e., a connected bipartite regular graph such that each line has a color in $\{0, 1, \dots, d\}$ and each node is incident to exactly one line of each color\footnote{In the following we will often have to go back and forth between graphs and triangulation. Therefore, in order to avoid confusion, we will adopt the terms \textit{node} and \textit{line} for, respectively, zero-dimensional and one-dimensional objects in a graph, while we will call \textit{vertex} and \textit{edge} a zero-dimensional and a one-dimensional object in the triangulation. When referring to edges on the boundary of two-dimensional polygons we might use the term \textit{sides}.}.

\begin{definition}
A closed $(d+1)$-colored graph is a graph $\cal G = (\cal V, \cal E)$ with node set $\cal V$ and line set $\cal E$ such that:
\begin{itemize}
\item $\cal V$ is bipartite; there is a partition of the node set ${\cal V} = V \cup \xbar V$, such that for any element $l \in \cal E$, $l= \{v, \bar v\}$ where $v\in V$ and $\bar v \in \xbar V$. The cadinalities satisfy $\vert{ \cal V} \vert = 2 \vert V \vert = 2 \vert {\xbar V} \vert$.
\item The line set is partitioned into $d+1$ subsets ${\cal E} = \cup_{i = 0}^d {\cal E}^i$, where ${\cal E}^i$ is the subset of lines with color $i$.
\item It is $(d+1)$-regular (i.e., all nodes are $(d+1)$-valent) with all lines incident to a given node having distinct colors.
\end{itemize}
\end{definition}
To distinguish, we call the elements $v \in V$ (${\bar v}\in {\xbar V}$) positive (negative) nodes and draw them with the colors clockwise (anti-clockwise) turning. We often denote by these postive (negative) nodes in colors black (white) in graphs.
The bipartition also induces an orientation on the lines, say from $v$ to ${\bar v}$.

We notice that $(d+1)$-colored graphs are dual to (colored) simplicial triangulations of piecewise linear (PL) orientable $(d+1)$-dimensional pseudomanifolds in $d$ dimensions \cite{Bandieri:1982, cristSurvey}. In particular, every node in the graph corresponds to a top dimensional simplex, every line is dual to a $(d-1)$-dimensional face and two nodes joined by a line of color $i$ represent a pair of $d$-simplices sharing the same $(d-1)$-face (i.e., an orientation reversing homeomorphism between the two boundary faces is implied). In fact, given a simplicial colored triangulation $\T$ of a PL pseudo-manifold $M$, one can consider the dual cellular decomposition $\T^*$ and notice that a colored graph is nothing but the $1$-skeleton of $\T^*$. Therefore, colored graphs are often referred to as \textit{graph encoding manifolds} (GEM), and play a fundamental role in the study of PL topological invariant from a combinatorial point of view, especially within the framework of crystallizations \cite{cristSurvey}. We remark that not every triangulation can be colored, though a refinement compatible with can always be found by means of barycentric subdivision.

We postpone a more detailed explanation of the topological description of colored graphs to the following sections. Nevertheless, it is useful to recall here how to embed a colored graphs in its dual triangulation. Consider a triangulation  $\T$ of a $4$-manifold $M$, and a colored graph $\cal G$ dual to $\T$, therefore $K(\cal G) = \T$ and $\vert K(\cal G)\vert  =M$. The most natural prescription is to embed the graph such that every component of the graph intersects its dual simplex transversally and at the barycenter. Since the graph is the $1$-skeleton of the dual cellular decomposition of $M$, it is only made of nodes and lines. Therefore, we will only have to embed nodes in the barycenter of $d$-simplices and have $i$-colored lines intersecting $i$-colored $(d-1)$-faces transversally. Examples are shown in fig.~\ref{fig:dual}. For example in four dimensions we will have nodes at the center of $4$-simplices and $i$-colored lines intersecting $i$-colored tetrahedra transversally. Though very simple, this embedding represents a very powerful tool to understand many topological properties of PL manifolds using colored graphs.

\begin{figure}[h]
    \begin{minipage}[t]{0.9\textwidth}
      \centering
\def\svgwidth{0.6\columnwidth}
\centering
\includegraphics[scale=.18]{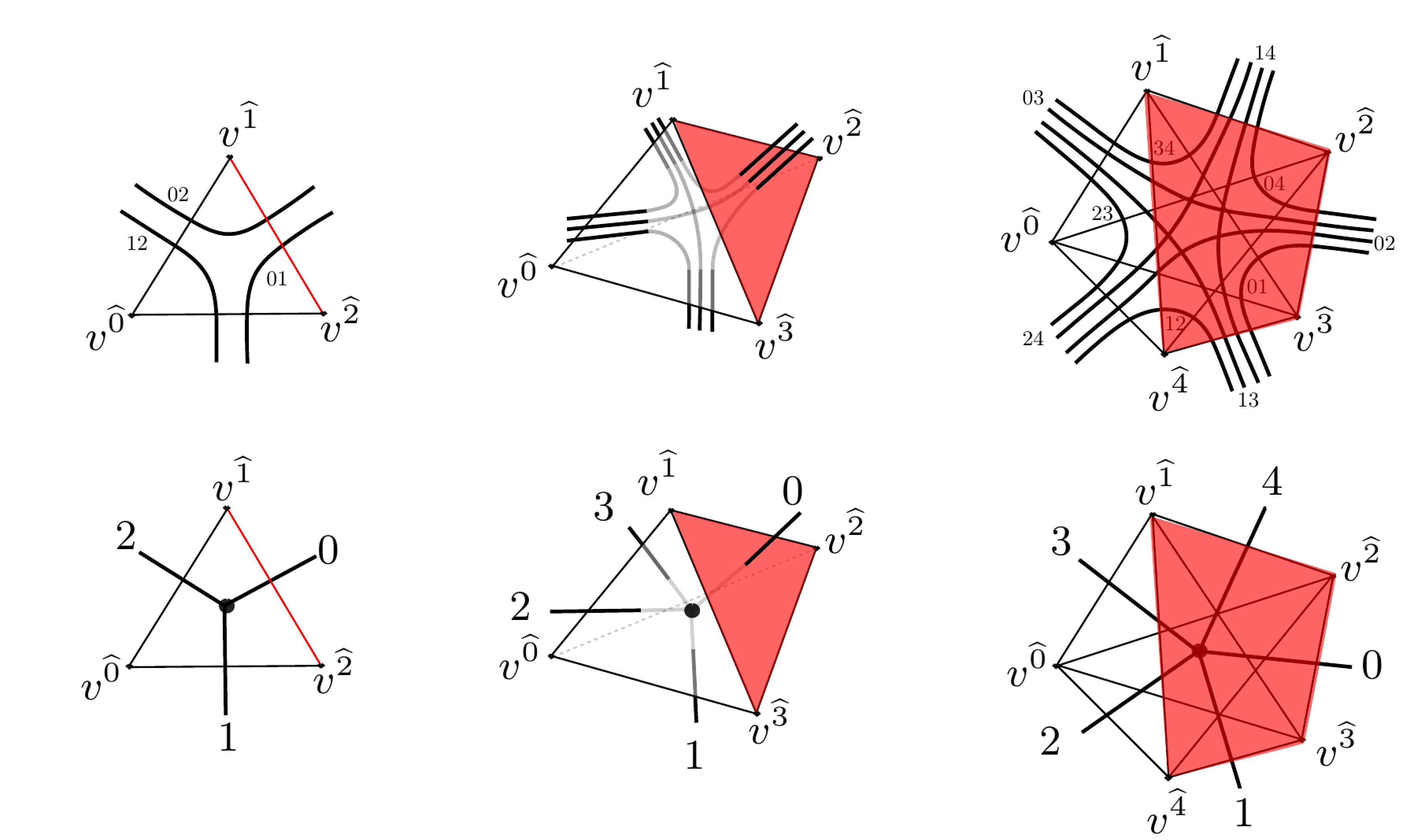}
\caption{We show $d$-simplices in $d=2, 3, 4$-dimensions, where we also embedded $d+1$-colored graphs.  From left to right, $d=2, 3, 4$, and on the top row, embedded tensor model graphs are shown in stranded representation, and on the bottom row, shown in colored representation.
We showin red,  $0$-colored faces (one-dimensional for rank $2$, two-dimensional for rank $3$, and three-dimensional for rank $4$).}
\label{fig:dual}
\end{minipage}
\end{figure}

As a final remark, we point out that bipartiteness of a colored graphs $\cal G$, which from a tensor model point of view stems from employing complex tensors and a real free covariance, implies orientability of $K({\cal G})$ \cite{Bandieri:1982}. Both in the GEM formalism and in tensor models, this condition can be relaxed if nonorientable (pseudo-)manifolds shall be considered, nevertheless, in this paper we restrict ourselves to the orientable case.

%%%%%%%%%%%%%%%%%%%%%%%%%%%%%%%%%%%%%%%%%%%%%%%%%%%
\subsection{Topology of colored graphs} 
\label{sec:topcolgraph}
%%%%%%%%%%%%%%%%%%%%%%%%%%%%%%%%%%%%%%%%%%%%%%%%%%%

As advertized, these colored graphs are extensively studied in topology especially in the form of crystallization \cite{Casali:2017tfh, Ferri:1982,Lins:1995}.
One can say that the colors therefore are responsible to encode enough topological information to construct a $d$-dimensional cellular complex, rather than the a-priori naive 1-complex of a graph. Most of the topological information is encoded within different kinds of embedded sub-complexes of $K({\cal G})$ and their combinatorial description in terms of colored graphs.

\paragraph{Bubbles.}

The first structure we present is that of \textit{bubbles}\footnote{Sometimes referred to as \textit{residues} in the literature}. Starting from a colored graph $\cal G$ dual to a colored triangulation $\T=K({\cal G})$, a $n$-bubble ${\cal B}^{i_1, \dots , i_n}_a$ is the $a$-th connected component of the subgraph spanned by the colors $i_1, \dots, i_n\in\{0,\dots, d\}$. 
In order to lighten the notation, we will indicate $d$-bubbles by their only lacking color and sometimes we will refer to them as ${\widehat{i}}$-bubble, for example in four dimensions we might consider the $\widehat{0}$-bubble $\B_a^{\widehat{0}}=\B_a^{1, 2, 3, 4}$.
Each bubble identifies a single simplex in $\T$, in particular given a $n$-bubble ${\cal B}^{i_1, \dots , i_n}_a$, its dual $K(\B^{i_1, \dots , i_n}_a)$ is PL-homeomorphic to the link of a $(d-n)$-simplex $\sigma_a$ in the first barycentric subdivision of $\T$. Upon the embedding procedure described above, we can think about $K(\B^{i_1, \dots , i_n}_a)$ as the boundary of a $n$-dimensional submanifold of $\T$, intersecting $\sigma_a$ transversally. The most important bubbles for our work are $d$-bubbles and $2$-bubbles. $d$-bubbles represent the link of vertices ($0$-simplices) in $\T$. A standard result states that $K({\cal G})$ is a manifold if and only if all $d$-bubbles are topological spheres. $2$-bubbles will be referred to as bicolored cycles\footnote{In the tensor models literature, we often refer to bicolored cycles as faces, however, in this paper, we will keep the word faces for general simplices.}, they identify $(d-2)$-simplices (triangles in four dimensions) and are often depicted in tensor models when employing the ``stranded'' notation for Feynman graphs. From a tensor model perspective, while nodes of $\cal G$ correspond to interaction vertices and lines to free propagators of the theory, bicolored cycles come from the contraction patterns of tensor indices.

\paragraph{Jackets.}
Let $\cal G$ be a $(d+1)$-colored graph. For any cyclic permutation $\eta = \{\eta_0, \dots, \eta_d\}$ of the color set, up to inverse, there exist a regular cellular embedding of $\cal G$ into an orientable surface $\Sigma_\eta$, such that regions of $\Sigma_\eta$ are bounded by bicolored cycles labeled by $\{\eta_i, \eta_{i+1}\}$ \cite{BenGeloun:2010wbk, GuraulargeN}. Then, we define a jacket ${\cal J}_{\eta}$ as the colored graph having the same nodes and lines as $\cal G$, but only the bicolored cycles $\{\eta_i, \eta_{i+1}\}$:

\begin{definition}
A colored {\textit jacket} ${\cal J}_{\eta}$ is a $2$-subcomplex of $\cal G$, labeled by a permutation $\eta$ of the set $\{0, \dots, d\}$, such that
\begin{itemize}
\item $\cal J$ and $\cal G$ have identical node sets, ${\cal V}_{\cal J} = {\cal V}_{\cal G}$;
\item $\cal J$ and $\cal G$ have identical line sets, ${\cal E}_{\cal J} = {\cal E}_{\cal G}$;
\item the bicolored cycle set of ${\cal J}_{\eta}$ is a subset of the bicolor set of $\cal G$: ${\cal F}_{\cal J} = \{ f \in {\cal F}_{\cal G} \vert f = \{\eta_i, \eta_{i+1}\}, i\in \mathbb Z_{d+1}\}$.
\end{itemize}
\end{definition}

From a tensor model perspective, jackets are merely ribbon graphs (only comprise of nodes, lines and bicolored cycles), like the ones generated by matrix models graphs. Therefore, jackets represent embedded surfaces in the cellular complex represented by colored tensor models graphs. Let us clarify this last point. The regular embedding of  $\cal G$ into $\Sigma_{\eta}$ defines a cellular decomposition of $\Sigma_{\eta}$ with polygonal $2$-cells having $(d+1)$-sides. Each $2$-cell is dual (in $\Sigma_{\eta}$) to a node of $\cal G$ and each side is dual to a line (furthermore, every vertex is dual to a bicolored cycle $\{\eta_i, \eta_{i+1}\}$). Therefore, sides inherit the colors carried by lines of $\cal G$. One may notice that the transversal intersection of a surface with a codimension-$1$ $i$-simplex is a one dimensional edge homeomorphic to such an $i$-colored side. Therefore, we can think about $K({\cal J}_{\eta})$ as an embedding of $\Sigma_{\eta}$ in $K({\cal G})$, such that it intersects transversally all the $(d-1)$-faces. If $d>3$, the dimensionality of $\Sigma_{\eta}$ is too low to define two different regions within the top dimensional simplices. If $d = 3$, though, $\Sigma_{\eta}$ splits every top dimensional simplex and have been shown to represent Heegaard surfaces of three-dimensional PL-manifold $K(\cal G)$ \cite{Ryan:2011qm}; we will be discuss this further in section~\ref{sec:jackets-heeg}.

It is evident that $\cal J$ and $\cal G$ have the same connectivity. We note here that the number of independent jackets is $d!/2$.
We define the {\textit {Euler characteristic}} of the jackets as $\chi ({\cal J}) = 2- 2 g_{\cal J} = \vert {\cal V}_{\cal J}\vert -  \vert {\cal E}_{\cal J}\vert  + \vert {\cal F}_{\cal J}\vert$, where $g_{\cal J}$ is the genus of the jacket and corresponds to the genus of $\Sigma_{\eta}$. Note that we only define jackets for the closed colored graphs here.
We also remark that jackets are also bipartite reflecting the definition above, and therefore represent orientable surfaces.

\begin{figure}[h]
    \begin{minipage}[t]{0.9\textwidth}
      \centering
\def\svgwidth{0.9\columnwidth}
\centering
\includegraphics[scale=.2]{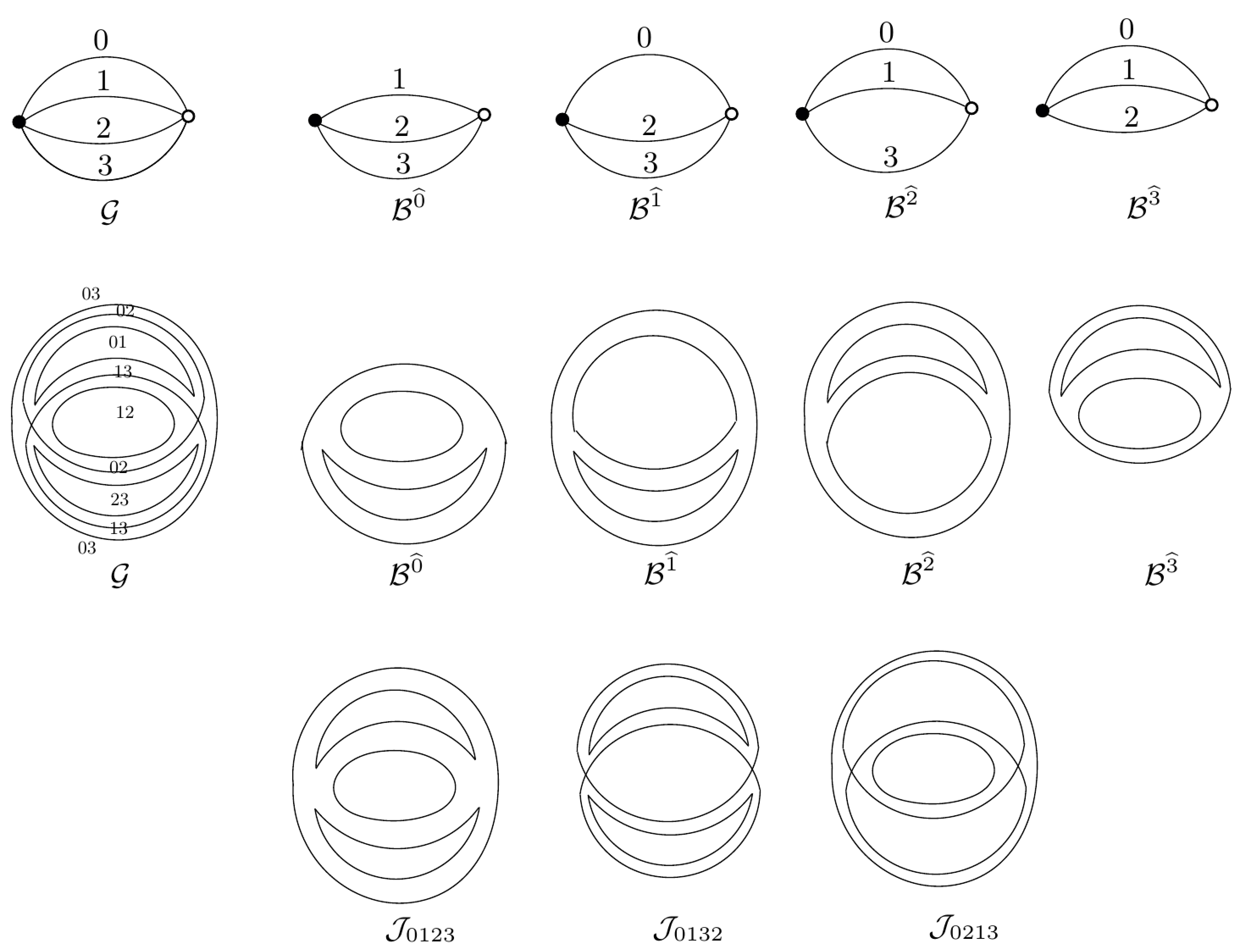}
\caption{
We show in the top row the colored representations of the elementary melon ${\cal G}$ in rank $3$ tensor bipartite colored model, its bubbles ${\cal B}^{\hat{0}}$, ${\cal B}^{\hat{1}}$, ${\cal B}^{\hat{2}}$, and ${\cal B}^{\hat{3}}$ from left to right.
In the middle row, we show the stranded representations of the same objects as the top row.
In the bottom row, we show in the stranded representation, the jackets of the elementary melon in rank $3$ tensor bipartite colored model.
}
\label{fig:bubblesjackets}
\end{minipage}
\end{figure}

\paragraph{Gurau degree.}

From a tensor model perspective, jackets play a crucial role in the large $N$ expansion of colored tensor models, as they define the so-called Gurau degree, which is the parameter that governs the large $N$ expansion. For completeness, we introduce the Gurau degree of a graph $\cal G$ as follows:
\begin{definition}
given colored graph $\cal G$ and the set of its its jackets, we define a combinatorial invariant, called {\textit {Gurau degree}}, as the sum of genera of all jackets of $\cal G$.
\begin{equation}
\omega({\cal G}) = \sum_{\cal J} g_{\cal J}.
\end{equation}
\end{definition}
It is easy to see that $\omega$ is a non-negative integer. 

A remarkable feature of Gurau degree is that if $\omega =0$, then the $K({\cal G})$ is a topological sphere, although the converse is not always true. While in $d=2$ the degree equals the genus of the triangulation dual to $\cal G$, it is not a topological invariant for $d > 2$. However, it is an important quantity in tensor models, as the classification of graphs organized by the Gurau degree allows for a $1/N$ expansion where $N$ is the size of the tensors, just like the $1/N$ expansion of matrix models according to the genus.
We defer a more detailed discussion on the large $N$ expansion of colored tensor models to other literature \cite{GuraulargeN}.

%%%%%%%%%%%%%%%%%%%%%%%%%%%%%%%%%%%%%%%%%%%%%%%%%%%
\section{Heegaard splittings of $3$-manifolds} \label{sec:heegaardsplitting}
%%%%%%%%%%%%%%%%%%%%%%%%%%%%%%%%%%%%%%%%%%%%%%%%%%%

In this section we introduce some of the concepts that are pedagogical to understanding trisections and to which we will refer often in later sections of the paper, namely handle decomposition and Heegaard splittings. We will begin defining such constructions for ojects in theTOP category (specifically for three-dimensional topological manifolds in the case of Heegaard splittings), and we will restrict later to the PL category, which is the main focus of this work.

%%%%%%%%%%%%%%%%%%%%%%%%%%%%%%%%%%%%%%%%%%%%%%%%%%%
\subsection{Attaching handles} 
\label{sec:attachinghandles}
 %%%%%%%%%%%%%%%%%%%%%%%%%%%%%%%%%%%%%%%%%%%%%%%%%%%

A handle decomposition of a closed and connected topological $d$-manifold $M$ is a prescription for the construction of $M$ by subsequently attaching handles of higher index.
We can define a $i$-handle in $d$ dimensions as a topological $d$-ball ${\D}^d$ parametrized as $\D^i\times\D^{d-i}$ and is glued to a manifold $K$ along $\Sp^{i-1}\times\D^{d-i}$, i.e., there exist an orientation reversing homeomorphism from $\Sp^{i-1}\times\D^{d-i}$ to a subset of $\partial K$. An $i$-handle can therefore be viewed as the thickening of an $i$-dimensional ball (which we call \textit{spine}); we will refer to the boundary of this ball as the \textit{attaching sphere} of the handle. 
A $(d-i)$-ball intersecting the spine transversally, will be called compression disc, and its intersection with the boundary of the handle will be referred to as \textit{belt sphere}. Note that, unless the handle decomposition of a manifold includes at least one top dimensional handle, the result will always have a boundary.

\begin{figure}[h]
    \begin{minipage}[t]{0.9\textwidth}
      \centering
\def\svgwidth{0.9\columnwidth}
\centering
\includegraphics[scale=.24]{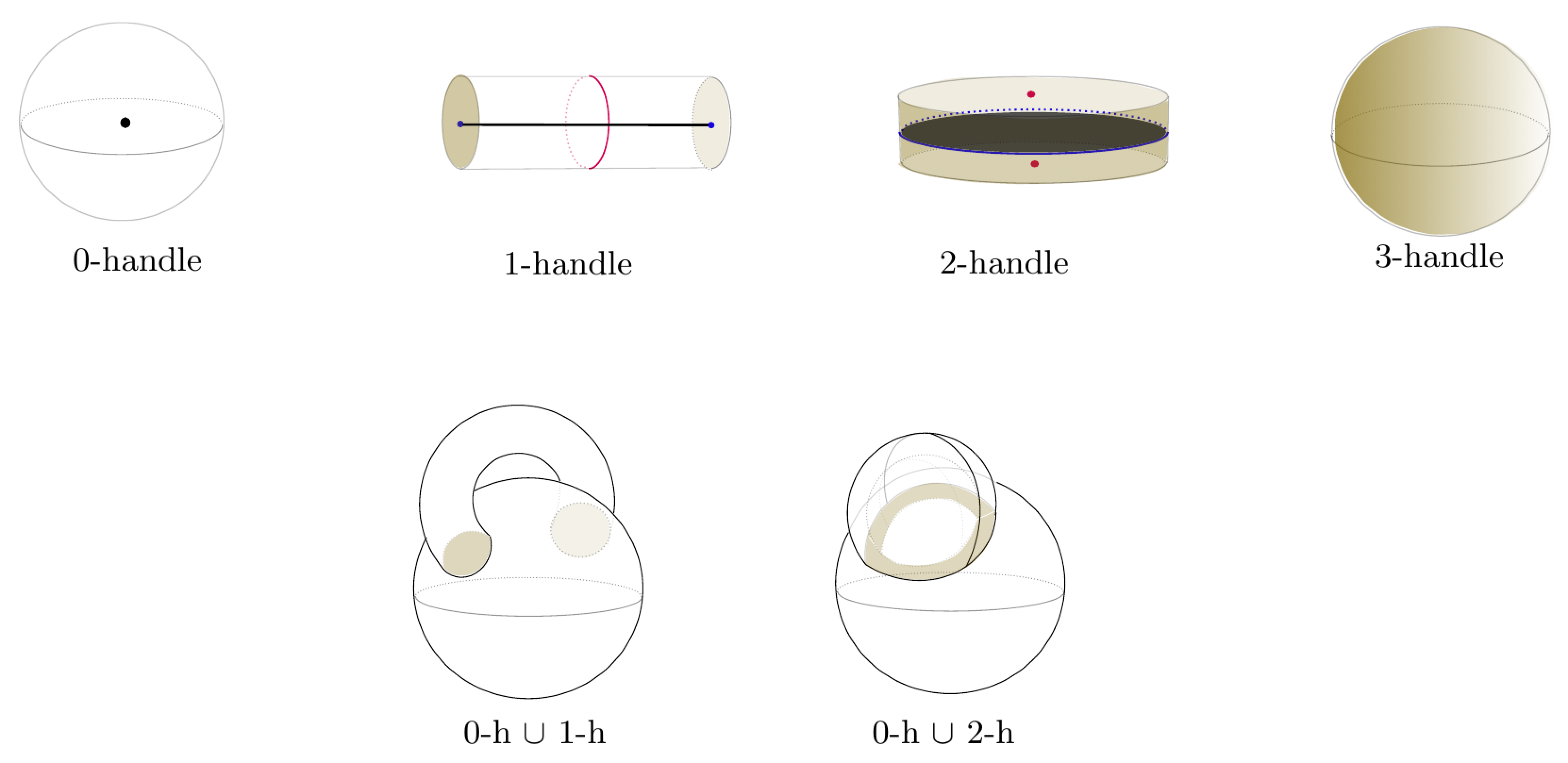}
\caption{Handles in three dimensions. 
The figure shows (from left to right) a three-dimensional a $0$-handle ($\D^3$), a $1$-handle ($\D^1 \times \D^2$ glued along $\Sp^0 \times \D^2$), a $2$-handle ($\D^2 \times \D^1$ glued along $\Sp^1 \times \D^1$), and a $3$-handle ($\D^3$ glued along $\Sp^2$). 
The gluing surfaces are colored in brown.
For $0$-handle, the spine is a point, for $1$-handle the spine is a line, and for $2$-handle the spine is a disc, all colored in solid black.
We also show in red belt spheres for $1$-handle and $2$-handle.
Lastly, we also illustrate how $1$-handles and $2$-handles attach to a $0$-handle at the gluing region which are colored in brown. 
}
\label{fig:handleanotomy}
\end{minipage}
\end{figure}

\begin{definition}\label{def:handlebody}
A \textbf{handlebody} $H$ (sometimes referred to as \textbf{1-handlebody}) is a manifold whose handle decomposition contains only a $0$-handle and $1$-handles. The genus $g$ of $H$ can be defined as the number of $1$-handles in its decomposition.
\end{definition}
Note that, if $H$ is three-dimensional, then $g$ equates the genus of $\partial  H$. Moreover, a manifold is a handlebody iff it collapses to a one-dimensional spine.

\begin{definition}\label{def:heegSplitt}
Let $H_1$ and $H_2$ be two three-dimensional handlebodies of genus $g$ and let $f$ be an orientation reversing homeomorphism from $\partial H_1$ to $\partial H_2$.  We call $(H_1, H_2, f)$ a \textbf{Heegaard splitting} of the 3-manifold $M$ if
\begin{equation}\label{eq:heegSplit}
M = H_1 \cup_f H_2\,.
\end{equation}
The common boundary $\Sigma = \partial H_1 = \partial H_2$ is then called a \textbf{Heegaard surface}.
\end{definition}

From now on, making use of a slight abuse of notation and for the sake of clarity, we will represent a Heegaard splitting with the triple $M=(H_1, H_2, \Sigma)$, by asserting $\Sigma = \partial H_1 = \partial H_2$ is provided by the homeomorphism $f$.

A Heegaard splitting allows us to represent a closed and compact 3-manifold\footnote{ 
In the present manuscript we focus on closed and orientable manifolds, nevertheless the definition of Heegaard splitting applies to a wider class of manifolds. In particular, we point out that in the case of non-orientable $3$-manifold, the Heegaard surface is non-orientable as well \cite{Rubinstein:1978}. Moreover, the definition of Heegaard splitting can be extended to manifold with boundary making use of compression bodies instead of handlebodies \cite{Meier:2016}. } $M$ via a surface and two sets of closed lines on the surface representing the homotopically inequivalent belt spheres of each handlebody. These curves, namely $\alpha$- and $\beta$-curves, encode the information on how $H_1$ and $H_2$ are glued to their boundaries. We refer to  $\alpha$- and $\beta$- curves collectively as 
\textit{attaching curves}.
The representation we just described is called a \textit{Heegaard diagram} for $M$. 
It is important to point out that cutting $\Sigma$ along the $\alpha$-curves or along the $\beta$-curves never leads to a disconnected surface, instead we obtain a $2$-sphere from which an even number of discs (two per each curve) have been removed. See fig.~\ref{fig:cutting}.

\begin{figure}[h]
    \begin{minipage}[t]{0.8\textwidth}
      \centering
\def\svgwidth{0.65\columnwidth}
\centering
\includegraphics[scale=0.4]{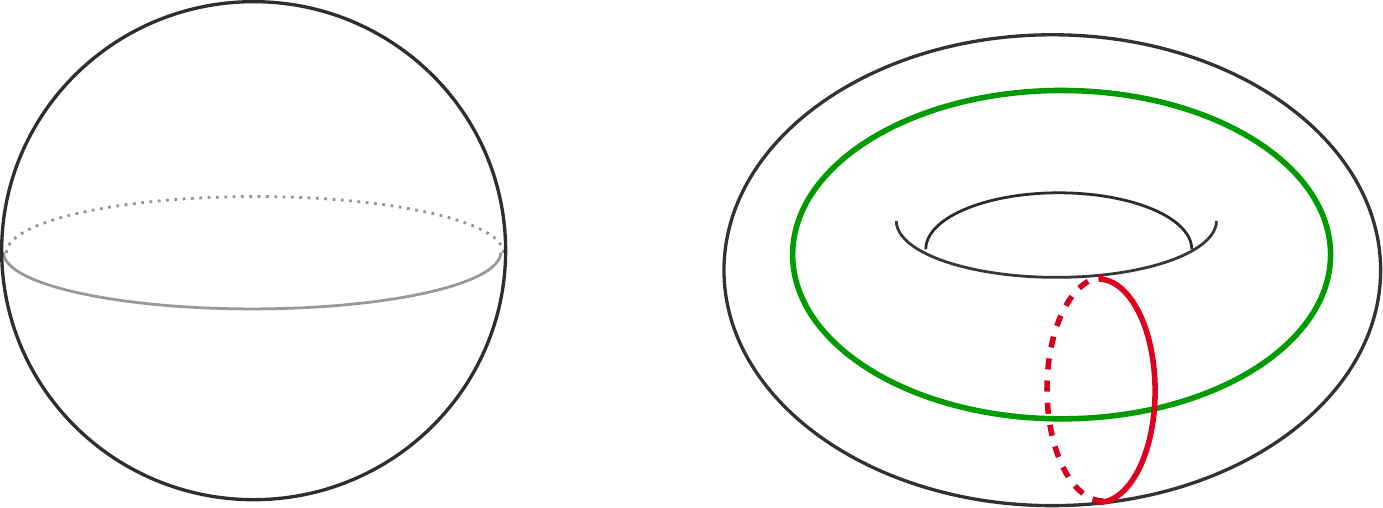}
\caption{Heegaard diagrams for $\Sp^3$. 
The picture shows two Heegaard diagrams (out of infinitely many with arbitraty genus $g$) for the sphere $\Sp^3$: for minimum genus $g=0$ (on the left), and for $g=1$ (on the right). 
The diagram with a Heegaard surface $g=0$ ($\Sp^2$) does not have any attaching curves.
The $\alpha$- and $\beta$-curves on the Heegaard surface $g=1$ ($\Sp^1 \times \Sp^1$)are shown in red and blue.
The toric Heegaard surface in the latter is the common boundary of two solid tori ($\D^2 \times \Sp^1$): we can view them such that inside this toric Heegaard surface, there is one solid torus, and there is another one outside. 
In particular, we can view the diagram as embedded in $\mathbb{R}^3$ plus a point at infinity (therefore in a space homeomorphic to $\Sp^3$). 
The outside solid torus is specified by the blue $\beta$-curve which is the boundary of a horizontally lying compression disc. 
Its spine would circle around the torus intersecting this compression disc transversally. 
Note that if one views the blue curve as the attaching sphere of a $2$-handle, the resulting manifold would be a topological ball $\D^3$, which can be easily capped-off to generate $\Sp^3$.}
\label{fig:heegaardsiagramS3}
\includegraphics[scale=0.4]{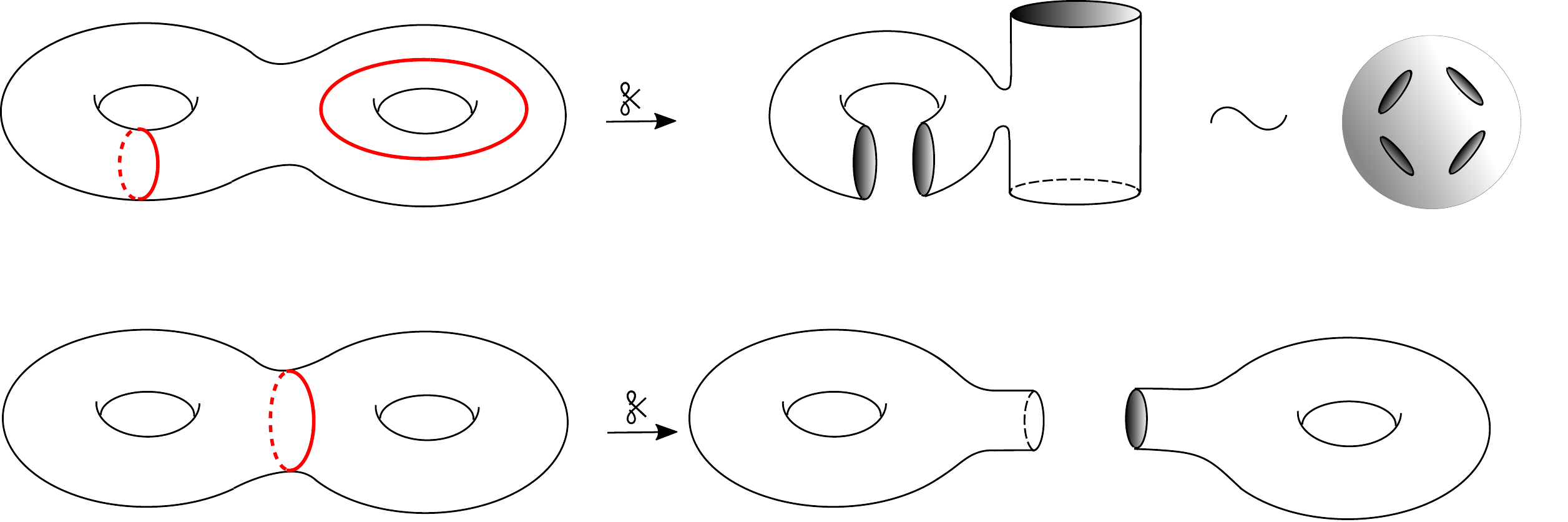}
\caption{Cutting along attaching curves. An example of viable attaching curves  (top) and a not viable one (bottom). Note that cutting along the latter separated the would-be Heegaard surface in two connected components.
}
\label{fig:cutting}
\end{minipage}
\end{figure}

We should point out the symmetry between $i$-handles and $(d-i)$-handles in $d$ dimensions. Since $\partial(\D^i\times\D^{d-i}) = (\Sp^{i-1}\times\D^{d-i})\cup(\D^{i}\times\Sp^{d-i-1})$, the difference between the two types of handles is which portion of the handle's boundary will glue to a onto a manifold and which part will remain  for other handles to be glued on. In particular, the $1$-handles and $3$-handles of $H_2$ in \eqref{eq:heegSplit}, glue onto $H_1$ as $2$-handles and $3$-handles respectively.

Finally, we point out that a Heegaard splitting of a 3-manifold is not unique, nevertheless two splittings of the same manifold (and the respective Heegaard diagrams), are always connected by a finite sequence of moves, called \textit{Heegaard moves}, consisting in:
\begin{itemize}
\item handle slides,
\item insertion/removal of  topologically trivial couples of $1$-handle and $2$-handle (i.e. glued in such a way that together they form a $3$-ball $\D^3$).
\end{itemize}

\begin{definition}
\label{def:heeg-genus}
Given a 3-manifold $M$, the minimal genus over all the possible Heegaard surfaces is a topological invariant. We call this number \textbf{Heegaard genus}.
\end{definition}

%%%%%%%%%%%%%%%%%%%%%%%%%%%%%%%%%%%%%%%%%%%%%%%%%%%%%%%%%%%%%%%%%%%%
\subsection{Connected sum and Heegaard splittings}
\label{sec:connected-sum}
%%%%%%%%%%%%%%%%%%%%%%%%%%%%%%%%%%%%%%%%%%%%%%%%%%%%%%%%%%%%%%%%%%%%

The \textit{connected sum} $M\,\sharp\, N$ of two $d$-manifolds $M$ and $ N$ is constructed by removing a topological $d$-ball $D^d$ from their interior and gluing $M$ and $N$ by identifying their boundaries (homemorphic to $S^{d-1}$). If $M$ and $N$ are both oriented,  there is a unique connected sum constructed through an orientation reversing map between the boundaries after the removal of the $d$-balls and the resulting manifold is unique up to homeomorphisms.

We define the \textit{boundary-connected sum} of two $d$-manifolds with boundaries, $M$ and $N$, as the manifold $ M\,\natural\,N$ obtained by performing a connected sum of their boundaries $\partial M\,\sharp\,\partial N$.
Note that the boundary connected sum of handlebodies $H_1$ and $H_2$ is a handlebody itself. 
The spine of $H_1\,\natural\,H_2$ can be represented by joining the two spines through a line or a point\footnote{The line connecting the two spines does not represent any handle, rather, the identification of two discs on the boundaries of the two handlebodies and, therefore, can be contracted to a point. Nevertheless is useful for the moment to consider it as a specification of the way the boundary-connected sum is performed.}.

A question that naturally arises is: given two $3$-manifolds $M$ and $N$, is there a way to represent a Heegaard splitting of $M\,\sharp\, N$ in terms of Heegaard splittings $M=\{H_1, H_2, \Sigma_{M}\}$ and $N=\{K_1, K_2, \Sigma_{N}\}$?
To answer this question, we consider 
a $3$-ball $\D_M$ (resp. $\D_N$) intersecting $\Sigma_M$ ($\Sigma_N$) transversally 
in one $2$-ball.
Since the result is unique up to homeomorphism, we can choose the ball to be removed as better suits us. 
Since the intersection of the $3$-ball with each element of the splittings is a ball of the appropriate dimension, the connected sum of $M\,\sharp\,N$ 
performed removing $\D_M$ and $\D_N$
will naturally give rise to a Haagaard splitting of the form $\{H_1\,\natural\,K_1,H_2\,\natural\,K_2,\Sigma_M\,\sharp\,\Sigma_N\}$.

\begin{figure}[h]
    \begin{minipage}[t]{0.8\textwidth}
      \centering
\def\svgwidth{0.5\columnwidth}
\centering
\includegraphics[scale=.15]{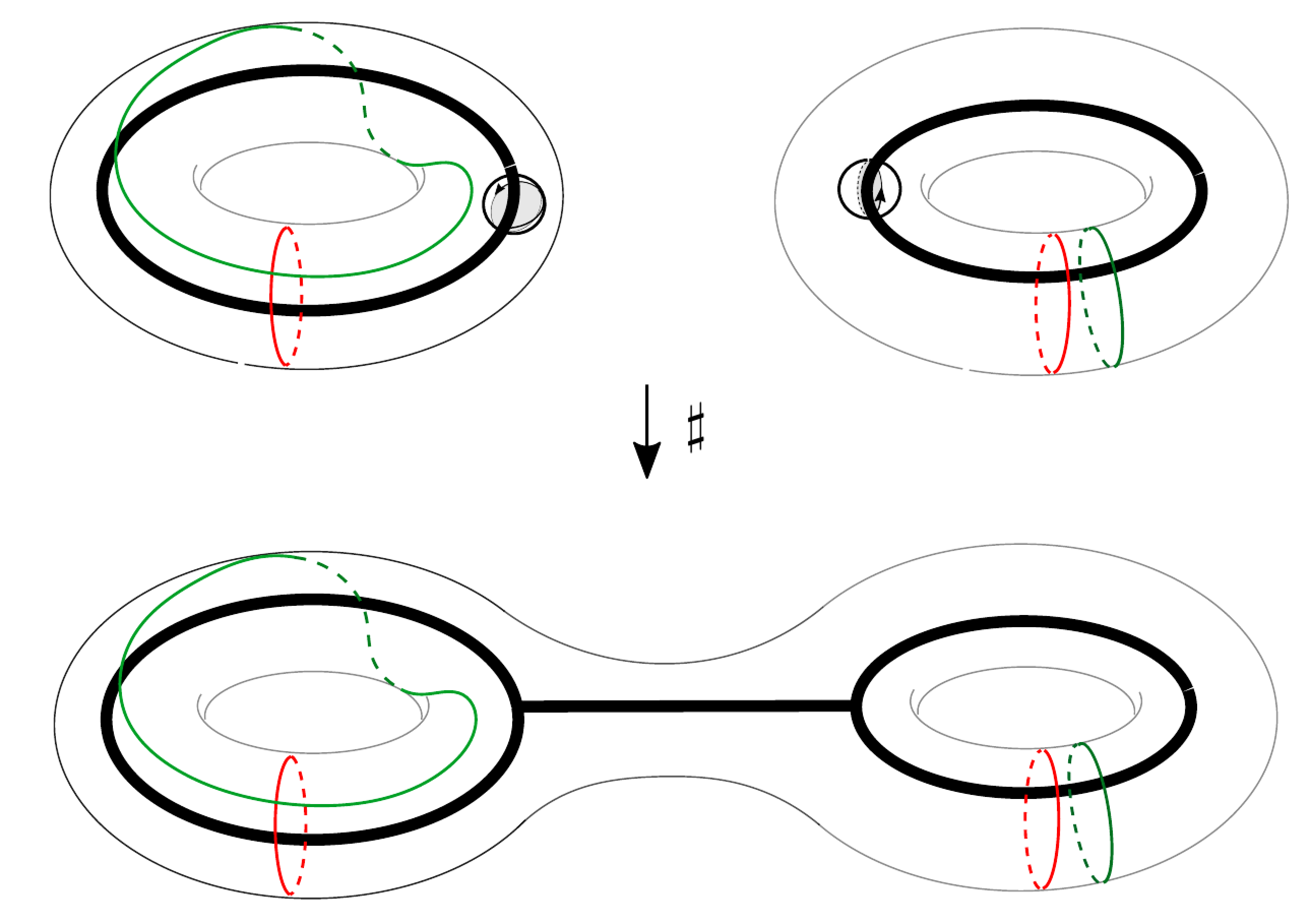}
\caption{Connected sum. We represent here the connected sum of $3$-manifolds via their Heegaard diagrams (Lens space $l(1, 1)$ on the left, $\Sp^1\times\Sp^2$ on the right). The picture shows the balls to be removed (castration) from the manifolds and how they intersect the Heegaard surfaces. Note that this correspond to the boundary-connected sum of the handlebodies. We show in the solid black thick line as the spine of the handlebodies, and the attaching curves are in red and in green. The arrows along the circles on the Heegaard surface show the reversed orientation.
}
\label{fig:connectsum}
\end{minipage}
\end{figure}

A few comments are in order. 
Firstly, we remark that the Heegaard splitting of closed manifolds is symmetric with respect to the two handlebodies. By this we mean that we can differentiate $H_1$ and $H_2$ through labels induced by the construction of the splitting, but ultimately their role (and therefore the role of $\alpha$-curves and $\beta$-curves) can be interchanged. For example, if we have in mind a handle decomposition of $M$ we can say that $H_1$ is given by the set of handles of index $i\leq 1$ while $H_2$ is given by the set of handles with $i\geq 2$ but, as we explained above, these characterizations can be easily switched for three-dimensional manifolds upon inverting the gluing order of the handles. If we induce the Heegaard splitting via a self-indexing Morse function $f$ via $f^{-1}(3/2)$, the role of the handlebodies can be switched upon sending $f$ to $-f+3$. In agreement with this feature of Heegaard splittings, we notice that $H_1$ and $H_2$ induce, as submanifolds of $M$, an opposite orientation of $\Sigma_M$. This might create an ambiguity in performing the connected sum $M\,\sharp\,N$ through the Heegaard splittings of $M$ and $N$ since reversing the orientation of one of the two Heegaard surfaces corresponds to a different boundary-connected sum of the handlebodies involved in the construction.
This ambiguity reflects the fact that the connected sum is unique only after specifying the orientation of the manifolds involved\footnote{An example of connected sum between three-dimensional manifolds in which reversing the orientation of one of the manifolds involved changes the result is $l(3,1)\,\sharp\,l(3,1)$, which is not homeomorphic to $l(3,1)\,\sharp\,\overline{l(3,1)}$, where $\overline{l(3,1)}$ represents $l(3,1)$ with the opposite orientation. A similar feature happens in four dimensions with the two possible connected sums of $\mathbb{CP}^2$ with itself.}.
Ultimately, a choice of $\alpha$- and $\beta$-curves for the two diagrams  corresponds to a choice of relative orientation for the two manifolds and specifies a connected sum constructed such that the set of $\alpha$-curves in $M\,\sharp\,N$ will be the union of the sets of $\alpha$-curves in $M$ and $\alpha$-curves in $N$ and similarly for the $\beta$-curves.

Secondly, we point out that the choice of a disc to be removed from each Heegaard surface during the connected sum operation is irrelevant, provided it does not intersect any attaching curve.
To convince oneself, it is sufficient to remember that cutting along all the $\alpha$-curves we obtain a pinched sphere on which any discs are equivalent, and similarly for the $\beta$-curves.

%%%%%%%%%%%%%%%%%%%%%%%%%%%%%%%%%%%%%%%%%%%%%%%%%%%
%\subsection{Smooth case and Morse theory formulation}
\subsection{Jackets as Heegaard surfaces} 
\label{sec:jackets-heeg}
%%%%%%%%%%%%%%%%%%%%%%%%%%%%%%%%%%%%%%%%%%%%%%%%%%%

Turning our attention to objects in the PL category, in particular to PL $3$-manifolds encoded in colored graphs, one might wonder whether there exists a natural formulation of Heegaard splittings in terms of combinatorial objects.
In \cite{Ryan:2011qm}, it is shown that the Riemann surfaces corresponding to the jackets of a rank-$3$ colored tensor model are Heegaard surfaces, and that if the corresponding triangulation is a manifold, then the triple $(K({\cal J}^{(ij, {\widehat{ij}})}), H^{(ij)}, H^{({\widehat{ij}})})$ is a Heegaard splitting of the triangulation. Although the complex structure of  the Riemann surfaces studied in \cite{Ryan:2011qm} was merely a consequence of the field content of the model examined, the Heegaard structure is purely combinatorial. In fact, this identification was already known in the crystallization theory literature, and led to the formulation of the concept of regular genus \cite{Gagliardi81}. Here, we revise such construction which will be of great importance in the following.

\begin{figure}[h]
     \begin{minipage}[t]{0.7\textwidth}
      \centering
\def\svgwidth{0.5\columnwidth}
\includegraphics[scale=.12]{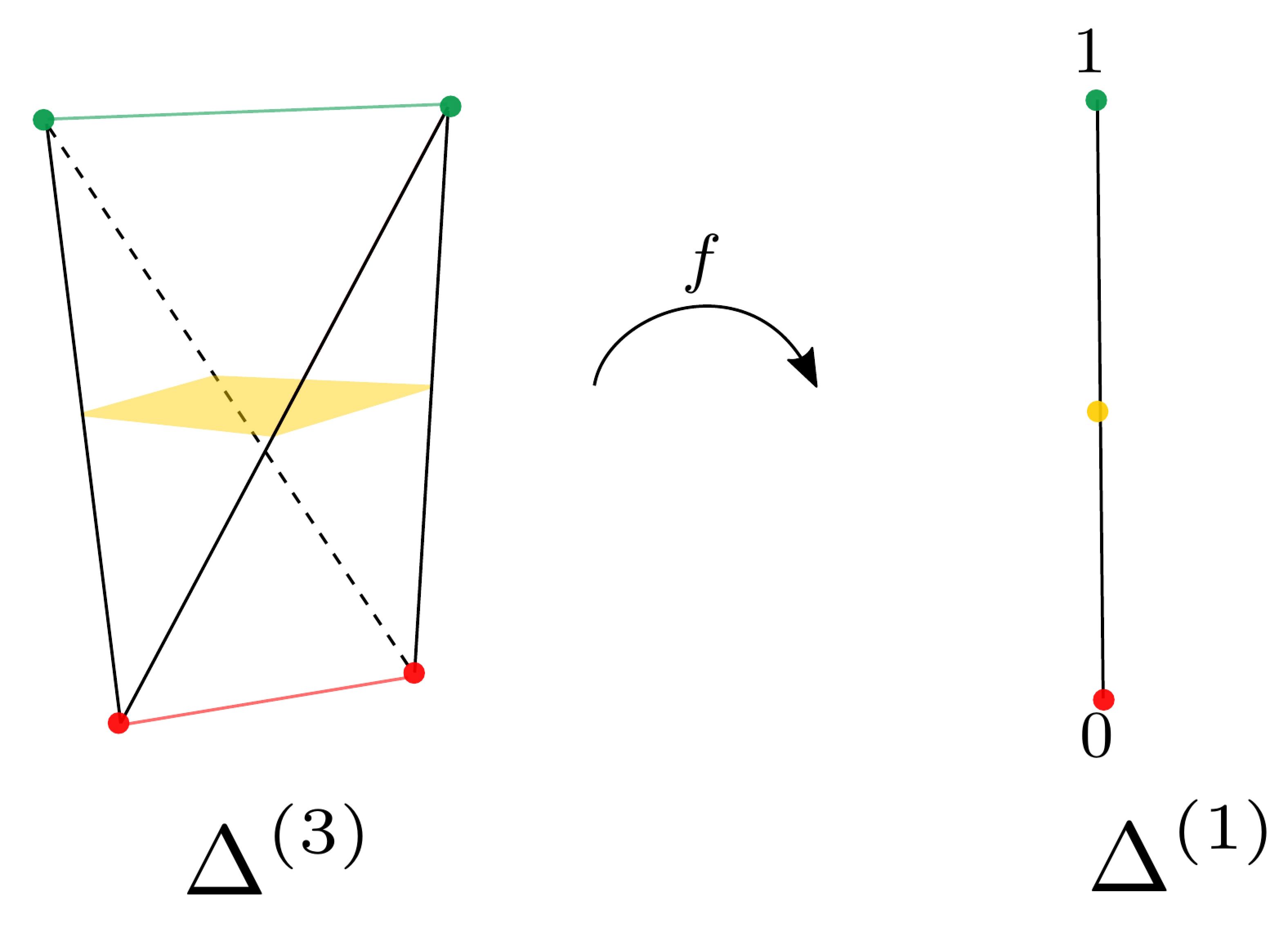}
\caption{Mapping of a 3-simplex $\simpT$.  $f$ is a map from a $d$-simplex 
$\Delta^{(d)}$ to a $d-2$-simplex 
$\Delta^{(d-2)}$; here, it is for $d=3$. As we will explain later, the $0$-skeleton of the first barycentric 
subdivision of 
$\Delta^{(d-2)}$ minus the $0$-skeleton of 
$\Delta^{(d-2)}$ itself defines the splitting and its preimage represents, in this case, a square in the Heegaard surface. The $0$-skeleton of 
$\Delta^{(d-2)}$ define the spine of the handlebodies. By linearly extending this identification we can reconstruct the entire tetrahedron.}
\label{fig:tetrahedronMorse}
\end{minipage}
\end{figure}

Let us consider a three-dimensional connected orientable closed manifold $M$ dual to a rank $3$ colored tensor model graph $\cal G$ which is introduced in section~\ref{sec:tensormodels}: $M=K({\cal G})$.
For every 3-simplex $\Delta^{(3)} \in  \T$, we consider a function $f$ mapping $\Delta^{(3)}$ onto $[0, 1]\in\mathbb{R}$ as in fig.~ \ref{fig:tetrahedronMorse}. 
We recall that in every $\Delta^{(3)}$, each edge is uniquely defined by a pair of colors $\{i, j\}$, where $i, j \in 0,1,2,3$. 
We can construct $f$ such that the preimage of the points $\{\{0\} , \{1\}\}$ under $f$ identifies everywhere in  $\T$
two non-intersecting edges of given colors $f^{-1}(0)=\{i, j\}$, $f^{-1}(1)=\{k, l\}$, $i, j, k, l\in\{0, 1, 2, 3\}$,
$i \ne j \ne k \ne  l$, while the preimage of any point in $(0, 1)$ gives us a square cross section of each $\Delta^{(3)}$.
We can glue these squares via their boundaries according to the colors, obtaining a surface embedded in $M$.
The surface $\Sigma$ constructed in this way is a realization of a quadrangulation represented by one the jackets ${\cal J}_{\{i, j\}\{k, l\}}$ of $\cal G$ and is dual to the corresponding matrix model obtained by removing the strands $\{i, j\}$ and $\{k, l\}$
\footnote{ Here we employ a slightly different notation for jackets with respect to the one introduced in section~\ref{sec:topcolgraph}. Notice that, if $d=3$, the set of bicolored cycles in the jacket is lacking only two elements from the set of bicolored cycles of $\cal G$. Thus, by writing ${\cal J}_{\{i, j\}\{k,l\}}$ we mean that $\{i, j\}\notin\{\{\eta_i, \eta_{i+1}\}\forall i\in \mathbb{Z}_4\}$ and similarly for $\{k,l\}$.
This notation is especially convenient in order to understand jackets in terms of Heegaard splittings.
}. 
Since the graph $\cal G$ is closed, bipartite and connected, so is $\cal J$. 
The surface $\Sigma$ therefore splits $M$ in two manifolds $H_0$ and $H_1$ with their common boundary being the surface $\Sigma$ itself. It is easy to notice that the spine of each $H_i$ is one-dimensional. In fact, it is given by the set of edges $f^{-1}(i)$ for $i\in\{0, 1\}$. Thus, $H_0$ and $H_1$ are handlebodies and a jacket $\cal J$ identifies a Heegaard surface $\Sigma$.

\begin{figure}[h]
    \begin{minipage}[t]{0.8\textwidth}
      \centering
\def\svgwidth{0.4\columnwidth}
\centering
\includegraphics[scale=.12]{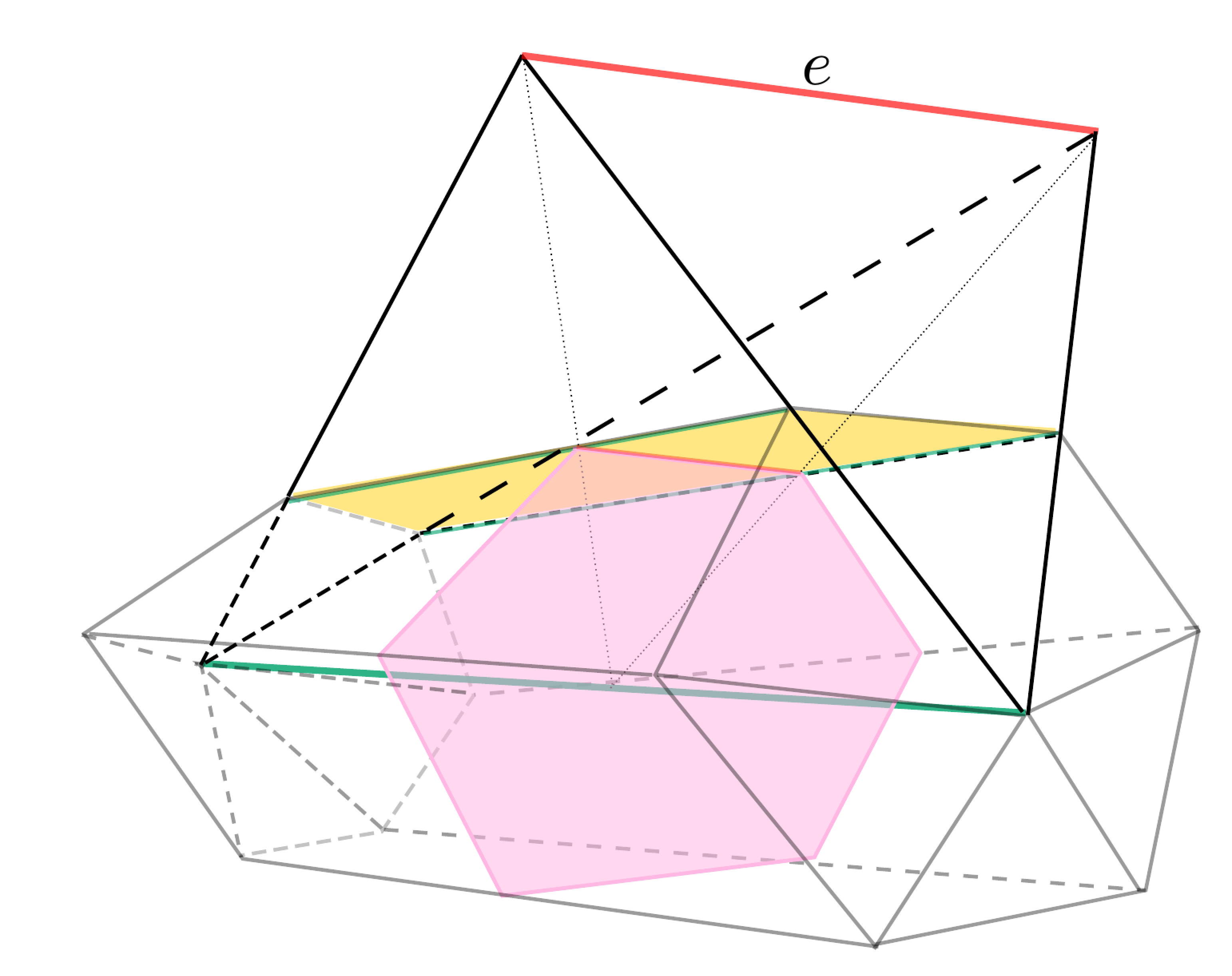}
\caption{A compression disc, an attaching curve and a spine of three-dimensional handlebody. The central green line is a single edge in the triangulation, shared by six $3$-simplices, which is to be identified as a spine of a three-dimensional handlebody. The rectangular faces (one of them colored in yellow) of the hexagonal prism are part of the Heegaard surface $\Sigma$. A compression disc is depicted in pink (a hexagon) and its intersection with the Heegaard surface is an attaching curve. As illustrated above, a segment of an attaching curve can be viewed as a projection on $\Sigma$ of the edge opposite to the spine ($e$) in each $3$-simplex.}
\label{fig:triangulated-compression}
\end{minipage}
\end{figure}

Once we identified a Heegaard splitting of $M$ in terms of combinatorial objects (i.e., via jackets) as described above, it is natural to wonder how the attaching curves arise. 
As we can see from fig.~\ref{fig:triangulated-compression}, for every edge $e$ in the spine of $H_i$ we can construct a compression disc in the shape of a polygon. The intersection of the compression disc with each of the tetrahedra sharing $e$ is a triangle (fig.~\ref{fig:triangulated-compression})
and the disc is therefore a polygon whose sides are as many as the number of the $3$-simplices sharing $e$.
Importantly, we see that the perimeter of the polygon is the projection of the edges opposite to the spine on the Heegaard surface.
This implies that, given the quadrangulation of $\Sigma$ defined by a jacket $\cal J$, we can draw the  attaching curves by connecting the opposite edges of each square.

A remark is in order. 
The construction of attaching curves drawn on a Heegaard surface we described so far is, in a way, overcomplete since it provides us with a redundant description; we will end up having many copies of the same curve (i.e. homotopically equivalent ones) and furthermore, curves that are homotopic to a point (which therefore should not be considered since they describe the attaching of a sphere). 
It is sufficient to consider only one representative of each equivalence class\footnote{We stress that an $\alpha$-curve and a $\beta$-curve can be homotopically equivalent and that the operation of modding out the equivalence class should be performed in either set independently.}, nevertheless, when constructing a trisection later on, a bit of care will be needed to convince ourselves
that such freedom does not imply any ambiguity in the construction.

For completeness, we compute here the genus of the Heegaard surface obtained with the procedure described above. Since $\Sigma = K({\cal J})$, we have that the genus $g_\Sigma$ is given by:
\begin{equation}
g_\Sigma = \frac{2-\chi_{\cal J}}{2} = \frac{1}{2}\left(2-V_{\cal J}+E_{\cal J}-F_{\cal J}\right)\,,
\end{equation}
where $\chi_{\cal J}$ is the Euler characteristic of $K({\cal J})$ and $V_{\cal J}$, $E_{\cal J}$ and $F_{\cal J}$ the vertices, edges and bicolored paths in $\cal J$ respectively. Since the vertices and the edges in the jacket are the same as those in $\cal G$, and they satisfy $2 E_{\cal G} = 4 V_{\cal G}$, we can further write:
\begin{equation}
\label{eq:jacket-genus-reduced}
g_\Sigma = 1 + \frac{1}{2}V_{\cal G} - \frac{1}{2}F_{\cal G} \geq 0\,.
\end{equation}

%%%%%%%%%%%%%%%%%%%%%%%%%%%%%%%%%%%%%%%%%%%%%%%%%%%%%%%%%%%%%%%%%%%%
\subsection{More Heegaard splittings in  triangulable manifolds} 
\label{sec:more-heeg-split}
%%%%%%%%%%%%%%%%%%%%%%%%%%%%%%%%%%%%%%%%%%%%%%%%%%%%%%%%%%%%%%%%%%%%

\begin{figure}[h]
    \begin{minipage}[t]{0.7\textwidth}
      \centering
\def\svgwidth{0.5\columnwidth}
\centering
\includegraphics[scale=0.5]{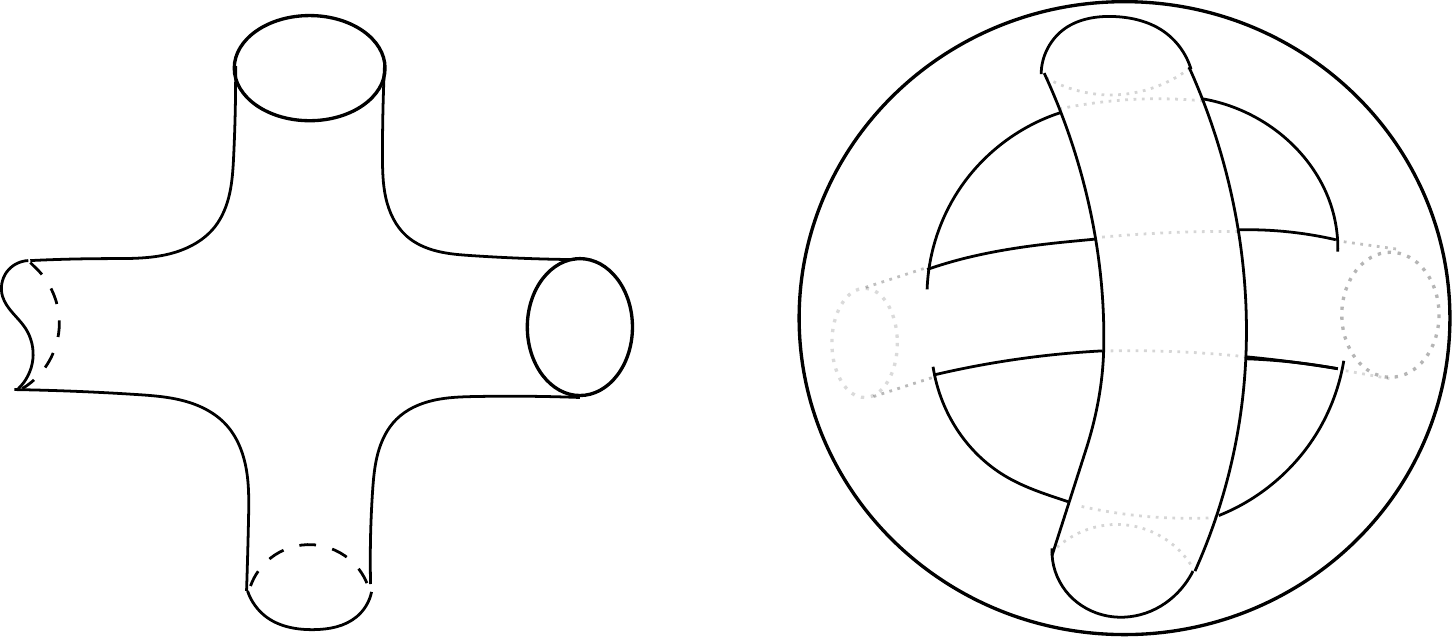}
\caption{A schematic representation of the two handlebodies obtained for a single tetrahedron using the $1$-skeletons of $\cal T$ and $\cal T^*$. 
}
\label{fig:1-skeleton-handlebodies}
\end{minipage}
\end{figure}

For later convenience, we illustrate now a different construction of Heegaard splittings from which we will borrow its technique later on.
Consider a triangulation $\cal T$  of a PL manifold $M$ and its dual cellular decomposition $\cal T^*$. 
The $1$-skeletons of $\cal T$ and $\cal T^*$ are perfect candidates to be identified as spines of $H_1$ and $H_2$. 
In fact, $H_1$ and $H_2$ are nothing but tubular neighborhoods of these two $1$-skeletons,
providing an orientation reversing homeomorphism between their boundaries, we can identify the Heegaard surface (see fig.~\ref{fig:1-skeleton-handlebodies}).
Note that if $\cal T$ is the triangulation associated with a colored graph $\cal G$, the 1-skeleton of $\cal T^*$ is the graph itself.
The Heegaard genus $g$ is then given by
\begin{equation}
\label{eq:trisec-genus-dual-graphs}
g = E_{\T} - V_{\T} +1 = E_{\T^*} - V_{\T^*} +1 
= V_{\T^*}+1
\,,
\end{equation}
where $E_{\T}$ ($E_{\T^*}$) and $V_{\T}$ ($V_{\T^*}$) are the number of edges and vertices in the $1$-skeletons of $\T$ ($\T^*$). 
The genus, then, corresponds to the number of independent loops of each graph, i.e., the dimension of the first homology groups of the $1$-skeletons. 
Note that, by definition, $V_{\T^*}$ corresponds to the number of tetrahedra in $\T$, which we denote by $F_{\Delta^{(3)}_{\T}}$, while $E_{\T^*}$ is the number of triangles in $\T$, which we denote by $F_{\Delta^{(2)}_{\T}}$.
Therefore eq.~\ref{eq:trisec-genus-dual-graphs} leads to the following identity for the Euler characteristic of $M$:
\begin{equation}
\chi(M) = V_{\T} - E_{\T} + F_{\Delta^{(2)}_{\T}} - F_{\Delta^{(3)}_{\T}}= 0\,,
\end{equation}
which is always true for odd-dimensional manifolds due to the Poincar\'e duality \cite{Nakahara}.

Finally, if we compare the present construction with the one obtained in sec.~\ref{sec:jackets-heeg} we can find from eq.~\eqref{eq:jacket-genus-reduced} (and using the fact that $\cal G$ is the 1-skeleton of $\T^*$):
\begin{equation}
\begin{split}
g_\Sigma - g &= -\frac{1}{2}\left(V_{\T^*}+F_{\cal J}\right) < 0\,.
\end{split}
\end{equation}

Therefore, we notice that $g_\Sigma < g$, which imply that this way of constructing a Heegaard splitting is actually less advantageous, as the topological invariant is the minimum genus of Heegaard surface.

%%%%%%%%%%%%%%%%%%%%%%%%%%%%%%%%%%%%%%%%%%%%%%%%%%%%%%%%%%%%%%%%%%%%
\section{Trisections}
\label{sec:trisections}
%%%%%%%%%%%%%%%%%%%%%%%%%%%%%%%%%%%%%%%%%%%%%%%%%%%%%%%%%%%%%%%%%%%%

A construction analogous to a Heegaard splitting (in three-dimensions) can be performed in four-dimensions, which is called \textit{trisection} \cite{GayKirby}. 
Note that one can perform trisections for non-orientable manifolds \cite{SpreerTillmann:2015}, however in this paper, we restrict ourselves to orientable manifolds.
Again, we start by working within the TOP category. We will restrict to objects in the PL category later in the paper.

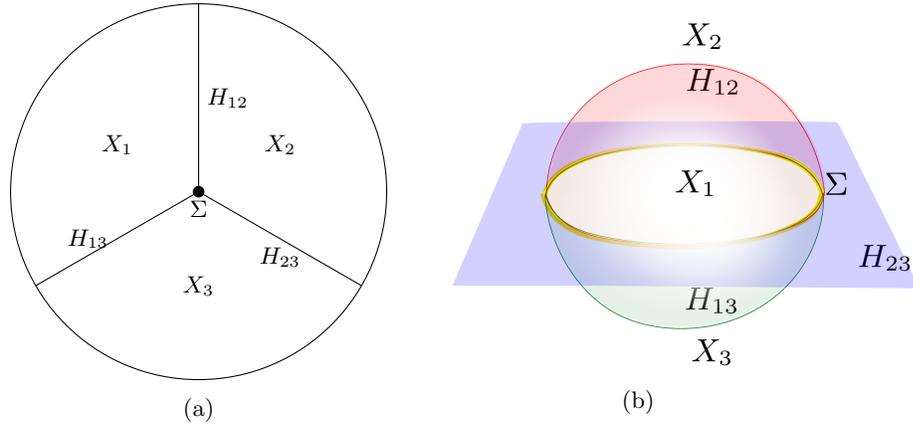
\begin{figure}[h]
    \begin{minipage}[t]{0.8\textwidth}
      \centering
\def\svgwidth{0.7\columnwidth}
\begin{subfigure}{0.4\textwidth}
\begin{tikzpicture}
\draw (0, 0) circle (2.5cm);
\filldraw (0, 0) circle (2pt);
\draw (0:0) -- (90:2.5);
\draw (0:0) -- (-30:2.5);
\draw (0:0) -- (210:2.5);
\node at (150:1.25) {${X}_1$};
\node at (30:1.25) {${X}_2$};
\node at (-90:1.25) {${X}_3$};
\node[right] at (90:1.25) {${H}_{12}$};
\node[below] at (-30:1.25) {${H}_{23}$};
\node[left] at (210:1.25) {${H}_{13}$};
\node[below] at (0,0) {$\Sigma$};
\end{tikzpicture}
\caption{}
\end{subfigure}
\begin{subfigure}{0.4\textwidth}
    \begin{minipage}[t]{1\textwidth}
      \centering
\def\svgwidth{0.7\columnwidth}
\centering
\includegraphics[scale=.15]{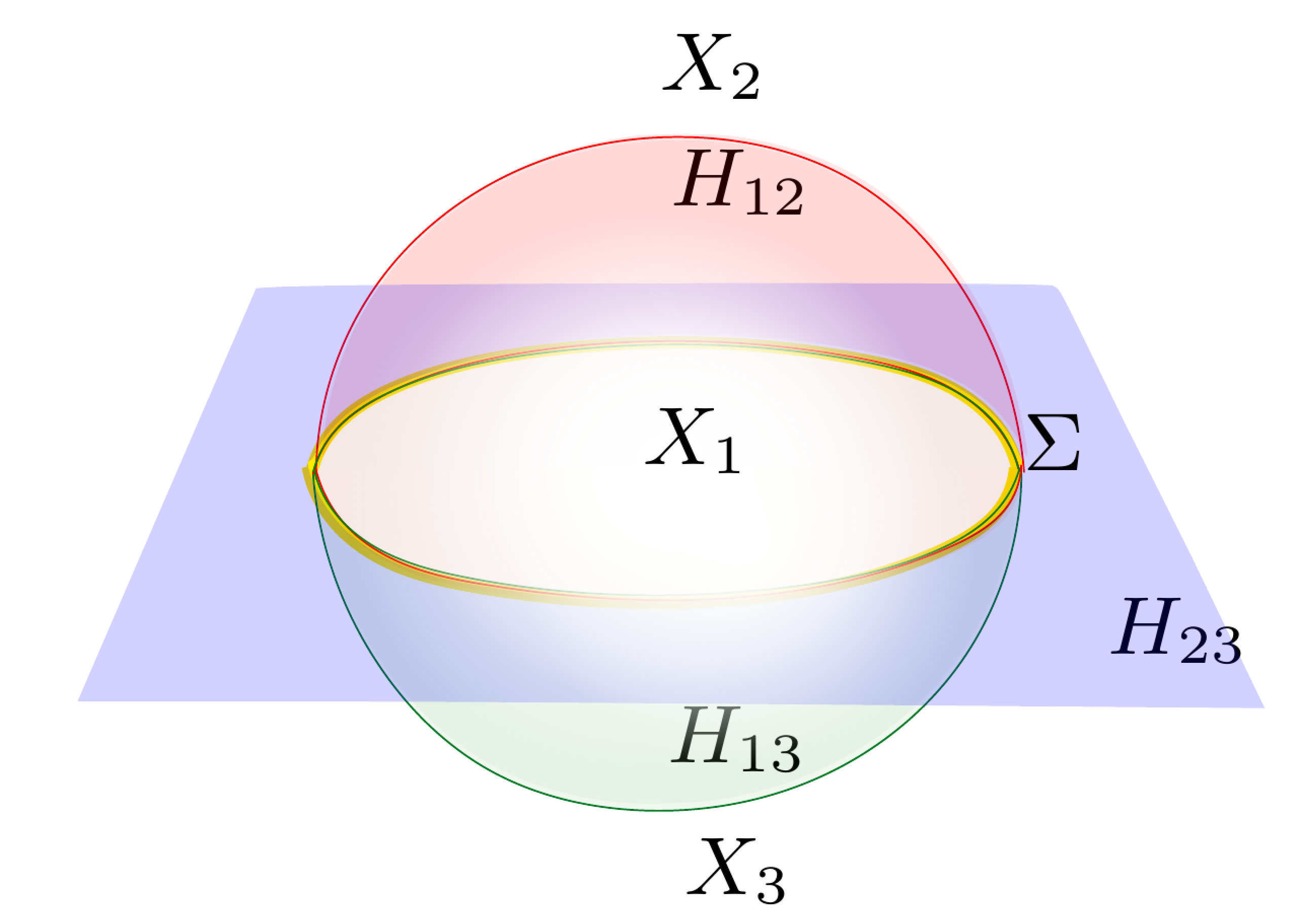}
\end{minipage}
\caption{}
\label{fig:trisec-3d-diagram}
\end{subfigure}
\caption{(a) Schematic representation of the trisection of a $4$-manifold $M$. 
$X_1$, $X_2$, and $X_3$ are four-dimensional submanifolds whose boundaries are $H_{12} \cup H_{13}$, $H_{12} \cup H_{23}$, and $H_{13} \cup H_{23}$ respectively. 
$H_{12}$,  $H_{13}$, and $H_{23}$ are three-dimensional handlebodies and $\Sigma$ is a Heegaard surfaces for the pairs $\{H_{12},  H_{13}\}$, $\{H_{12}, H_{23}\}$, and $\{H_{13}, H_{23}\}$.
$\Sigma$ is called the central surface of a trisection.
(b) Lower three-dimensional representation of the trisected manifold. $H_{12}$ is represented in half of $\Sp^2$ colored in red, $H_{13}$ in green, and $H_{23}$ in blue surface. Inside of $\Sp^2$, namely $\D^3$ bounded by $H_{12}$ and $H_{13}$ represents $X_1$, whereas the outside space above (below) $H_{12}$ ($H_{13}$) and $H_{23}$ represents $X_2$ ($X_3$). The yellow circle represents $\Sigma$. This representation of a trisection is very crude and strictly speaking wrong (e.g., keep in mind that $M$ ought to be closed); obviously, all the submanifolds and the given manifold itself are in principle general, however in this representation, they are depicted in a very special way.
}
\label{fig:trisection}
\end{minipage}
\end{figure}

\begin{definition}\label{def:trisection}
Let $M$ be a closed, orientable, connected 4-manifold. A trisection of $M$ is a collection of three submanifolds  $X_1, X_2, X_3 \subset M$ such that:
\begin{itemize}
\item each $X_i$ is a four-dimensional handlebody of genus $g_i$,
\item the handlebodies have pairwise disjoint interiors $\partial X_i\supset (X_i\cap X_j)\subset\partial X_j$ and $M= \cup_i X_i$,
\item the intersection of any two handlebodies  $X_i\cap X_j = H_{ij}$ is a three-dimensional handlebody,
\item the intersection of all the four-dimensional handlebodies $ X_1 \cap X_2 \cap X_3$ is a closed connected surface $\Sigma$ called \textbf{central surface},
\end{itemize}
for $i, j\in{1, 2, 3}$.
\end{definition}
Note that any two of the three-dimensional  handlebodies $\{H_{ij}, H_{jk}, \Sigma\}$ form a Heegaard splitting of $\partial X_j$.

In four dimensions, we have the following  \textit{extending theorem} \cite{montesinos}:

\begin{theorem}\label{th:extending-th}
Given a four-dimensional handlebody $H$ of genus $g$ and an homeomorphism $\phi:\partial H\rightarrow\partial H$, there exists a unique homeomorphism $\Phi:H\rightarrow H$ which extends $\phi$ to the interior of $H$.
\end{theorem}
It implies that closed $4$-manifolds are determined by their handles of index $i\leq 2$ and that there is a unique cap-off determining the remaining $3$- and $4$-handles (recall the symmetric roles of $i$-handles and $(4-i)$-handles in four dimensions). However, in the context of trisections, the extending theorem plays an even bigger role, for it can be applied to each handlebody $X_i$ in definition \ref{def:trisection}. Consequenstly, a trisection of $M$ is fully determined by the three three-dimensional handlebodies $H_{ij}$ which, in turn, can be represented by means of Heegaard diagrams.

Hence, similarly to the three-dimensional case of Heegaard splittings, 
a trisection can be represented with a diagram consisting of the central surface\footnote{From now on we may adopt the term ``central surface'' for both the case of trisections and Heegaard splittings 
when a feature is clearly common to the central surface of a trisection and the Heegaard surface of a Heegaard splitting.}
$\Sigma$ and three sets of curves: $\alpha$-curves, $\beta$-curves and $\gamma$-curves (collectively, attaching curves).
These curves are constructed, as before, by means of compression discs and represent the belt spheres of the $1$-handle of each of the three-dimensional handlebodies $H_{ij}$. 
A trisection diagram therefore combines the three Heegaard diagrams for $\partial X_i$ into a single diagram. Therefore, one can say that the construction of trisection, together with the extending theorem, allows us to study four-dimensional topology, within a two-dimensional framework. 
Again, infinitely many possible trisection diagrams are viable for a given manifold and they are connected by a finite sequence of moves generalizing Heegaard moves. 
We therefore have the following:

\begin{definition}
\label{def:trisection-genus}
Given a 4-manifold $M$, the minimal genus over 
all the possible central surfaces trisecting $M$ is a topological invariant. We call this number \textbf{trisection genus}.
\end{definition}

We remark that the connected sum of two $4$-manifolds $M=\{H_{12}, H_{23}, H_{13}, \Sigma_M\}$ (defining implicitly the handlebodies $X_1$, $X_2$ and $X_3$) and $N=\{K_{12}, K_{23}, K_{13}, \Sigma_N\}$ (defining $Y_1$, $Y_2$ and $Y_3$) can be constructed in analogy to the three-dimensional case by removing $4$-balls which intersect all the elements of each trisection in balls of the appropriate dimension. The resulting manifold will 
support
a trisection of the form $M\,\sharp\,N = \{H_{12}\,\natural\,K_{12}, H_{23}\,\natural\,K_{23}, H_{13}\,\natural\,K_{13}, \Sigma_M\,\sharp\,\Sigma_N\}$ implicitly defining the handlebodies $X_1\,\natural\,Y_1$, $X_2\,\natural\,Y_2$ and $X_3\,\natural\,Y_3$.

%%%%%%%%%%%%%%%%%%%%%%%%%%%%%%%%%%%%%%%%%%%%%%%%%%%%%%%%%%%%%%%%%%%%
\subsection{Stabilization}
\label{sec:stab}
%%%%%%%%%%%%%%%%%%%%%%%%%%%%%%%%%%%%%%%%%%%%%%%%%%%%%%%%%%%%%%%%%%%%

\begin{figure}[h]
    \begin{minipage}[t]{0.9\textwidth}
      \centering
\def\svgwidth{0.9\columnwidth}
\centering
\begin{subfigure}{0.4\textwidth}
\includegraphics[scale=.4]{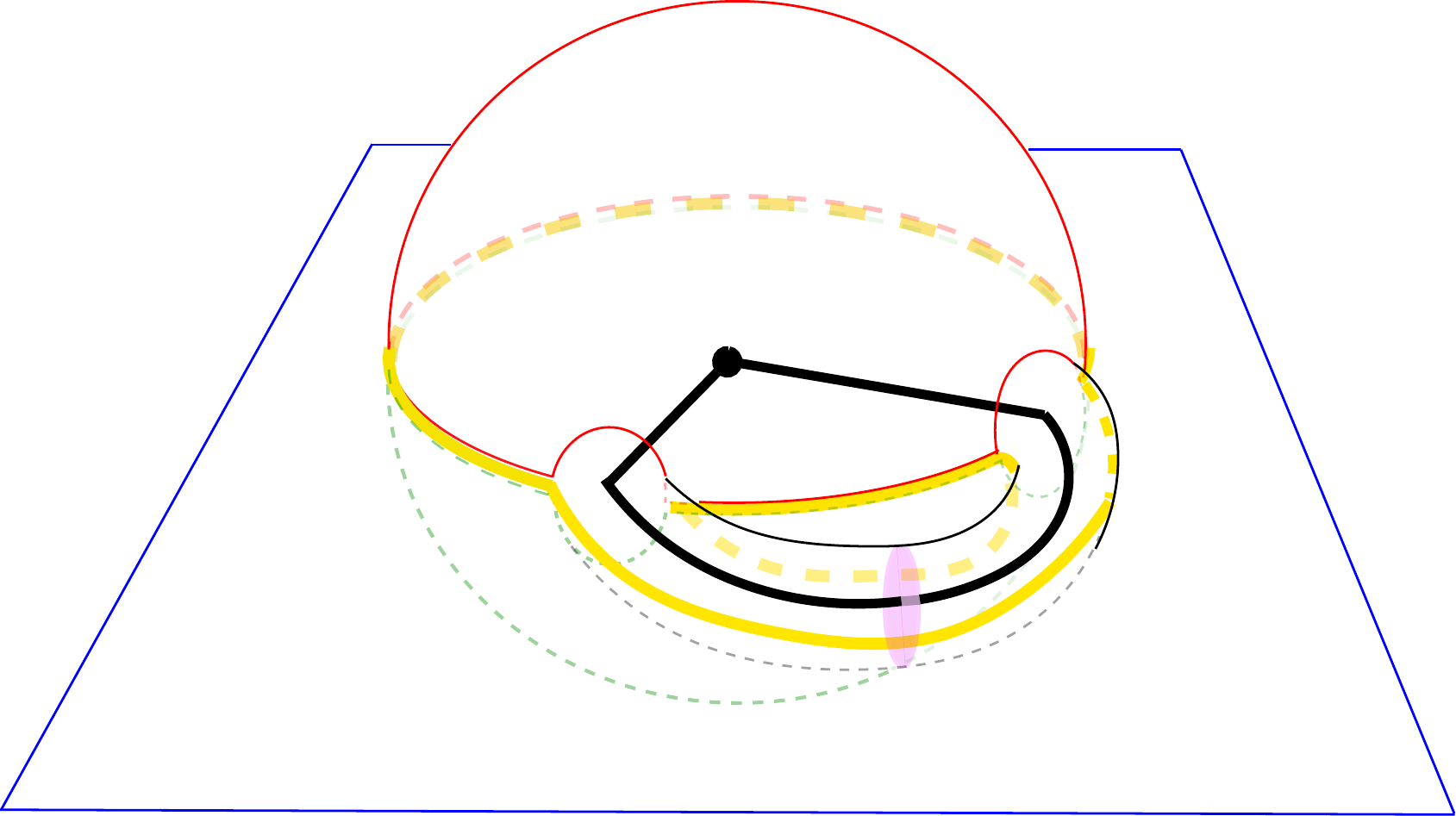}
\caption{}
\label{fig:stabilization-3d-scheme}
\end{subfigure}
\begin{subfigure}{0.55\textwidth}
\begin{minipage}{.2\textwidth}
\flushleft
\includegraphics[scale=.25, left]{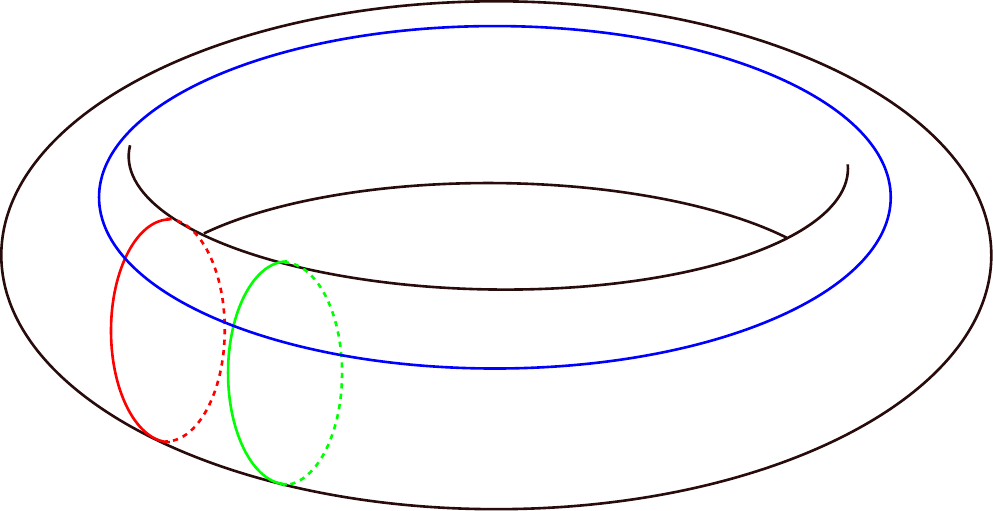}
\end{minipage}\hspace{2cm}
\begin{minipage}{.2\textwidth}
\flushright
\includegraphics[scale=.25, right]{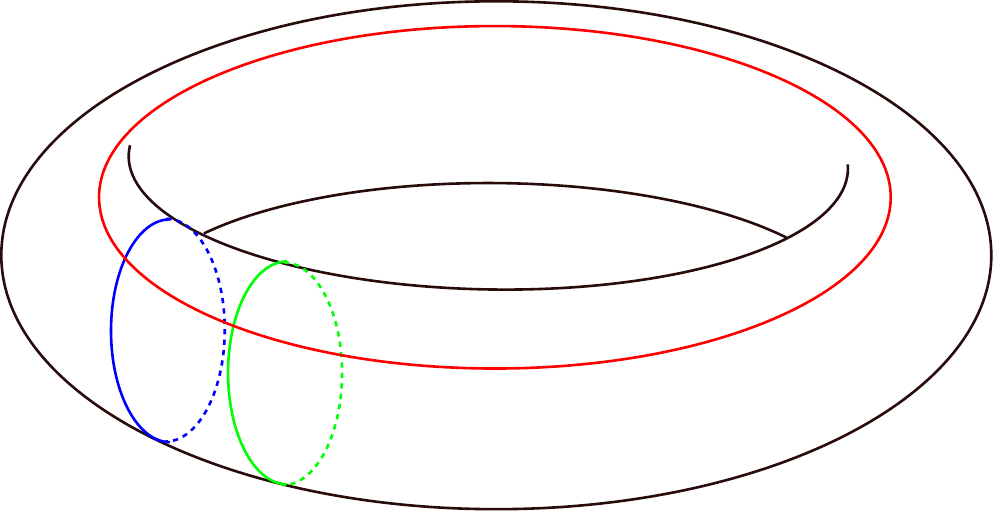}
\end{minipage}
\par\medskip
\centering
{\includegraphics[scale=.25]{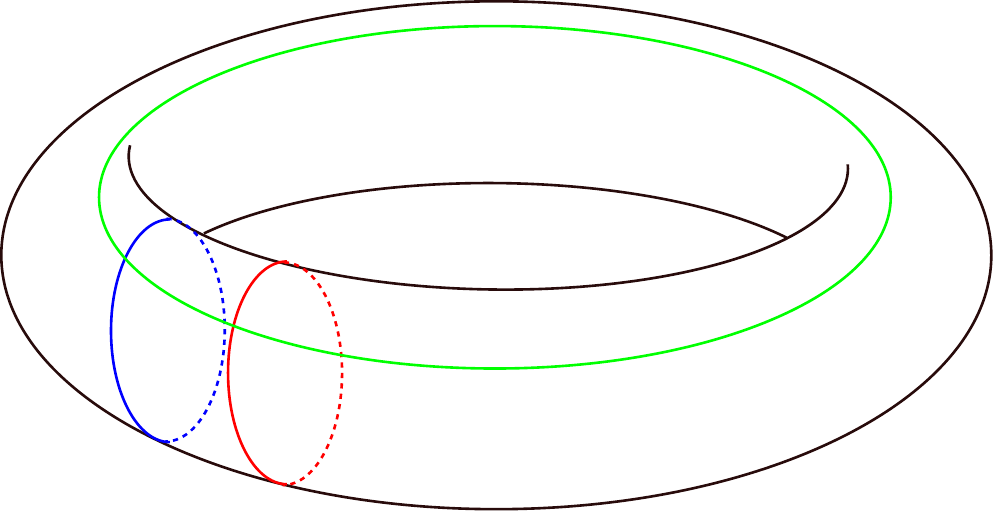}}
\caption{}
\label{fig:genus-one-trisection-sphere}
\end{subfigure}
\caption{Lower dimensional (in three dimensions) representation of stabilization in four dimensions. Fig.~\ref{fig:stabilization-3d-scheme} represents the process of adding a four-dimensional $1$-handle to one of the handlebodies preserving the trisection structure. The figure follows the conventions of fig.~\ref{fig:trisec-3d-diagram}. We also show the resulting spine of the four dimensional handlebody as well as the compression disc for the new handle. Fig.~\ref{fig:genus-one-trisection-sphere} shows the three genus-$1$ trisection diagrams for $\Sp^4$. Stabilization can be represented at the level of trisection diagrams as the connected sum with one of these three diagrams.} 
\label{fig:stabilization}
\end{minipage}
\end{figure}

Both in the context of Heegaard splittings and of trisections there is a move that increases the genus of the central surface by one. It is instructive to illustrate how this can be achieved and to point out small differences between the four-dimensional and three-dimensional cases.

We consider a three dimensional manifold $M^{(3)}$ with a Heegaard diagram of genus $g$, $\Sigma_M$, and  the genus $1$ Heegaard diagram of $\Sp^3$, which we call $T_\Sp$ (see fig.~\ref{fig:heegaardsiagramS3}).
Since $\Sp^3$ has trivial topology we have the following identity:
\begin{equation}
M^{(3)}\,\sharp\,\Sp^3 = M^{(3)}\,.
\end{equation}
As explained in sec.~\ref{sec:connected-sum}, this operation can be represented with the diagram $\Sigma_M\,\sharp\,T_\Sp$ which has genus $g'=g+1$. We can understand this operation in terms of carving a handle out of one of the two handlebodies in $M^{(3)}$ and adding it to the other one.
For a given $d$-manifold $N$ with boundary, the operation of drilling out a tubular neighborhood of a properly embedded ball $\D^{d-k-1}\subset N$ is equivalent to adding a $k$-handle whose attaching sphere bounds a ball $\D^k$ in $\partial N$. The properly embedded $(d-k-1)$-ball is bounded by the belt sphere of a $(k+1)$-handle which we may add in order to cancel the $k$-handle and to recover $N$. This describes how to increment the genus of the central surface of a Heegaard diagram if we consider the case\footnote{ Note that for $k=1$ we are identifying two discs on the boundary of a handlebody and represent their identification through the spine of the resulting $1$-handle. From this point of view, we can treat the operation of increasing the genus of a handlebody and the connected sum of two handlebodies (see fig.~\ref{fig:connectsum}) on the same footing, with the only difference being whether the considered discs lie on the boundary of the same handlebody or not. Note that in both cases it is sufficient to specify the spine of the new handle in order to recover the full topological information.} $d=3$ and $k=1$; note that a $2$-handle for one handlebody plays the role of a $1$-handle for the other handlebody.
In this way it is clear how we are actually not changing anything in the overall manifold but rather rearranging its handle decomposition.

In four dimensions there exist a similar operation which takes the name of \textit{stabilization}. The genus $1$ trisection diagrams of $\Sp^4$ are shown in fig.~\ref{fig:genus-one-trisection-sphere} and each represents a trisection where two handlebodies $X_i$ and $X_j$ are $4$-balls while the third $X_k$ has genus-$1$. Note that the boundary $\partial X_k$ has the topology of $\Sp^1\times\Sp^2$ as can be seen from each diagram by removing the curve circulating around the toroidal direction.

If we consider a $4$-manifold $M^{(4)}$ we can clearly increment the genus of its central surface by considering the connected sum of its trisection diagram with one of the three in fig.~\ref{fig:genus-one-trisection-sphere} . Although this is not within the investigation scope of the present work, we should mention that the stabilization operation allows us to always obtain a trisection where all the four-dimensional handlebodies have the same genus. This type of trisection is referred to as \textit{balanced}. In fact, it is worth noticing that the stabilization operation, although affecting the topology of all the three-dimensional handlebodies $H_{ij}$, only affects one of the four-dimensional one, while leaving the other two unmodified.

Stabilization too can be understood as a specific carving operation. As before, we identify a $\D^1$ that will constitute the spine of the carved $1$-handle. Since we are going to increase the genus of, say, $X_1$, the $1$-ball will need to be properly embedded in the complement $M^{(4)}\setminus X_1$ (we will carve the handle out of the complement and add it to $X_1$). The central surface simultaneously represents the boundary of all the three-dimensional handlebodies which, therefore, need to have their genus increased as well. Since we are only specifying one $1$-handle, and with simple symmetry considerations, it is easy to guess that the spine shall be a disc $\D^1$ embedded in $H_{23}$, with endpoints on the central surface. Fig.~\ref{fig:stabilization} shows a schematic representation of this procedure following the same conventions of fig.~\ref{fig:trisec-3d-diagram}.
Under such a move, the topology of $X_2$ and $X_3$ remains unaffected. To understand this, it  is sufficient to notice that $X_2$ in fig.~\ref{fig:trisec-3d-diagram} (respectively $X_3$) intersects only half of the $1$-handle and the intersection is a $4$-disc intersecting $\partial X_2$ (respectively $\partial X_3$) in $\D^3$ (see fig.~\ref{fig:stabilization}). In other words, carving the $1$-handle leads to two manifolds $X'_2$ and $X'_3$ satisfying:
\begin{equation}
\begin{split}
X'_2\,\natural\,\D^4 = X_2\,,\\
X'_3\,\natural\,\D^4 = X_3\,.
\end{split}
\end{equation}

Note that the portion of the boundary of the four-dimensional $1$-handle that does not constitute the attaching sphere has the topology of $\D^1\times\Sp^2$. Upon the following decomposition
\begin{equation}
\D^1\times\Sp^2 = (\D^1\times \D^2)\cup_{\D^1\times\Sp^1}(\D^1\times\D^2)\,,
\end{equation}
we can understand it as a pair of three-dimensional 1-handles with common boundaries and ``parallel'' spines. Therefore a regular neighborhood of a one-dimensional disc properly embedded in one of the three-dimensional handlebodies intersects all the elements of a trisection without spoiling the construction, but rather defining an alternative trisection for the same manifold.

%%%%%%%%%%%%%%%%%%%%%%%%%%%%%%%%%%%%%%%%%%%%%%%%%%%%%%%%%%%%%%%%%%%%
\subsection{Subdividing $4$-simplices}
\label{sec:cutting-simplices}
%%%%%%%%%%%%%%%%%%%%%%%%%%%%%%%%%%%%%%%%%%%%%%%%%%%%%%%%%%%%%%%%%%%%

We would like to understand trisections from {\it colored triangulations}, i.e., triangulations dual to colored graphs which can be generated by colored tensor models.
It amounts to formulating trisections relying on combinatorics.
We will do so, by generalising the three-dimensional Heegaard splittings formulated in the colored tensor models  \cite {Ryan:2011qm}.
From now on, we therefore restrict to the PL category.

\begin{figure}[h]
     \begin{minipage}[t]{0.85\textwidth}
\def\svgwidth{0.85\columnwidth}
\includegraphics[scale=.15]{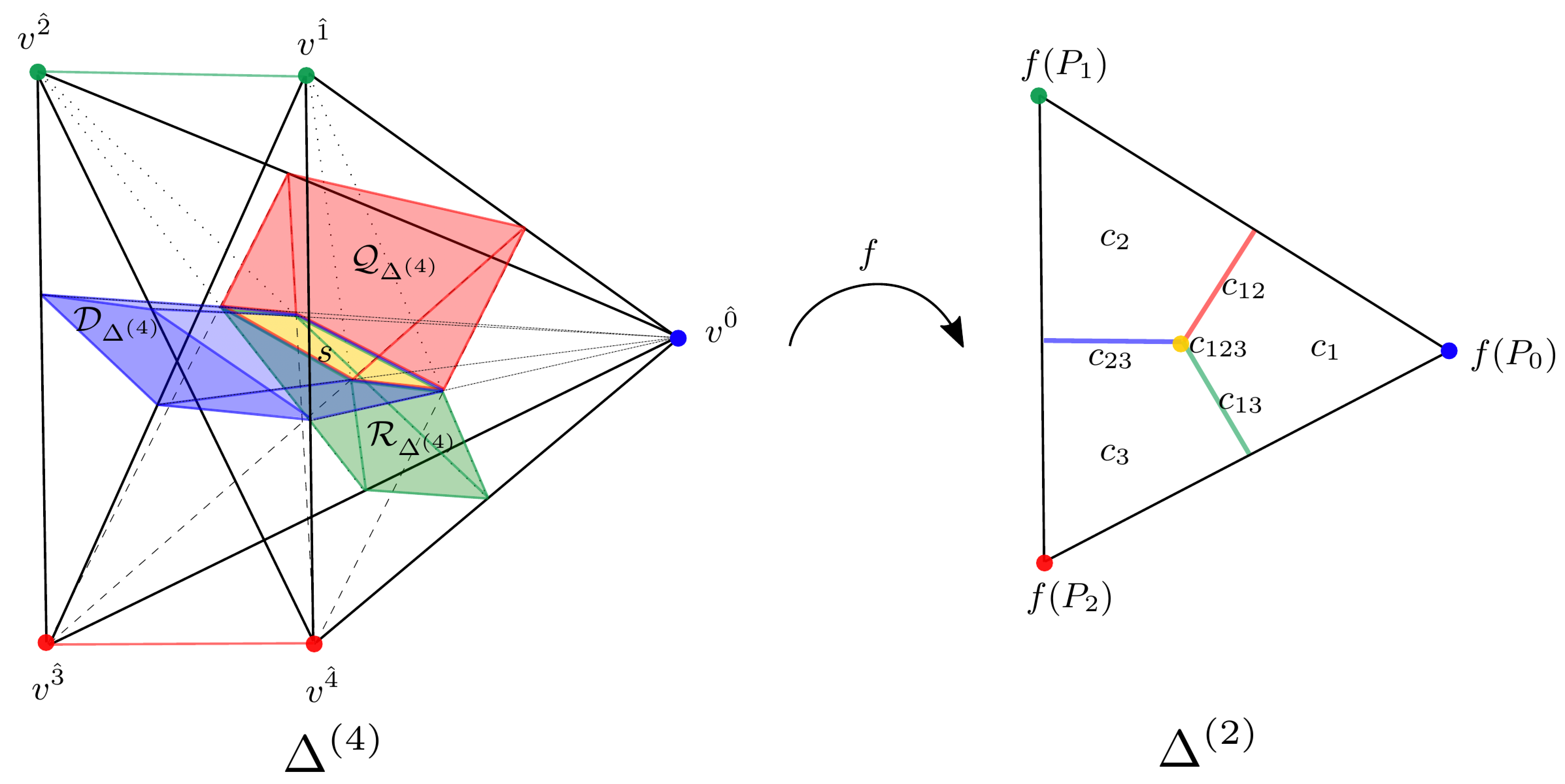}
\caption{
Illustration of the linear map from $\simpF$ to $\simpT$. The sets $P_i$ partitioning the vertices of the $4$-simplex $\simpF$, as well as their images in 2-simplex $\simpT$, are shown. Removing the $0$-skeleton of $\simpT$ from the $0$-skeleton of its first barycentric subdivision provides with the cubical decomposition. The preimage of $\simpT$ under $f$ splits the $4$-simplex in three four-dimensional pieces, whose boundaries are three-dimensional blocks ${\cal D}_{\simpF}$, ${\cal Q}_{\simpF}$,  and ${\cal R}_{\simpF}$. The latter three three-dimensional blocks meet at one two-dimensional square $s$.
${\cal D}_{\simpF}$ is  in blue,  ${\cal Q}_{\simpF}$ in red, ${\cal R}_{\simpF}$ in green, and  the common surface $s$ in yellow.
$f({\cal D}_{\simpF})=c_{23}$, $f({\cal Q}_{\simpF})=c_{12}$, $f({\cal R}_{\simpF})=c_{13}$, and $f(s)=c_{123}$.
}
\label{fig:4simplexmap}
\end{minipage}
\end{figure}

Following \cite{bell2017}, let us consider a $4$-simplex $\simpF$ and define a  partition of its vertices in three sets $P_0$, $P_1$ and $P_2$ such that one vertex belongs to one of the sets and the rest is divided in two pairs. For example, labeling each vertex with the color of its opposite 3-face we might have that the vertex $v^{\widehat{0}}$ is assigned to $P_0$, the vertices $v^{\widehat 1}$ and $v^{\widehat 2}$ are assigned to $P_1$ and the vertices $v^{\widehat 3}$ and $v^{\widehat 4}$ to $P_2$. 
Given such a partition, any pair of sets $(P_i, P_j)$ is identified with a $n$-dimensional subsimplex on the boundary of $\simpF$
while the third set $P_k$ is identified with the opposite $(3-n)$-dimensional subsimplex, where $i, j, k \in \{0,1,2\}$ and $i\ne j \ne k$. 
For example, $(P_1 = \{v^{\widehat 1}, v^{\widehat 2}\}, P_2  = \{v^{\widehat 3}, v^{\widehat 4}\})$ and $P_0  = \{v^{\widehat 0}\}$ give a 3-simplex spanned by $v^{\widehat 1}, \; v^{\widehat 2}, \; v^{\widehat 3},\; v^{\widehat 4}$  and a 0-simplex $v^{\widehat 0}$, whereas  $(P_0 = \{v^{\widehat 0}\}, P_1 = \{v^{\widehat 1}, v^{\widehat 2}\})$ and $P_2 = \{v^{\widehat 3}, v^{\widehat 4}\}$ give a 2-simplex spanned by $v^{\widehat 0}, \; v^{\widehat 1}, \; v^{\widehat 2}$ and a 1-simpex with endpoints $v^{\widehat 3}$ and $ v^{\widehat 4}$.

Then, we can define a map $f$ from $\simpF$ to  
$\simpT$ such that each set $P_i$ is sent into each of the three vertices in $\simpT$ and 
extend it linearly to the interiors of $\simpF$ and $\simpT$.
We proceed by considering the subcomplex spanned by the $0$-skeleton of the first barycentric subdivision of $\simpT$ 
 minus the $0$-skeleton of $\simpT$. 
The resulting cubical decomposition of $\simpT$ is shown in fig.\ref{fig:4simplexmap}. 
$\simpT$ is decomposed in three $2$-cubes 
$c_i$
with $i \in 1,2,3$, pairwise intersections of which result in $1$-cubes 
$c_{ij}$,
all sharing a central $0$-cube, $c_{123}$.
The preimage of this construction under $f$ gives us the splitting of $\simpF$ we are looking for. 
Notice that the boundary faces of $\simpT$ (each spanned by two vertices) are subdivided into two $1$-cubes. The preimage of $f$ therefore induces splittings of the subsimplices on the boundary of $\simpF$ identified with the pairs $(P_i, P_j)$.
Focusing on $\Delta^{(3)}\ni \{ v^{\widehat 1}, v^{\widehat 2}, v^{\widehat 3}, v^{\widehat 4}\}$,
which is sitting opposite to $v^{\hat 0}$, 
and considering  the partition $P_1 = \{v^{\widehat 1}, v^{\widehat 2}\}$, $P_2  = \{v^{\widehat 3}, v^{\widehat 4}\}$ and $P_0  = \{v^{\widehat 0}\}$, 
$\Delta^{(3)}$ is mapped via $f$ to a 1-simplex of $\simpT$ in precisely the same manner as in 
fig.~\ref{fig:tetrahedronMorse}.
The coning  of the splitting  surface of $\Delta^{(3)}$ 
with respect to $v^{\widehat 0}$, 
generates a square prism which we call ${\cal D}_{\simpF}$, whose image under $f$ is $1$-cube $c_{23}$. 
Similarly, in the two $2$-subsimplices of $\simpF$ defined by 
$\{P_0, P_1\}$ and $\{P_0, P_2\}$,
 we identify a one-dimensional cross section, which then will be coned  toward $P_2$ and $P_1$ respectively.
These conings will generate triangular prisms ${\cal Q}_{\simpF}$ and ${\cal R}_{\simpF}$, whose images are $c_{12}$ and $c_{13}$ in $\simpT$.
The intersection ${\cal Q}_{\simpF} \cap {\cal R}_{\simpF} \cap D_{\simpF}$ 
 is a two-dimensional cube\footnote{The bidimensionality of the central square is ensured by the fact that all the pairwise intersections of the three-dimensional blocks are transverse.}. 
Fig.~\ref{fig:4simplexmap} 
shows such coning operations.

%%%%%%%%%%%%%%%%%%%%%%%%%%%%%%%%%%%%%%%%%%%%%%%%%%%%%%%%%%%%%%%%%%%%
\subsection{Splitting $4$-bubbles}
\label{sec:split4bubbles}
%%%%%%%%%%%%%%%%%%%%%%%%%%%%%%%%%%%%%%%%%%%%%%%%%%%%%%%%%%%%%%%%%%%%

At this point, one would like to induce the above subdivision in every simplex $\sigma$ of a triangulation $\T$ of a given manifold $M$ and prove the emerging structure of a trisection, namely see that each of the sets ${\cal D}=\bigcup_\sigma {\cal D}_\sigma$, ${\cal Q}=\bigcup_\sigma {\cal Q}_\sigma$, ${\cal R}=\bigcup_\sigma {\cal R}_\sigma$ is connected and homeomorphic to a handlebody. In order to achieve this{\footnote{
Indeed we will have to define a new structure related to $\cal D$, which improves its topological properties in order to obtain a handlebody.
}} we will have to perform a few manipulations. For later reference, we will call the  attaching curves determined by manipulations of $\cal Q$ (respectively $\cal R$, $\cal D$) as $\alpha$ (respectively $\beta$, $\gamma$).

Let us therefore consider a colored triangulation  $\T$ of a $4$-manifold $M$, and a colored graph $\cal G$ dual to $\T$, i.e., $K(\cal G) = \T$ and $\vert K(\cal G)\vert  =M$.
If we take seriously a partition of the vertices induced by colors, we notice soon that 
the main immediate obstacle is achieving the connectedness of $ {\cal Q}$ and ${\cal R}$.
Evidently, the union ${\cal Q} \cup {\cal R} $ consists of disconnected  three-dimensional polytopes surrounding the vertices of $\T$ which belong to the isolated partition set whose element is only one vertex per $4$-simplex.

Let us elaborate on the structure of ${\cal Q}$ and ${\cal R}$.
In the triangulation $\T$, a $4$-bubble $\B_a^{\widehat{i}}$ identifies a three-dimensional subcomplex which surrounds a vertex $v_a^{\widehat{i}}$. 
In particular,  $v_a^{\widehat{i}}$ sits opposite to a $3$-face of color $i$ in every $4$-simplex containing it and the triangulation dual to $\B_a^{\widehat{i}}$, $K(\B_a^{\widehat{i}})$, is PL-homeomorphic to the union of such 3-faces. 
Moreover, we point out that such a triangulation, $K(\B_a^{\widehat{i}})$, is also homeomorphic to the link of $v_a^{\widehat{i}}$ which, for the case of $M$ 
being a manifold, turns out to be a topological $3$-sphere\footnote{Nevertheless, colored tensor models and colored graphs generate in general pseudo-manifolds and, therefore, the topology of $K(\B_a^{\widehat i})$ might turn out to be very different. 
We comment on this case in section \ref{sec:pseudo-mfd}.}. 
Given the combination of colors defining the $4$-bubble, though, a possibly more accurate way to address the corresponding triangulation is not as the union of the $3$-faces situating opposite to $v_a^{\widehat{i}}$, but rather as the union of a set of three-dimensional cross sections parallel to such $3$-faces which cut $4$-simplices midway between $v_a^{\widehat{i}}$ and its opposite $3$-faces, namely, ${\cal Q}_{\sigma_a} \cup {\cal R}_{\sigma_a}$ in fig.~\ref{fig:doublepentachoron}.
See fig.~\ref{fig:bubble-scheme} for a lower dimensional representation of $K(\B^{\widehat{0}})$.

\begin{figure}[h]
    \begin{minipage}[t]{0.8\textwidth}
      \centering
\def\svgwidth{0.6\columnwidth}
\centering
\includegraphics[scale=.12]{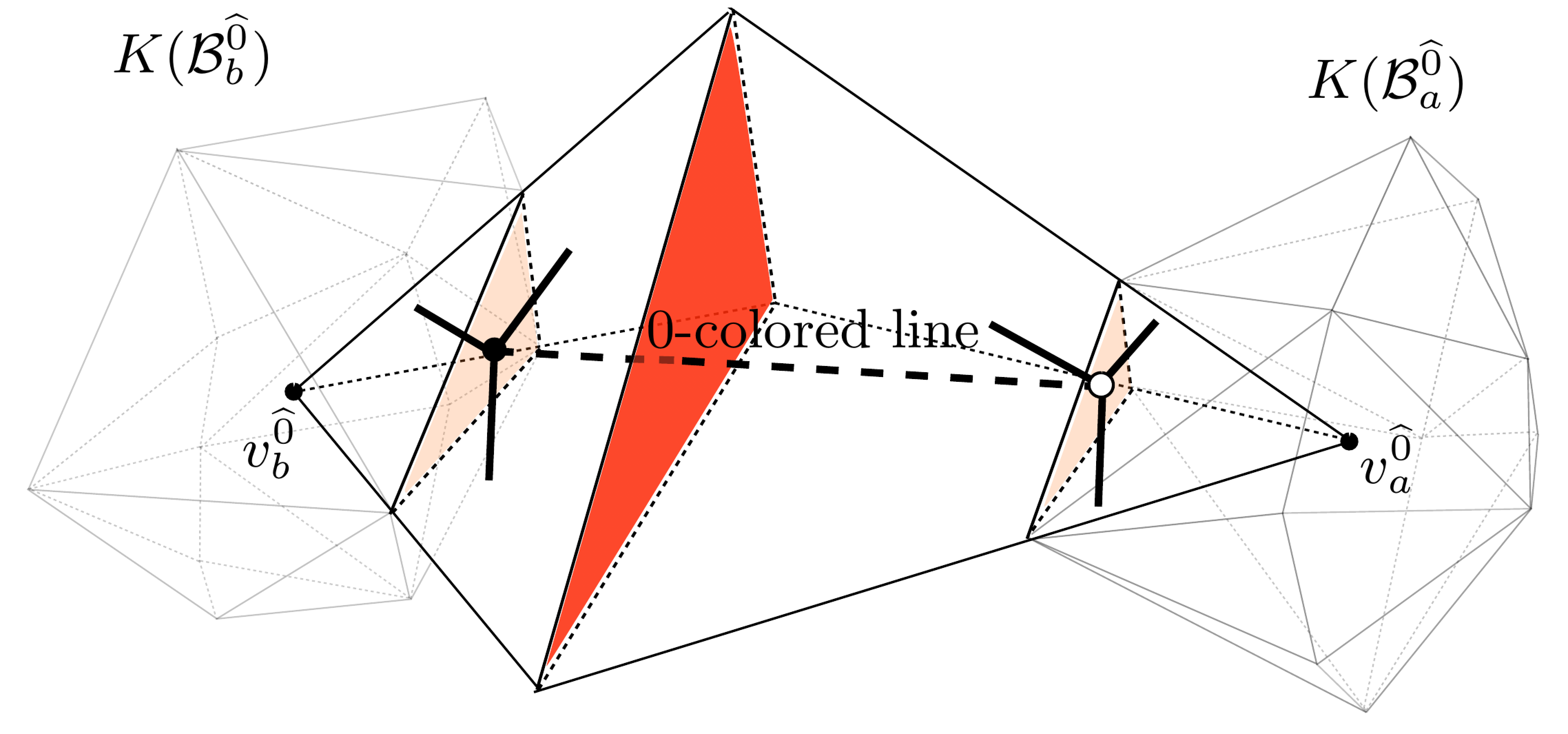}
\caption{Representation of two components of $K(\B^{\widehat{0}})$ in a three-dimensional complex. In three dimensions, $K(\B^{\widehat{0}})$ is a two-dimensional complex whose edges are shown in grey. In the picture we present two components:$K(\B^{\widehat{0}}_a)$ and $K(\B^{\widehat{0}}_b)$, surrouding $v_a^{\widehat{0}}$ and $v_b^{\widehat{0}}$ respectively. 
The same $0$-colored face (shown in red and shared by two 3-simplices) gives rise to two different building blocks (shown in orange) in $K(\B^{\widehat{0}})$.}
\label{fig:bubble-scheme}
\end{minipage}
\end{figure}

Consequently, given the set $\Delta_a =\{\sigma\; 4\text{-simplex s.t. } v_a^{\widehat{i}}\in \sigma\}$, the $4$-bubble identifies the union
\begin{equation}
K(\B_a^{\widehat{i}}) = \bigcup_{\sigma\in\Delta_a}{\cal Q}_\sigma\cup{\cal R}_\sigma\,.
\end{equation}
For later reference, we call the four-dimensional neighborhood{\footnote{Note that we choose to call this $X_1^a$ as it will be part of one of the trisection four-dimensional handlebodies defined earlier in definition~\ref{def:trisection}.}} of $v_a^{\widehat{0}}$ bounded by 
$K(\B_a^{\widehat{0}})$,
$X^a_{1}$
and we define the following unions: ${\cal Q}_a = \bigcup_{\sigma\in\Delta_a}{\cal Q}_\sigma$, ${\cal R}_a = \bigcup_{\sigma\in\Delta_a}{\cal R}_\sigma$. 

We pick $0$ as a special color and define a specific partition\footnote{Here we picked the color $0$ to identify the vertex that in every $4$-simplex is ``isolated'' by the partition, nevertheless, we stress that at this level any permutation of the colors would be an equivalent choice.}, i.e.,
$P_0 =\{v^{\widehat{0}}\}$, $P_1 =\{v^{\widehat{1}}\}\cup\{v^{\widehat{2}}\}$ and $P_2 =\{v^{\widehat{3}}\}\cup\{v^{\widehat{4}}\}$. 
Then we consider the $4$-bubbles $\B_a^{\widehat{0}}$ and, in each such 4-bubble, the jacket ${\cal J}_{P_1, P_2}={\cal J}_{\{1, 2\}, \{3, 4\}}$. Combining the constructions described in sec.~\ref{sec:jackets-heeg} and sec.~\ref{sec:cutting-simplices}, we readily obtain the sets ${\cal Q}$ and ${\cal R}$. Nevertheless, each of these sets, is disconnected and constituted by as many connected components as many vertices $v^{\widehat{0}}$ are in the triangulation $\T$. 
Recalling how jackets identify Heegaard surfaces for the realizations of $4$-colored graphs, it is easy to see that ${\cal Q}_a$ and ${\cal R}_a$ are the two handlebodies in a Heegaard splitting of a given $\B_a^{\widehat{0}}$. Looking at the Heegaard splittings $\B_a^{\widehat{0}} = (H_{1, a}, H_{2, a}, K({\cal J}(\B_a^{\widehat{0}})))$, we have that:
\begin{equation}\label{eq:interm-bubble-decomposition}
\begin{split}
{\cal Q}_a = H_{1, a}\,,\\
{\cal R}_a = H_{2, a}\,,\\
{\cal Q} = \bigsqcup_a H_{1, a}\,,\\
{\cal R} = \bigsqcup_a H_{2, a}\,,
\end{split}
\end{equation}
with $\sqcup$ representing the disjoint union of sets.

It is now clear that there is a limitation of partitioning the vertices in the triangulation according to colors if we try to identify a trisection naively.
Moreover, the information on $\cal D$, although formally present,  appears to be 
implicit and  hidden in the construction. 
In previous works \cite{bell2017}, as we briefly mentioned, these problems have been tackled in two different ways. In \cite{bell2017}, the authors perform Pachner moves on the triangulation. The specific type of Pachner move employed ($2\to 4$ Pachner move) increases the number of  $4$-simplices in $\T$ without affecting the topology (replaces a $4$-ball with another $4$-ball having the same triangulation on the boundary). This allows to connect the spines of the four-dimensional handlebodies at will, as well as to clearly infer the structure of compression discs for all the three-dimensional handlebodies. Nevertheless, Pachner moves are not compatible with the colors in the present case, since the complete graph with six vertices cannot be consistently $5$-colored. In  \cite{Casali:2019gem}, on the other hand, the authors considered a special class of colored graphs encoding crystallizations. By definition, all $K(B^{\widehat{i}})$
are connected in crystallization theory. 
Such requirement imposes a limited amount of nodes in the graph encoding a manifold $M$, which results in a very powerful tool to study the topology of PL-manifolds\footnote{As we will explain later, the authors of \cite{Casali:2019gem} actually consider a wider class of graphs. Nevertheless they still base their construction on connectedness of some chosen $\widehat{i}$-bubble.}.  However, crystallization graphs
only reflect a small amount of cases of interest to the tensor model community. In the following section we present an alternative approach which allows to generalize the construction of trisections to a wider class of graphs.

%%%%%%%%%%%%%%%%%%%%%%%%%%%%%%%%%%%%%%%%%%%%%%%%%%%%%%%%%%%%%%%%%%%%
\subsection{Connecting $4$-bubbles}
\label{sec:connect4bubbles}
%%%%%%%%%%%%%%%%%%%%%%%%%%%%%%%%%%%%%%%%%%%%%%%%%%%%%%%%%%%%%%%%%%%%

\begin{figure}[h]
\centering
     \begin{minipage}[t]{0.8\textwidth}
      \centering
\def\svgwidth{0.7\columnwidth}
\includegraphics[scale=.18]{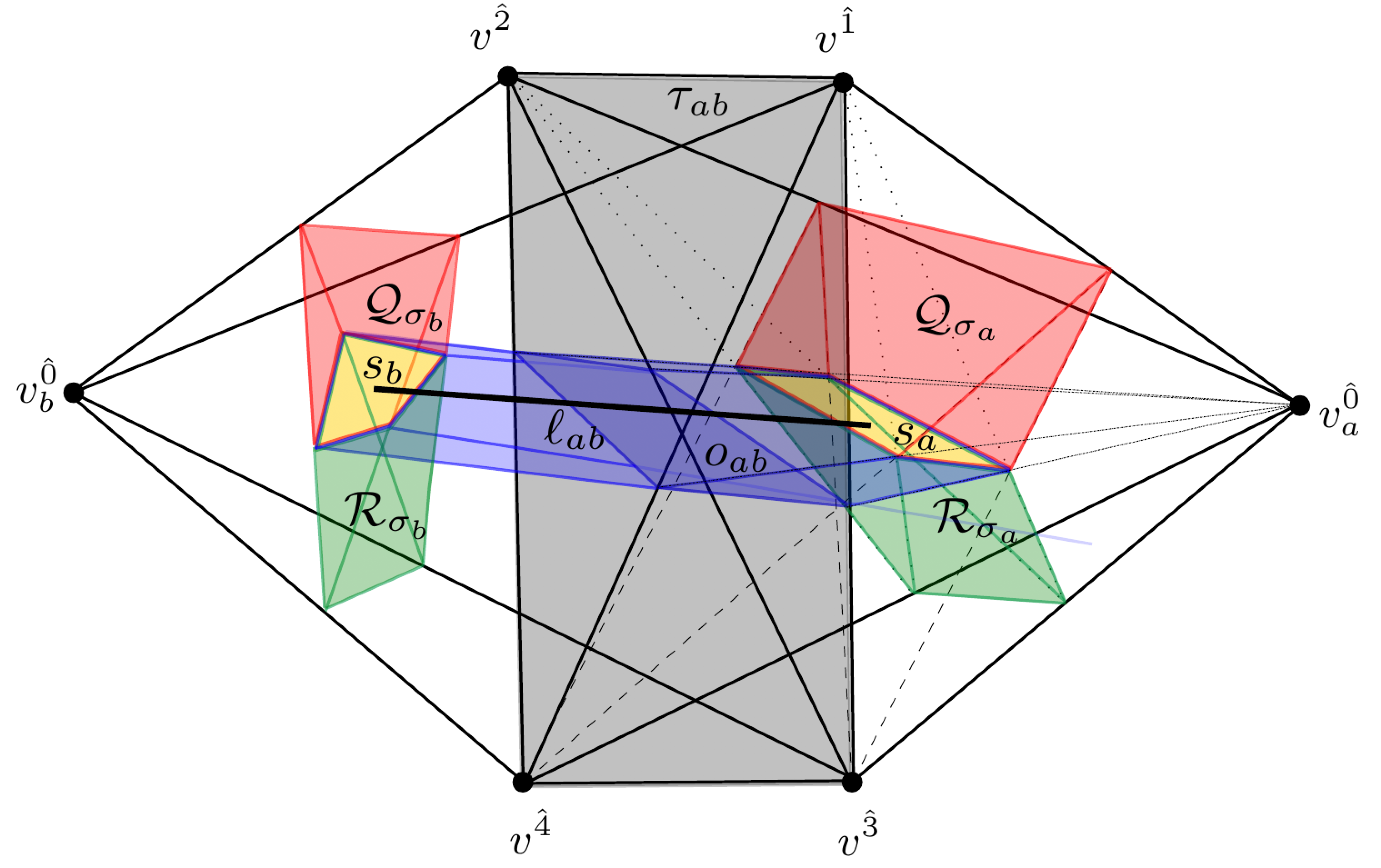}
\caption{
$\pi_{ab}$.
}
\label{fig:doublepentachoron}
\end{minipage}
\end{figure}

In order to overcome the issues coming from having disconnected  realizations of $4$-bubbles, we follow a similar 
construction of Heegaard splittings discussed in section~\ref{sec:more-heeg-split}.
Let us start by embedding the colored graph dual to $\T$ in $\T$ itself via the prescription described in section~\ref{sec:coloredtensormodels}. 
We consider four-dimensional regular neighborhoods of the $0$-colored lines embedded in $\T$.
Topologically, each such four-dimensional neighborhood $n$ is $\D^3\times \D^1$, and its boundary  is $(\Sp^2\times \D^1)  \cup (\D^3\times\Sp^0)$.
This boundary intersects ${\cal D}$ (three-dimensional) transversally and, therefore, the longitudinal component $\Sp^2\times \D^1$ is split by $\cal D$ into two parts: $\partial_+ n$ and $\partial_- n$, each of topology  $\D^2\times \D^1$. As a convention we fix $\partial_+ n$ to be between $\cal D$ and $\{v^{\widehat{1}}, v^{\widehat{2}}\}$ and $\partial_- n$ between $\cal D$ and $\{v^{\widehat{3}}, v^{\widehat{4}}\}$.

\begin{construction}
\label{const:aboutD}
Given a colored triangulation $\T$ of a manifold $M$, dual to a colored graph $\cal G$, and a choice of a jacket for its $\widehat{0}$-bubbles, ${\cal J}(\B^{\widehat{0}}_a)$, there exist three 3-submanifolds of $\T$: $\cal Q'$, $\cal R'$ and $\cal D'$, such that they share the same boundary
\begin{equation}
\Sigma = \partial{\cal Q}' = \partial{\cal R}' = \partial{\cal D}'\,,
\end{equation}
and which are constructed carving regular neighborhoods of the embedded $0$-colored lines of $\cal G$ as:
\begin{equation}
\begin{split}
{\cal Q}' &=[ {\cal Q}\setminus (\bigcup_l n)]\cup [\bigcup_l \partial_+ n]\,,\\
{\cal R}' &=[ {\cal R}\setminus (\bigcup_l n)]\cup [\bigcup_l \partial_- n]\,,\\
{\cal D}' &= {\cal D}\setminus [\bigcup_l \dot{n}]\,,
\end{split}
\end{equation}
where $l$ runs over the set of $0$-colored lines and $\dot{n}$ indicates the interior of $n$.
\end{construction}
In order to understand construction~\ref{const:aboutD}, let us consider two vertices $v_a^{\widehat{0}}$ and  $v_b^{\widehat{0}}$ sitting opposite to the same $0$-colored $3$-face, $\tau_{ab}$ and  call  $n_{ab}$, the regular neighborhood of the $0$-colored line $\ell_{ab}$ dual to $\tau_{ab}$.
We call the $4$-simplex spanned by $v_a^{\widehat 0} \cup \tau_{ab}$, $\sigma_{a}$, and similarly for $b$.

One can view the $3$-ball $\D^3$ in this four-dimensional regular neighborhood of a $0$-colored line, $n_{ab}$, as a retraction of the 
tetrahedron ${\cal Q}_{\sigma_a} \cup {\cal R}_{\sigma_a} = K(\B^{\widehat{0}}_a) \cap \sigma_a$  (or for $b$)
 inside each 4-simplex, $\sigma_{a}$ (or $\sigma_{b}$),  where ${\cal Q}_{\sigma_a}$ is ${\cal Q} \cap \sigma_a$, etc.
Using $n_{ab}$, we perform a connect sum  of  the  $3$-submanifolds defined by $4$-bubbles $K(\B^{\widehat{0}})$'s and, at the same time, perform a boundary-connect sum of the four-dimensional neighborhoods of vertices in the triangulation.

The union $\sigma_a\cup\sigma_b$ via their shared face $ \tau_{ab} $ defines a polytope{\footnote{These are called a double pentachora in \cite{bell2017}.}}  $\pi_{ab}$, spanned by $v_a^{\widehat 0} \cup \tau_{ab}   \cup v_b^{\widehat 0}$.
In each $\pi_{ab}$, there are two central squares $s_a$ and $s_b$ which are the intersections of $\pi_{ab}$ with the  realization of jackets of $\widehat{0}$-bubbles $K ({\cal J} (\B_a^{\widehat{0}}))$ and $K( {\cal J} (\B_b^{\widehat{0}}))$ respectively.
We note that 
a neighborhood of the barycenter of $s_a$ (resp. $s_b$) intersects $K(\B^{\widehat{0}}_a)$ (resp. $K(\B^{\widehat{0}}_b)$) in a $3$-ball satisfying the requirements presented in sec.~\ref{sec:connected-sum}. 
Therefore, by removing such neighborhoods and identifying their boundaries, we can easily construct the connected sum $K(\B^{\widehat{0}}_a)\,\sharp\,K(\B^{\widehat{0}}_b)$ preserving the Heegaard splitting defined by the chosen jackets, i.e., connecting the components of $\cal Q$ (resp. $\cal R$) surrounding  $v_a^{\widehat{0}}$ and $v_b^{\widehat{0}}$. 
For later convenience we require the neighborhood of the barycenter of $s_a$ (resp. $s_b$) to be small enough not to intersect $\partial s_a$ (resp. $\partial s_b$). 
Note that, by construction, this also yields $ X_1^a\,\natural\,X_1^b$. 
As we discussed in sec.~\ref{sec:connected-sum}, we can represent the boundary-connected sum of handlebodies through a line connecting the boundaries. 
This is precisely the role of $\ell_{ab}$; $X_1^{ab}=X_1^a\cup n_{ab}\cup  X_1^b$ is homeomorphic to $ X_1^a\,\natural\,X_1^b$. 
The intersections $q_a=s_a\cap n_{ab}$ and  $q_b=s_b\cap n_{ab}$ identify smaller squares splitting each $\D^3$ in $(\D^3\times\Sp^0)\subset\partial n_{ab}$.
The interiors of $q_a$ and $q_b$ now belong to the interior of $X_1^{ab}$ while their boundaries define a surface
\begin{equation}
\Sigma_{ab}=\partial q_a \times \ell_{ab} = \partial q_b \times \ell_{ab} \,.
\end{equation}
It is now straightforward to see that we just constructed the connected sum of the surfaces dual to ${\cal J}(\B_a^{\widehat{0}})$ and ${\cal J}(\B_b^{\widehat{0}})$ by simply considering the following union:
\begin{equation}
[K({\cal J}(\B_a^{\widehat{0}}))\setminus q_a] \cup \Sigma_{ab} \cup [K({\cal J}(\B_b^{\widehat{0}}))\setminus q_b]\,.
\end{equation}
With a similar construction and following the arguments of sec.~\ref{sec:connected-sum}, it is not hard to see that we also constructed the boundary-connected sums ${\cal Q}_{\sigma_a}\,\natural\,{\cal Q}_{\sigma_b}$ and ${\cal R}_{\sigma_a}\,\natural\,{\cal R}_{\sigma_b}$.
Notice that the boundary-connected sum of the three-dimensional handlebodies is made preserving the combinatorics defined by the chosen jacket, and this clarifies any ambiguity due to a choice of orientation.

Since $\ell_{ab}$ is transversal to $\tau_{ab}$, it is easy to see that it lies inside ${\cal D}_{\sigma_a}\cup{\cal D}_{\sigma_b}$.  The intersections $(n_{ab}\cap{\cal D}_{\sigma_a})\cup(n_{ab}\cap{\cal D}_{\sigma_b})$ identify what shall be carved out of ${\cal D}_{\sigma_a}$ and ${\cal D}_{\sigma_b}$. Here, we require $(n_{ab}\cap\partial{\cal D}_{\sigma_a}) = \emptyset$ (and similarly for $\partial{\cal D}_{\sigma_b}$) in order to avoid singularities. The operation is, thus, very similar to a stabilization up to the fact that we are identifying balls on the boundaries of two disconnected handlebodies. The boundary of the (three-dimensional) carved region in ${\cal D}_{\sigma_a}\cup{\cal D}_{\sigma_b}$ is, again, $\Sigma_{ab}$. Hence, $\Sigma_{ab}$ is identified as the central surface obtained through such a carving operation.

In general, there are more than one $0$-colored 3-faces sitting opposite to the same pair $v_a^{\widehat{0}}$ and $v_b^{\widehat{0}}$; we denote this number $E_{ab}$.
It means that there are $E_{ab}$-many embedded $0$-colored lines connecting the two realizations of the bubbles $K(\B_a^{\widehat{0}})$ and $K(\B_b^{\widehat{0}})$. 
Repeating the above procedure for all $E_{ab}$ lines not only defines the boundary-connected sums ${\cal Q}_a\natural{\cal Q}_b$ and ${\cal R}_a\natural{\cal R}_b$, but also adds to each of them $E_{ab}-1$ extra $1$-handles via stabilization.

We are left to clarify how ${\cal D}=\bigcup_\sigma {\cal D}_\sigma$ behaves under the iterated carving operation. 
Let us first notice that each ${\cal D}_{\sigma}$ is bounded by six rectangular faces. 
One, as we defined earlier, is $s$ and is determined by the intersection of $\sigma$ with the realization of a jacket of a $\widehat{0}$-bubble. $s$ is the only face of ${\cal D}_{\sigma}$ whose interior lies in the interior of $\sigma$. 
The interior of the other five faces lies inside the interior of one of the five boundary faces of $\sigma$. 
Hence, each boundary face of ${\cal D}_{\sigma}$ naturally carries a single color from the colored graph $\cal G$. 
The face carrying the color $0$ is the one sitting opposite to $s$ and we call it $o$. 
For every ${\cal D}_{\sigma}$ in $M$ there is one and only one ${\cal D}_{\sigma'}$ sharing $o$ with ${\cal D}_{\sigma}$. 
The union ${\cal D}_{\sigma}\cup{\cal D}_{\sigma'}$ can be thought as the effective building blocks of $\cal D$ and they are in one to one correspondence with the $0$-colored lines of $\cal G$. 
These building blocks are also bounded by ten faces; in $\pi_{ab}$, we have: $s_a$, $s_b$, four lateral faces carrying colors $i\neq 0$ coming from ${\cal D}_{\sigma_a}$ and four lateral faces carrying colors $i\neq 0$ coming from ${\cal D}_{\sigma_b}$. 
Note that faces of the same color coming from ${\cal D}_{\sigma_a}$ and ${\cal D}_{\sigma_b}$ are glued to each other via a boundary edge. 
When we compose such blocks to build $\cal D$, each block glues to another sharing a lateral face according to the colors.

\begin{figure}[h]
    \begin{minipage}[t]{0.8\textwidth}
      \centering
\def\svgwidth{0.9\columnwidth}
\begin{subfigure}{0.5\textwidth}
\centering
\includegraphics[scale=.4]{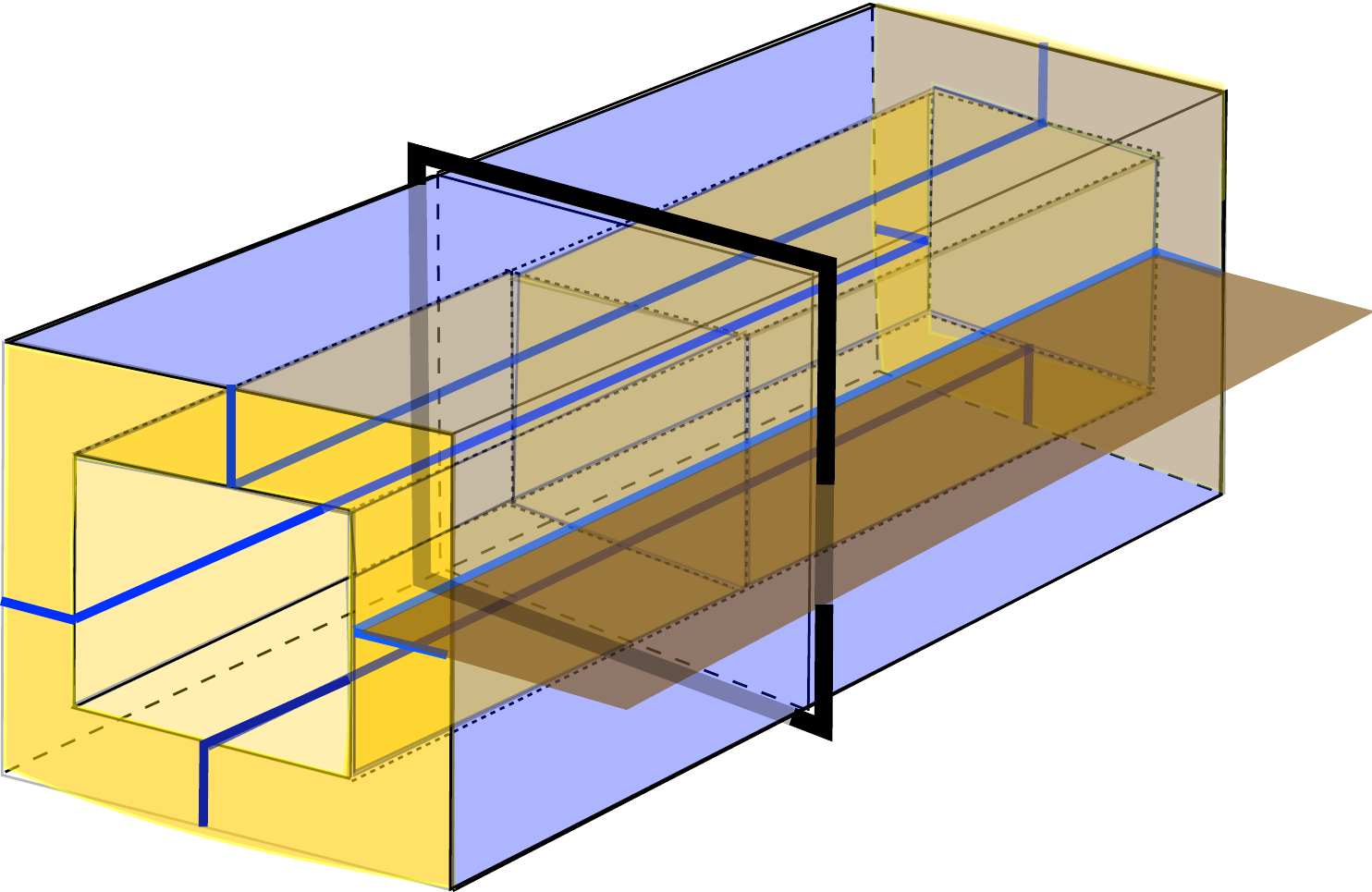}
\caption{}
\label{fig:carved-D}
\end{subfigure}%
\begin{subfigure}{0.5\textwidth}
\centering
\includegraphics[scale=.35]{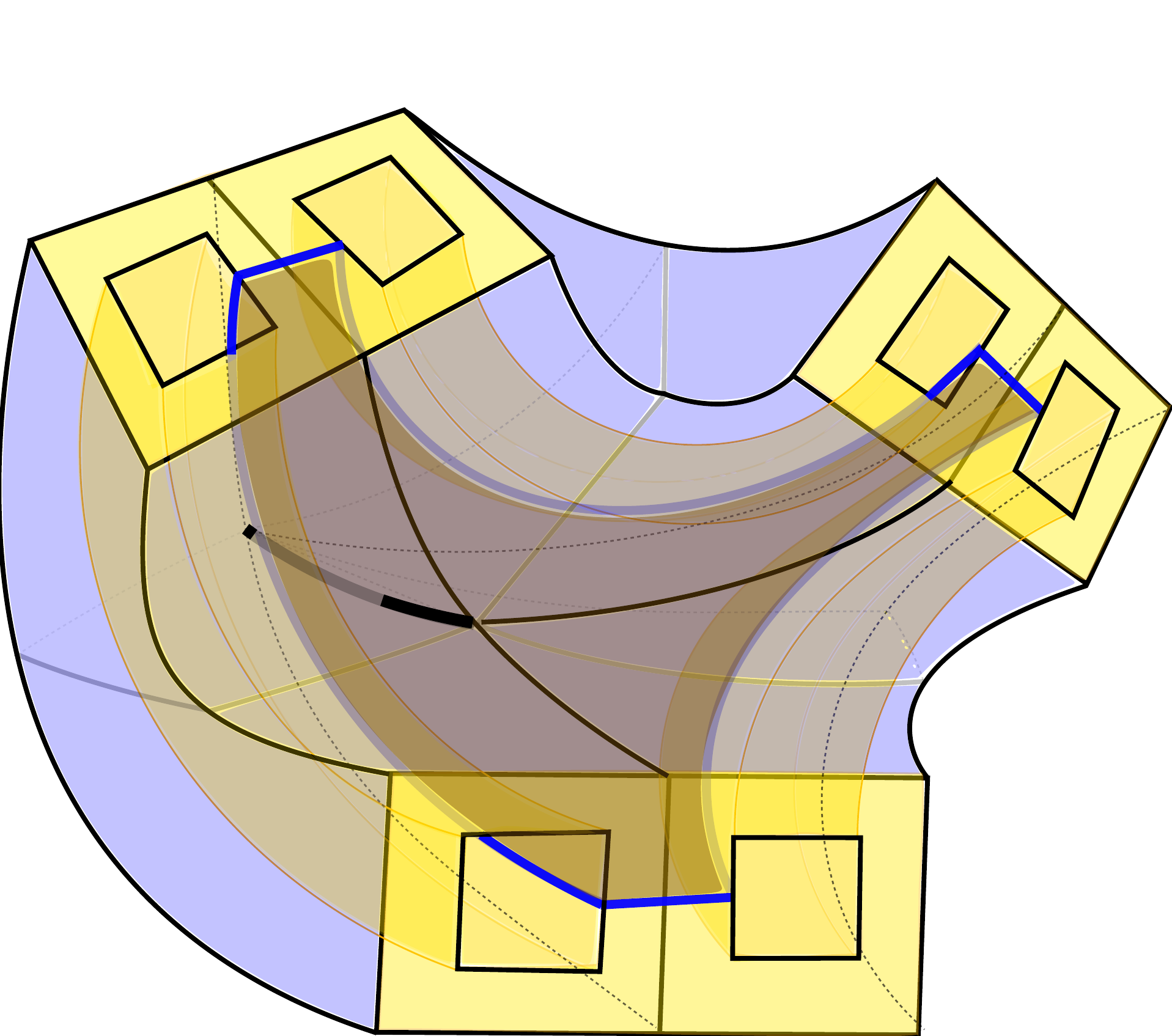}
\caption{}
\label{fig:carved-triple-D}
\end{subfigure}
\caption{Structure of $\cal D'$. Fig.~\ref{fig:carved-D} shows an effective building block of $\cal D'$, namely ${\cal D}_{ab}^{(\ell)}$.
There are eight blue $2$-faces which are to be glued to other effective building blocks of $\cal D'$. 
The yellow surface is going to be part of the central surface of the trisection and, therefore, will constitute the boundary of $\cal D'$, i.e., to be glued onto $\cal Q$ and $\cal R$.
Four pieces of the $\gamma$-curves describing $\cal D'$ are pictured in blue lines.
The brown rectangle with one of the $\gamma$-curves as boundary represents part of a compression disc.
The spine of the effective block of $\cal D'$ is shown in thick solid black loop piercing through the compression disc.
Fig.~\ref{fig:carved-triple-D} shows three effective building blocks of $\cal D'$ glued along their  $i$-colored faces (let us pick $i=1$).
Here we show a $\gamma$-curve in blue, circulating along all of the three effective building blocks and defining the boundary of a compression disc (shown in brown).
In this example, the $\gamma$-curve is defined by the color set $\{0, 1\}$.
All the other $\gamma$-curves, which we do not show, only travel through one block
and then move away on other patches of the central surface which are not shown in the picture.
As before, patches of the central surface are shown in yellow an lateral faces in blue.
}
\label{fig:carved-both-D}
\end{minipage}
\end{figure}

It is important to realize that the embedding of $0$-colored lines connects opposite faces of such building blocks, namely $s_a$ and $s_b$, therefore a tubular neighborhood of a $0$-colored line always intersects $s_a$ and $s_b$.
After carving such neighborhoods out of ${\cal D}$,
each building block  is turned into a solid torus (pictorially, we can think of tunneling through them along a $0$-colored line, see fig.~\ref{fig:carved-D}). In $\pi_{ab}$, we refer to such new effective building blocks as 
\begin{equation}
{\cal D}_{ab}^{(\ell)} = ({\cal D}_{\sigma_a}\cup{\cal D}_{\sigma_b})\setminus n_{ab}\,,
\end{equation}
and the resulting entire structure corresponds to
\begin{equation}
{\cal D}' = {\cal D}\setminus({\mathlarger{\cup}}_{\ell_{ab}} n_{ab}) = \bigcup_{\ell_{ab}} {\cal D}^{(\ell)}_{ab}={\mathlarger{\natural}}{\cal D}^{(\ell)}_{ab}\,.
\end{equation}

Before moving on, an important remark is in order. So far we discussed the case of $v^{\widehat{0}}_a$ is different from $v^{\widehat{0}}_b$. Nevertheless, it may easily happen that $\tau_{ab}$ opposes to the same $\widehat{0}$-colored vertex (in fact, it is sufficient that the two $4$-simplices in $\pi_{ab}$ share one more face, beside $\tau_{ab}$, for this to be true). In this case, as explained in section \ref{sec:stab}, most of the features we just discussed would still hold. Simply, instead of performing a connected sum between two ${\widehat 0}$-bubbles, we would be adding a $1$-handle to a single ${\widehat 0}$-bubble via stabilization (as in fig.~{\ref{fig:stabilization}})and increase by one the genus of the central surface defined by $K({\cal J}(\B_a^{\widehat{0}}))$. In particular, this situation would correspond to a single building block ${\cal D}_{ab}^{(\ell)}$ in which two lateral faces of the same color $i$ are identified. One can understand such operation as the retraction to a point of a disc on the boundary of ${\cal D}_{ab}^{(\ell)}$, bounded by a trivial element in the first homotopy group of the $2$-torus\footnote{Remember that two faces of the same color in ${\cal D}_{ab}^{(\ell)}$ already share a side.}. Topologically, such ${\cal D}_{ab}^{(\ell)}$ would therefore remain a solid torus.

We are now ready to state the main result of this work.

\begin{theorem}
\label{theor:ourtheor}
$\cal Q'$, $\cal R'$ and $\cal D'$ are handlebodies.
\end{theorem}
\begin{proof}
The submanifolds $\cal Q'$ and $\cal R'$, as explained in the construction~\ref{const:aboutD}, are stabilizations of the boundary connected sum of the handlebodies $\{{\cal Q}_a\}$ and $\{{\cal R}_a\}$ respectively and, as such, are handlebodies themselves. Their spines are defined as described in sections \ref{sec:connected-sum}, \ref{sec:jackets-heeg} and \ref{sec:stab}, i.e., via the bicolored paths defining the jacket, joined by the embedded $0$-colored lines of $\cal G$.

${\cal D}'$ is the boundary-connected sum of the building blocks ${\cal D}_{ab}^{(\ell)}$ performed via their lateral faces. Since the ${\cal D}_{ab}^{(\ell)}$ are solid tori, ${\cal D}'$  is a handlebody by construction. The prescription to perform such boundary-connected sum is encoded in the combinatorics of $\cal G$. 
Eventually, no lateral 2-face of ${\cal D}_{ab}^{(\ell)}$ will be left free (for any $a$ and $b$ in the graph ) and the only contributions to the boundary of $\cal D'$ will come from $s_a\setminus q_a$, $s_b\setminus q_b$ and $\Sigma_{ab}$ (for any $a$ and $b$). 
Its spine can be identified by noticing that each solid torus ${\cal D}^{(\ell)}_{ab}$ can be collapsed along $\ell_{ab}$ onto a $\Sp^1$ homeomorphic to the boundary of $o_{ab}$. The spine of ${\cal D}'$ can, thus, be constructed by gluing the spines of each building block\footnote{We recall that the boundary of each $o_{ab}$ face consists of four sides carrying colors $i\neq 0$.}.
\end{proof}

Let us turn our attention draw a set of $\gamma$-curves on the boundary of $\cal D'$. 
Four sectors  of compression discs can be built in each ${\cal D}_{ab}^{(\ell)}$ intersecting the central surface on $\Sigma_{ab}$ as well as on $s_a \setminus q_a$ and $s_b \setminus q_b$ (see fig.~\ref{fig:carved-D}).
The resulting four  arcs of  $\gamma$-curves correspond to arcs of four circles coplanar to the axis of revolution of the torus boundary of each ${\cal D}_{ab}^{(\ell)}$. Each arc starts from one of the sides of $s_a$ (determined by a color $i\neq 0$), proceed along $\Sigma_{ab}$ (therefore parallel to a $0$-colored line of $\cal G$), and end on the side of $s_b$ carrying the same color as the side they started from, as depicted in fig.~\ref{fig:carved-both-D}.
Here, each arc will connect to another one coming from a neighboring building block of $\cal D'$.
Thanks to the combinatorics of $\cal G$, inherited by the building block of $\cal D'$,  the composition of a $\gamma$-curve through the union of such arcs will go on according to the $\{0i\}$-colored cycles
in the graph and close after as many iteration as the half of the length of the $\{0i\}$-cycle.
Therefore from each $0$-colored tetrahedron 
$\tau_{ab}$, four $\gamma$-curves depart each going around a boundary triangle. 
We remark here that this procedure will give us redundant $\gamma$-curves.

We conclude this section by simply performing the following identifications with respect to our definition \ref{def:trisection}:
\begin{equation}
\begin{split}
H_{23} &= {\cal D}' \,,\\
H_{12} &= {\cal Q}'\,,\\
H_{13} &= {\cal R}'\,.
\end{split}
\end{equation}

\begin{figure}[h]
    \begin{minipage}[t]{0.8\textwidth}
      \centering
\def\svgwidth{0.9\columnwidth}
\centering
\includegraphics[scale=.4]{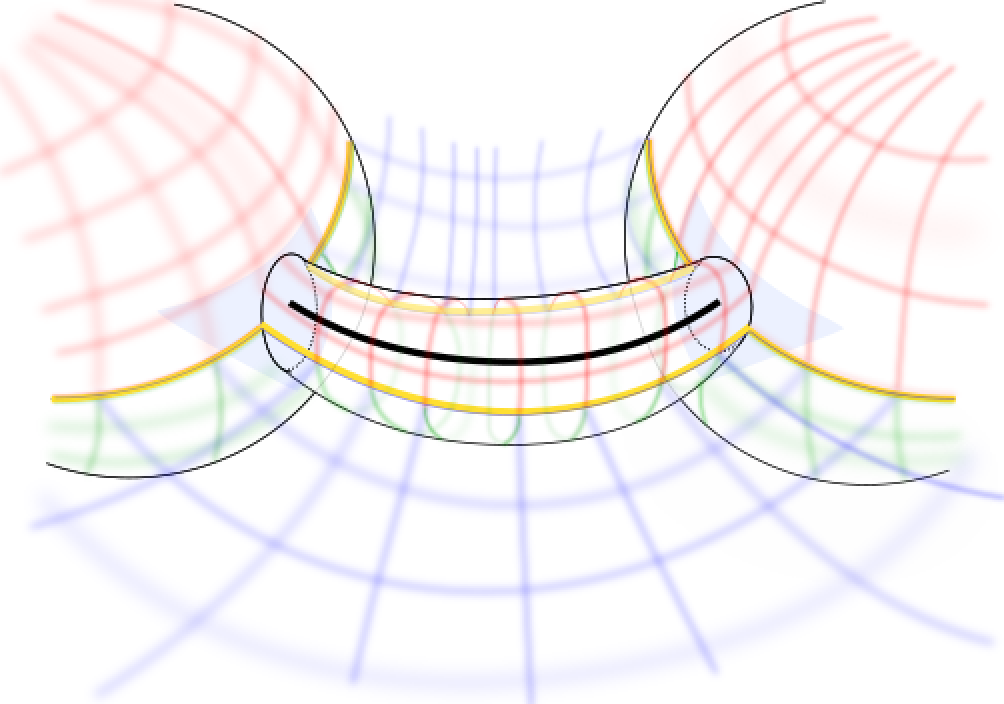}
\caption{
We illustrate how we connect the isolated dual of $4$-bubbles via the carving operation explained in sec.~{\ref{sec:connect4bubbles}}.
For simplicity, the figure is an analogue in lower three dimensions rather than four.
Two $\Sp^2$s on the right and on the left represent $K(B^{0}_a)$ and $K(B^{0}_b)$ and they are connected via a tubular neighborhood of $\ell$ (here represented with a solid black line).
Part of $H_{12}$ is shown as a red surface, $H_{13}$ in green, and $H_{23}$ in blue. 
Part of the central surface to be is depicted as lines in yellow.
The three-dimensional space above the red and the blue surfaces is an analogue of $X_2$, the one below is $X_3$, whereas the tubular neighborhood $\ell$ and the spheres constitute $X_1$.
The light blue triangles represent ${\cal Q}_{\sigma_a} \cup {\cal R}_{\sigma_a}$ and ${\cal Q}_{\sigma_b} \cup {\cal R}_{\sigma_b}$. (See fig.~{\ref{fig:doublepentachoron}}.)
}
\label{fig:bubblesconnected}
\end{minipage}
\end{figure}

%%%%%%%%%%%%%%%%%%%%%%%%%%%%%%%%%%%%%%%%%%%%%%%%%%%%%%%%%%%%%%%%%%%%
\subsection{Four-dimensional handlebodies}
\label{sec:4dhandlebodies}
%%%%%%%%%%%%%%%%%%%%%%%%%%%%%%%%%%%%%%%%%%%%%%%%%%%%%%%%%%%%%%%%%%%%

Let us briefly comment on the four-dimensional pieces $X_1$, $X_2$, and $X_3$ we obtained with our prescription. 
As we discussed at the beginning of section~\ref{sec:trisections}, theorem~\ref{th:extending-th} implies that
there is a unique cap-off of ${\cal D}'\cup{\cal Q}'\cup{\cal R}'$, i.e., there is a unique way of defining $X_1$, $X_2$, and $X_3$ using only $3$- and $4$-handles such that the pairwise unions  ${\cal D}'\cup{\cal Q}'$, ${\cal D}'\cup{\cal R}'$ and ${\cal Q}'\cup{\cal R}'$, are the boundaries of $X_1$, $X_2$, and $X_3$.
Due to the symmetric nature of $i$-handles and $(d-i)$-handles in $d$ dimensions, all  $X_1$, $X_2$, and $X_3$ are guaranteed to be handlebodies. 
The statement, therefore, is equivalent to saying that there is a unique set of three handlebodies with the given boundaries. 
Nevertheless one might wonder whether, given a triangulation, these handlebodies actually reconstruct the PL-manifold or not. 
In fact, embedding ${\cal D}'$, ${\cal Q}'$ and ${\cal R}'$ in the triangulation as we illustrated above provides us with three four-dimensional submanifolds $\overline{X}_1$, $\overline{X}_2$ and $\overline{X}_3$. 
These manifolds share the same boundaries as $X_1$, $X_2$, and $X_3$ but they are a priori different. If that were the case, $\overline{X}_1$, $\overline{X}_2$ and $\overline X_3$ would automatically not be handlebodies due to the aforementioned uniqueness. In order to clarify this point we look for the spines of $\overline{X}_1$, $\overline{X}_2$ and $\overline{X}_3$.
\begin{corollary}
\label{corollary:4dhbody}
Given a colored triangulation $\T$ of a manifold $M$, dual to a colored graph $\cal G$, and a choice of a jacket for its $\widehat{0}$-bubbles, ${\cal J}(\B^{\widehat{0}}_a)$, construction \ref{const:aboutD} defines a trisection of $M$.
\end{corollary}

\begin{proof}
Since the three-dimensional handlebodies ${\cal Q}'$, ${\cal D}'$ and ${\cal R}'$ satisfy the hypothesis of definition \ref{def:trisection} by construction (i.e., they share the same boundary and their interiors are disjoint), we can focus on the four-dimensional submanifolds $\overline{X}_1$, $\overline{X}_2$ and $\overline{X}_3$. Their interior is disjoint by construction, therefore the only issue is to prove that they are handlebodies.
$\overline{X}_1$ is bounded by ${\cal Q}'\cup {\cal R}'$. 
Its spine is easily found by collapsing $K(\B_a^{\widehat{0}})$ to points \footnote{For the moment we are only dealing with manifolds rather than pseudomanifolds therefore this just represents the retraction of a topological ball to its center.} and keeping the connection
encoded by $0$-colored embedded lines. 
Therefore $\overline{ X}_1$ is a handlebody by construction.
$\overline{ X}_2$ is bounded by ${\cal Q}'\cup {\cal D}'$.
Bearing in mind the linear map from a $\Delta^{(4)}$ to $\Delta^{(2)}$ as in section \ref{sec:cutting-simplices} 
(fig.~\ref{fig:4simplexmap}), 
we notice that in every four-simplex, $\overline{X}_2$ can be retracted to an edge identified by the set of colors $\{1, 2\}$ via its endpoints: $v^{\widehat{1}}$ and $v^{\widehat{2}}$. 
The set of these edges therefore constitutes a spine of $\overline{X}_2$.  
Moreover, $\overline{X}_2$ is connected since its boundary $\partial\overline{X}_2$ is connected by construction. 
This is enough to prove that $\overline{X}_2$ too is a handlebody.
The argument for $\overline{X}_3$ follows in complete analogy with the one for $\overline{X}_2$ upon replacing the set of colors $\{1, 2\}$ with $\{3, 4\}$ and the boundary  $\partial\overline{X}_2 = {\cal Q}'\cup {\cal D}'$ with $\partial\overline{X}_3 ={\cal R}'\cup {\cal D}'$.

The uniqueness of the handlebodies with the given boundary implies $\overline{X}_1 = { X}_1$, $\overline{X}_2 = { X}_2$ and $\overline{X}_3 = { X}_3$.
\end{proof}

%%%%%%%%%%%%%%%%%%%%%%%%%%%%%%%%%%%%%%%%%%%%%%%%%%%%%%%%%%%%%%%%%%%%
\subsection{Central surface and trisection diagram}
\label{sec:central-surface}
%%%%%%%%%%%%%%%%%%%%%%%%%%%%%%%%%%%%%%%%%%%%%%%%%%%%%%%%%%%%%%%%%%%%

In this present section, we discuss the trisection diagram encoded in what we illustrated in section \ref{sec:trisections}.

Let us slowly reveal the topological information somewhat deeply hidden in our construction.
From our construction, in general, the genus of the central surface will not coincide with the trisection genus.
In a rare case the genus of the central surface is equal to the trisection genus,  one could imagine it being a very special type of triangulation and is suppressed in the statistical theory dictated by the tensor model. 
This is not necessarily a dramatic problem, provided that there is a clear understanding of $\alpha$-, $\beta$- and $\gamma$-curves. 
This information of curves, however, is also not necessarily trivial to extract since we generate many copies of the same curve which, in principle, intersect other curves on the diagram differently and choosing one curve over the other corresponds to a different diagram with the same central surface\footnote{Therefore connected by a series of handle slides and by as many handle addition as handle cancellations.}. 
Nevertheless we are hopeful that future works might unentangle this information and overcome this ambiguity.

To start, we look at the genus of the central surface. 
Let us define the following graph $\widetilde{\cal G}$ derived from a colored graph $\cal G$.
Starting from the original colored graph $\cal G$, we collapse all the $\widehat{0}$-bubbles to points which will become the  
nodes of $\widetilde{\cal G}$. Then, we connect these nodes  via the $0$-colored lines of $\cal G$ encoding
the same combinatorics of the original graph $\cal G$.
Effectively,  the $0$-colored lines of $\cal G$ simply become the lines of $\widetilde {\cal G}$.
Note that the number of connected components of a graph is preserved under this operation; if $\cal G$ is connected, $\widetilde{\cal G}$ is connected.
The number of loops\footnote{We refer here to the notion of loops of a graph that is commonly used in physics in the framework of Feynman diagrams, not to the graph theoretical notion of a line connecting a node to itself. 
What we refer to as loop is, in graph theory, sometimes referred to as independent cycles.} of $\widetilde {\cal G}$ corresponds to the dimension of its first homology group and evaluates to:
\begin{equation}
L=\vert {\cal E} \vert-\vert {\cal V} \vert +1\,,
\end{equation}
where $\vert {\cal E} \vert$ is the number of lines and $\vert {\cal V} \vert$ is the number of 
nodes of $\widetilde {\cal G}$.
By construction, $\vert {\cal V} \vert$ corresponds to the number of different  $\widehat{0}$-bubbles which, in turn, is the number of vertices
opposing to $0$-colored tetrahedra.
$\vert \cal E \vert$, on the other hand, corresponds to the number of $0$-colored tetrahedra and evaluates to $p/2$ for a triangulation of $p$ simplices\footnote{Note that we are considering only orientable manifolds and, therefore, the original graph $\cal G$ is bipartite.}.
\begin{proposition}
Construction \ref{const:aboutD} defines a trisection with a central surface $\Sigma$ of genus $g_c$ given by
\begin{equation}
\label{eq:central-genus}
g_c = \sum_{a=1}^{\vert \cal V \vert} g_{{\cal J}(\B_a^{\widehat{0}})} +L\,,
\end{equation}
with $g_{{\cal J}(\B_a^{\widehat{0}})}$ being the genus of the jacket ${\cal J}_{\{1, 2\}\{3, 4\}}$ of the bubble $\B_a^{\hat{0}}$.
\end{proposition}
Notice that $g_c$ is invariant under the insertion of $d$-dipoles in the $0$-colored lines, while inserting a $d$-dipole in a line of color $i \ne 0$ increases $L$, and therefore $g_c$, by one. In fact, as we show in the appendix \ref{app:examples}, the 
elementary melon 
yields the genus $1$ trisection diagram for $\Sp^4$ and the insertion of a $d$-dipole can be understood as the connected sum with 
 the elementary melon
at the level of the colored graph.

Let us look at the curves we have drawn on $\Sigma$. We remark that the genus $g_c$ also corresponds to the number of independent $\alpha$-, $\beta$- and $\gamma$-curves.
The $\gamma$-curves are obtained as paths on $\Sigma$ and composed by segments parallel to the  lines of $\widetilde {\cal G}$, and segments crossing the boundaries between different $s$'s, according to an associated color $i\neq 0$. 
The composition of these segments according to the combinatorics of ${\cal G}$ will force the $\gamma $ curve to close in a loop (see fig.~\ref{fig:carved-triple-D}).
This tells us that the $\gamma$-curves are isomorphic to embedded $\{0 i\}$-cycles  in $\T$.
Note that by representing the graph $\cal G$ in stranded notation, these curves are literally drawn on the surface\footnote{Note that every vertex of $\cal G$ corresponds to a square in the surface dual to ${\cal J}(\B^{\widehat{0}})$ and the $0$-colored embedded lines are interpreted as handles. 
Therefore the $\{0 i\}$-strand is really isomorphic to one of the $\gamma$-curves.}

Similarly, given a chosen jacket ${\cal J}_{\{i, j\}\{k, l\}}$ the $\alpha$- and $\beta$-curves are given by the $(i, j)$- and $(k, l)$-strands of $\cal G$ (see fig.~\ref{fig:trisection-curves}). 
Furthermore, we shall add one $\alpha$- and one $\beta$-curve for every 
 line $\widetilde{\cal G}$. 
These last additions correspond to the attaching cuves of the Heegaard splitting of $\Sp^1 \times \Sp^2$ in the genus one trisection diagram of $\Sp^4$ (see section~\ref{sec:stab}).

\begin{figure}[h]
    \begin{minipage}[t]{0.8\textwidth}
      \centering
\def\svgwidth{0.3\columnwidth}
\centering
\includegraphics[scale=.35]{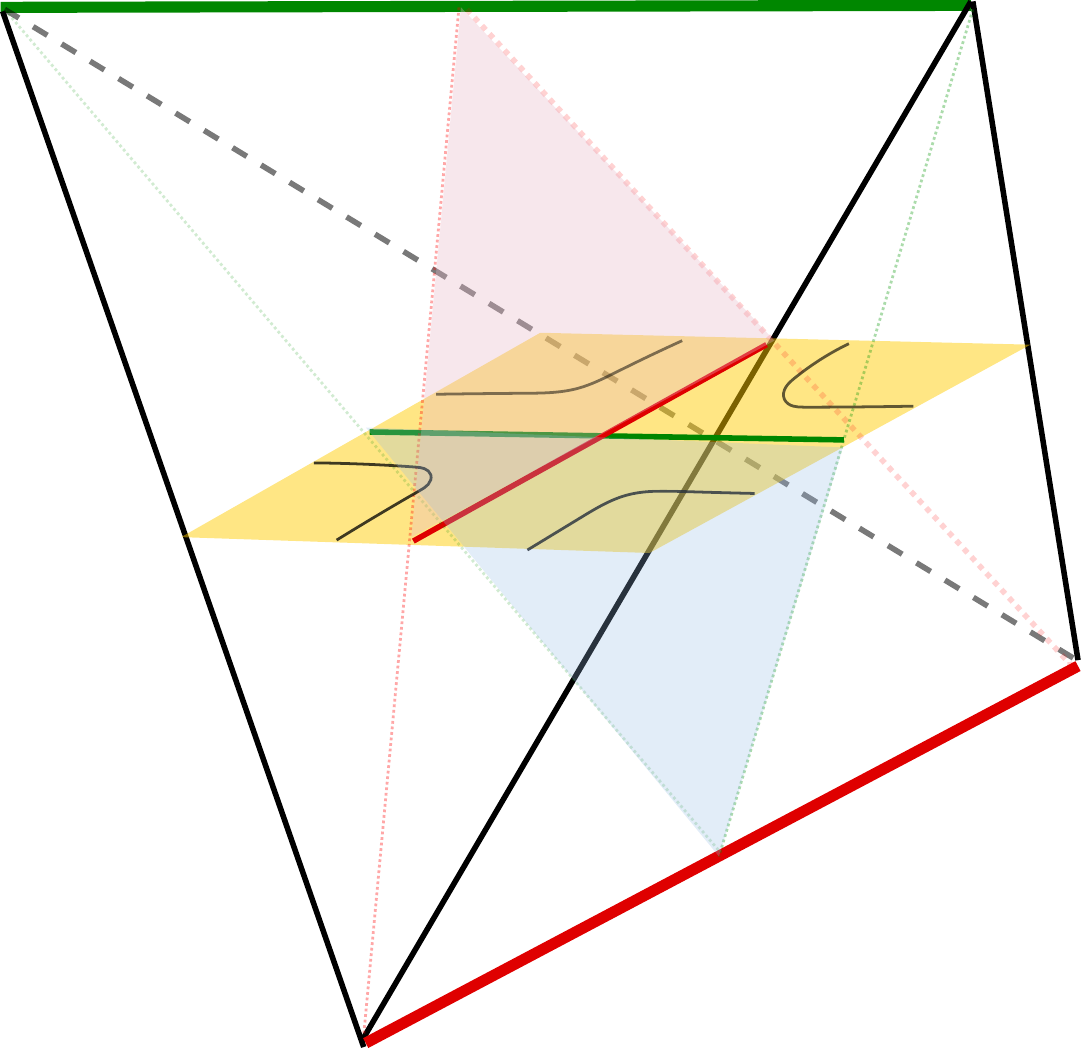}
\caption{
This sub-$3$-simplex is an element in the triangulation dual to a $\widehat{0}$-bubble.
We can identify $\alpha$-curve (red line)  and $\beta$-curve (green line) on the central surface (yellow square) by projecting the edges (red and green) of the sub-$3$-simplex sitting opposite to the central surface down to the central surface itself. These edges become part of the spine of the corresponding three-dimensional handlebodies.
Part of compression discs are shown in pink (light blue) which is bounded by an $\alpha$-curve ($\beta$-curve).
We highlight the matrix model dual to a jacket ${\cal J}(\B^{\widehat{0}})$.  The yellow  square which is part of central surface is nothing but the quadrangulation dual to a jacket ${\cal J}(\B^{\widehat{0}})$. 
We illustrate that the trisection curves ($\alpha$-curve in red and $\beta$-curve in green) coincide with drawing the strands of the original graph directly on the quadrangulation of the central surface.
}
\label{fig:trisection-curves}
\end{minipage}
\end{figure}

As we stated above, not all these curves are independent. Each of them is a viable attaching curve, 
but not all of them should be considered at the same time.
For the $\alpha$-curves (and similarly for the $\beta$-curves), we can constrain slightly more; 
the independent ones should be chosen to be $g_{{\cal J}(\B_a^{\widehat{0}})}$-many in each realization of a  $\widehat{0}$-bubble plus $L$-many among the extra ones we draw around the now embedded lines of $\widetilde {\cal G}$ (up to Heegaard moves). 
Remember that attaching curves of a graph are defined by the condition that cutting along them we obtain a connected punctured sphere (see fig.~\ref{fig:cutting}). 
$L$ is by construction the maximal number of lines we can cut before disconnecting the graph $\widetilde {\cal G}$.
Once these first $L$ curves are cut, we can proceed identifying the rest of the $\alpha$-curves given by each of the $\vert \cal V \vert$-many $\widehat 0$-bubbles through ${\cal J}(\B^{\widehat{0}})$. 

So far, we have treated color $0$ to be special, however, of course that is an arbitrary choice for an easy illustration, and any other color choice will suffice.
Hence, there are $15$ possible trisections (up to handle slides) that can be generated with our construction ($5$ choices of $4$-bubbles and $3$ choices of jackets per each choice of $4$-bubble).

A final remark is in order. If we compare our results with the one presented in \cite{Casali:2019gem}, the genus of the central surface we obtain is obviously higher and less indicative of the topological invariant. A more striking difference is that we have an extra combinatorial contribution. By construction, and due to the properties of the graphs considered, the result presented in \cite{Casali:2019gem} is only affected by the Heegaard splitting of an embedded $3$-manifolds, in particular, the Heegaard splitting of the link of a vertex. Moreover, for a closed compact $4$-manifold $M$, such link is always PL-homeomorphic to $\Sp^3$.
We can, thus, understand the trisection genus of a manifold $M$, which is a smooth invariant, as a lower bound for the possible Heegaard splittings of embedded spheres induced by colored triangulations of $M$. In our construction, though, an extra contribution to the genus of the central surface is produced in the form of $L$ in equation \eqref{eq:central-genus}. One may wonder whether this contribution is actually necessary or just an artifact of our construction of trisections. In other words, if the relevant topological information could indeed be rephrased in terms of Heegard splittings of embedded $3$-manifolds, it might be enough to consider the connected sum of the realizations of $4$-bubbles, without systematically stabilizing the trisection with $L$ extra of $1$-handles.

\subsection{Singular manifolds}
\label{sec:pseudo-mfd}

What we have discussed so far strictly applies only to manifolds, i.e., to graphs where all  $\widehat{i}$-bubbles are dual to PL-spheres. Nevertheless, colored graphs generated by a colored tensor model of the form \eqref{def:simplicial-tensor-model} encode pseudo-manifolds as well.
It is natural to wonder whether our construction might encode any sensible topological information for such wider class of graphs.
In \cite{Casali:2019gem} such an extension has been made clear starting from crystallization graphs. 
We will follow similar steps in order to extend the same construction beyond graphs encoding closed compact manifolds.

Let us restrict to the case of $\widetilde{M}=K({\cal G})$ being singular manifolds.
Then, all the $\widehat{i}$-bubbles are dual to PL-manifolds and the singularity is only around vertices in $\T$ (rather then higher dimensional simplices). 
One can obtain a compact manifold $M$ out of $\widetilde{M}$ by simply removing open neighborhoods of the singular vertices in $\T$. The number of connected components of $\partial M$ will increase by the number of singular vertices with respect to the number of connected components of $\partial \widetilde{M}$. Conversely, one can obtain a singular manifold by coning all the boundary components of a manifold with (non-spherical) boundary. If $\cal G$ is a closed graph, then the above correspondence is a bijection between the set of manifolds with non-spherical boundary components and singular manifolds. 

Though such bijection allows us to work with manifolds in a larger class of graphs,
the definition of trisections as formulated in definition~\ref{def:trisection} only applies to closed manifolds. Hence, we shall extend it to include boundary components in order to connect with our combinatorial construction. Following \cite{Casali:2019gem} we define a \textit{quasi-trisection} by allowing one of the four-dimensional submanifolds not to be a handlebody:

\begin{definition}\label{def:quasi-trisection}
Let $M$ be an orientable, connected 4-manifold with $n$ boundary components $\partial M_1\,, \dots\,,\partial M_n$. A quasi-trisection of $M$ is a collection of three submanifolds  $X_1, X_2, X_3 \subset M$ such that:
\begin{itemize}
\item each $X_1$ and $X_2$ are four-dimensional handlebodies of genus $g_1$ and $g_2$ respectively,
\item $X_1$ is a  compression body with topology ${\mathlarger{\natural}}_{r=1}^n (\partial M_r\times [0,1]) \bigcup_{s=0}^{g_1}{\rm h}_s$, ${\rm h}_s$ being $1$-handles,
\item $X_i$s have pairwise disjoint interiors $\partial X_i\supset (X_i\cap X_j)\subset\partial X_j$ and $M= \cup_i X_i$,
\item the intersections  $X_i\cap X_j = H_{ij}$ are three-dimensional handlebody,
\item the intersection of all the four-dimensional handlebodies $ X_1 \cap X_2 \cap X_3$ is a closed connected surface $\Sigma$ called \textbf{central surface}.
\end{itemize}
\end{definition}

Let us further denote with $G_s^{(0)}$ the set of connected $5$-colored graphs with only one $\widehat{0}$-bubble and with all $\widehat{i}$-bubbles dual to topological spheres, and let us denote with $\xbar{G}_s^{(0)}$ the set of connected $5$-colored graphs whose only non-spherical bubbles are $\widehat{0}$-bubbles (but we do not restrict the number of such bubbles). Obviously, an element in $\xbar{G}_s^{(0)}$ describes a manifold that can be decomposed into the connected sum of realizations of elements of $G_s^{(0)}$. The connected sum, in this case, can be performed at the level of two graphs ${\cal G}_1$ and ${\cal G}_2$ by cutting a $0$-colored line in each graph and connecting the open lines of ${\cal G}_1$ to the open lines of ${\cal G}_2$. The construction of trisections we illustrated in the previous sections can be straightforwardly applied to graphs in $\xbar{G}_s^{(0)}$ and is easy to see that the outcome satisfies the conditions in def.~\ref{def:quasi-trisection}. In this regard, the result is the simplest generalization of the result presented in \cite{Casali:2019gem}. A more complicated extension would require the inclusion of singular vertices defined by different color sets; we leave such study for future works.

%%%%%%%%%%%%%%%%%%%%%%%%%%%%%%%%%%%%%%%%%%%%%%%%%%%%%%%%%%%%%%%%%%%%
\section{Conclusions}
\label{sec:conclusions}
%%%%%%%%%%%%%%%%%%%%%%%%%%%%%%%%%%%%%%%%%%%%%%%%%%%%%%%%%%%%%%%%%%%%

We have formulated trisections in the colored triangulations encoded in colored tensor models, restricting to the ones which are realized by manifolds (as opposed to pseudo-manifolds). 
We utilized the embedding of colored tensor model graphs in their dual triangulations to facilitate our construction of trisections.
Generally speaking, the genus of the central surface of the trisection, given a colored tensor model graph, is higher as the graph is bigger (i.e., the number of nodes is larger).
Therefore, statistically speaking, it is unlikely to obtain the trisection genus (which is a topological invariant) of the corresponding manifold of a given colored tensor model graph. Nevertheless, it would be interesting to investigate whether the construction of trisections might lead to new insights on the organization of the partition function of colored tensor models.

With the Gurau degree classifying tensor model graphs, we can achive a large $N$ limit, where we only select the dominating melonic graphs which are a subclass of spheres. Melons in the continuum limit have been shown to behave like  branched polymers  with Hausdorff dimension $2$  and the spectral dimension $4/3$ \cite{Bonzom:2011zz, Gurau:2013cbh}.
Reflecting and motivated by the quantum gravity context, we dream of a possibility of finding a new parameter for colored tensor model which may classify the graphs in a new large $N$ limit, which may then give some new critical behavior.
There have been works in this direction \cite{Bonzom:2012wa, Bonzom:2015axa, Bonzom:2016dwy, BenGeloun:2017xbd}, where the authors studied how to achieve different universality classes than the melonic branched polymer (tree). 
In \cite{Lionni:2017yvi}, given random discrete spaces obtained by gluing families of polytopes together in all possible ways, with a systematic study of different building blocks, the author achieved the right scalings for the associated tensor models to have a well-behaved $1/N$ expansion.
So far, one could achieve in addition to the tree-like phase, a two-dimensional quantum gravity planar phase, and a phase transition between them which may be interpreted as a proliferation of baby universes \cite{Lionni:2017xvn}.
In \cite{Valette:2019nzp}, they have defined a new large $N$ expansion parameter, based on an enhanced large $N$ scaling of the coupling constants. These are called generalized melons, however, this class of graphs is not yet completely classified, and it is not proven yet what kind of universality class they belong to in the continuum limit, but strong hints point toward branched polymers.
In our present case, knowing that in rank $3$, the realisation of a jacket is identified to be a Heegaard surface, and knowing that jackets govern the Gurau degree which is responsible for the melonic large $N$ limit, 
it is tempting to delve further into the possibility of finding a specific parameter for rank $4$ colored tensor model based on trisections which may classify the graphs in the large $N$ limit.
Our next hope is to explore possibilities around trisections to find such a parameter. 

Looking at the structure of equation \eqref{eq:central-genus} and its properties under $d$-dipoles insertion/contraction we expect melons to persist in dominating the large $N$.
Nevertheless, a different parameter of topological origin might be induced by the above construction. An example is the intersection form, which we plan to investigate in the future following \cite{Feller:2016}. Hopefully investigations in this direction might shed some light on the path integral of tensor models beyond the leading order in the large $N$.

%%%%%%%%%%%%%%%%%%%%%%%%%%%%%%%%%%%%%%%%%%%%%%%%%%%%%%%%%%%%%%%%%%%%
\section*{Acknowledgements}
%%%%%%%%%%%%%%%%%%%%%%%%%%%%%%%%%%%%%%%%%%%%%%%%%%%%%%%%%%%%%%%%%%%%
We would like to thank Andrew Lobb for giving us a lecture on Morse theory, for supervising us on a study on trisections and for other discussions while he was visiting OIST as an excellence chair in the OIST Math Visitor Program. We would also like to thank David O'Connell for leading study sessions on trisections with us. Furthermore, we thank Maria Rita Casali and Paola Cristofori as well as Razvan Gurau for checking our formulation and the manuscript.

%%%%%%%%%%%%%%%%%%%%%%%%%%%%%%%%%%%%%%%%%%%%%%%%%%%%%%%%%%%%%%%%%%%%
\appendix
\section{Examples}
\label{app:examples}
%%%%%%%%%%%%%%%%%%%%%%%%%%%%%%%%%%%%%%%%%%%%%%%%%%%%%%%%%%%%%%%%%%%%

In this section we report some particularly simple examples of trisections constructed via our procedure.

\begin{figure}[h]
    \begin{minipage}[t]{0.9\textwidth}
      \centering
\def\svgwidth{0.8\columnwidth}
\centering
\includegraphics[scale=.18]{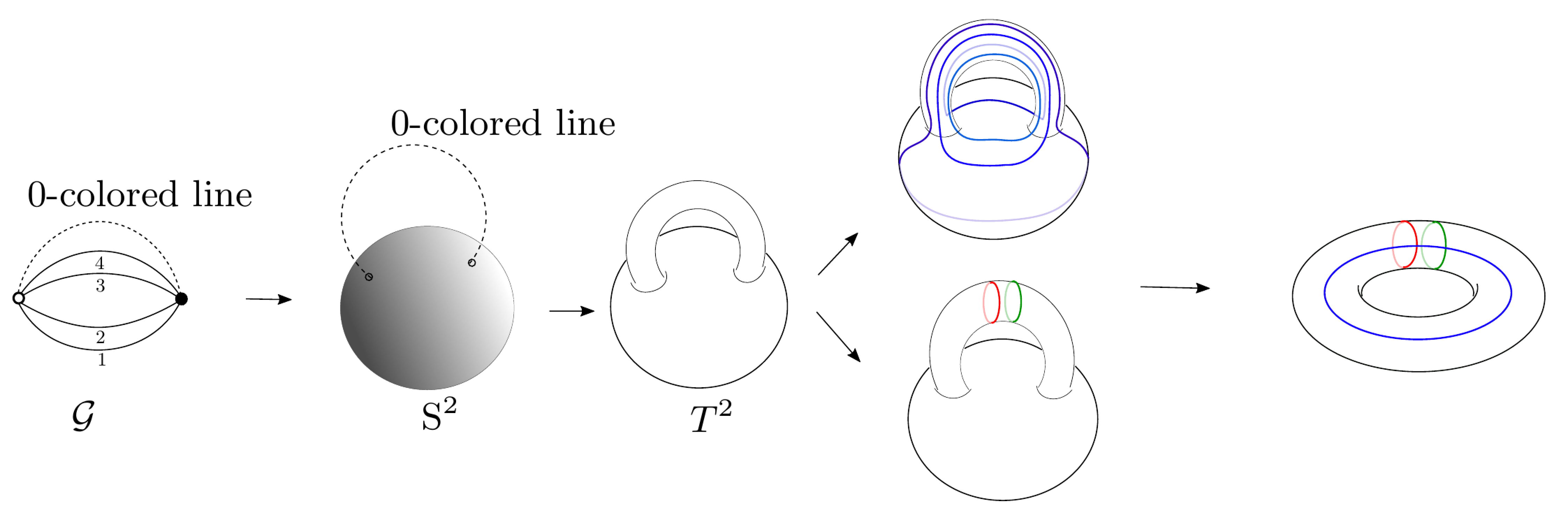}
\caption{}
\label{fig:exampleelementarymelon}
\end{minipage}
\end{figure}

The first graph we consider is the elementary melon, shown in fig.~\ref{fig:exampleelementarymelon}. This is the simplest graph we can draw and consists of only two nodes sharing all the lines. In fact, this is the graph corresponding to the crystallization of $\Sp^4$. Due to the melonic nature of this graph, we know that all the jackets are spheres. Also, all the bubbles are melons as well. Therefore, it affords the perfect playground to understand advantages and disadvantages of the procedure presented in sec.~\ref{sec:connect4bubbles}, as well as the differences with the work presented in \cite{Casali:2019gem}. As we know from the smooth case, the trisection genus of $\Sp^4$ is $g_{\Sp^4}=0$. Following the work in \cite{Casali:2019gem}, the trisection genus can be directly computed through the jackets of a bubble $\B^{\widehat{i}}$. Since all the bubbles are melons as well, their jackets have indeed genus $g_{\cal J}=0$. Following our construction, though, we add an extra handle to the central surface following the $i$-colored line. As shown in fig.~\ref{fig:exampleelementarymelon}, this step comes with the introduction of attaching curves. Following the conventions of the main text, we one $\alpha$-curve and one $\beta$-curve parallel to each other (red and green in the figure), and four $\gamma$-curves which collapse to the same one (in blue). As anticipated, the result is one of the genus one trisection diagrams for $\Sp^4$ that can be used to stabilize a trisection diagram.

\begin{figure}[h]
    \begin{minipage}[t]{0.9\textwidth}
      \centering
\def\svgwidth{0.8\columnwidth}
\centering
\includegraphics[scale=.2]{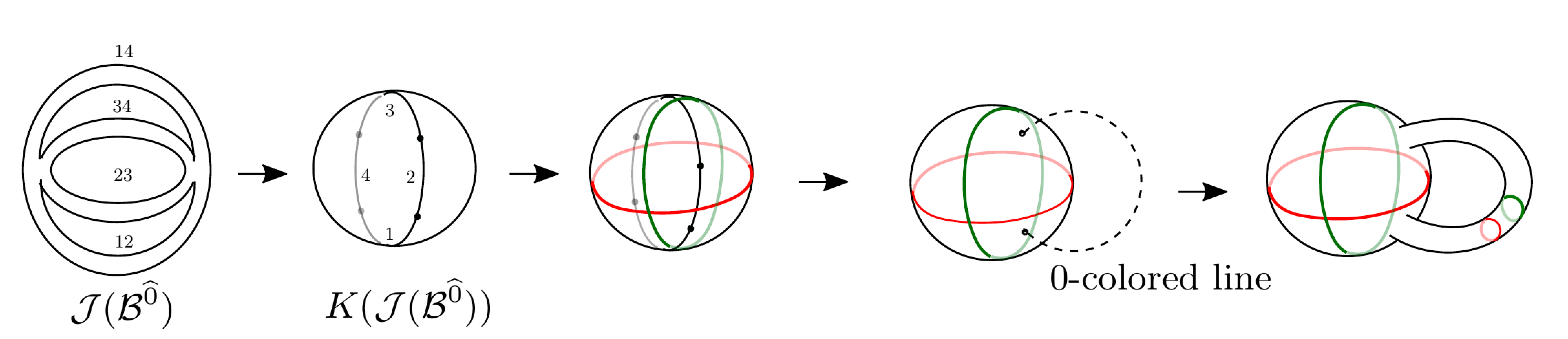}
\caption{}
\label{fig:example2}
\end{minipage}
\end{figure}

Fig.~\ref{fig:exampleelementarymelon} does not take into account the attaching curves coming from the jacket. This can be justified by the fact that the jacket is spherical and, therefore, every closed curve on it is homotopically trivial. Nevertheless, one may wonder whether retaining such curves until the end of the construction gives rise to further possibilities. In this example, we see easily from fig.~\ref{fig:example2} that the curves obtained by the spherical jacket of $\B^{\widehat{0}}$ end up being either redundant or trivial.

\begin{figure}[H]
    \begin{minipage}[t]{0.9\textwidth}
      \centering
\def\svgwidth{1\columnwidth}
\centering
\includegraphics[scale=.2]{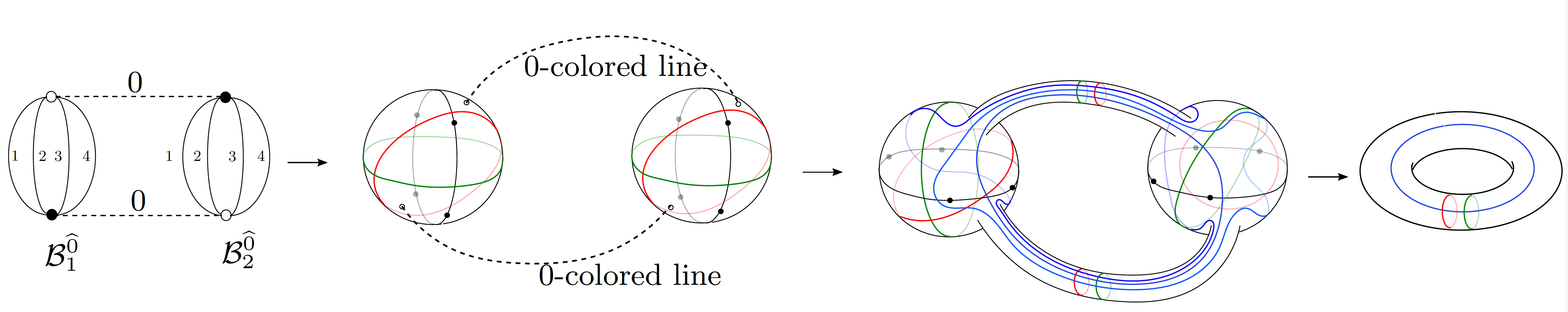}
\caption{}
\label{fig:example3}
\end{minipage}
\end{figure}

Another interesting example is given by the pillow diagram. This diagram is melonic and results from inserting a $d$-dipole into the elementary melon. Holding on to our choice of having $0$ as the special color, we have two possible ways of inserting such a dipole: inserting the dipole in a $0$-colored line or inserting it in a $i$-colored line for $i\neq 0$. As discussed in sec.~\ref{sec:central-surface}, such choices lead to different results. In fig.~\ref{fig:example3}, one can see how inserting a dipole in a $0$-colored leads to the same diagram we found before. Fig.~\ref{fig:example4}, on the other hand, shows the construction of a trisection diagram with genus $g=2$, due to the insertion of a $d$-dipole in an $i$-colored line of the elementary melon, for $i\neq 0$. Here, we observe that, up to isotopy, we can obtain five different trisection diagrams of genus $2$ for the sphere $\Sp^4$. Nevertheless, one may notice that, as expected, all these diagrams are connected by an appropriate handle slide. In fact, $(2)$ is obtained starting from $(1)$ by a handle slide of one of the blue curves, $(3)$ is obtained from $(1)$ by a handle slide of one of the red curves, while $(4)$ and $(5)$ are different handle slides a red curve in $(2)$. Note that the handle slides of the red curves relate a curve coming from the quadrangulation dual to the jacket with one of those added through stabilization.

\begin{figure}[H]
    \begin{minipage}[t]{0.9\textwidth}
      \centering
\def\svgwidth{0.75\columnwidth}
\centering
\includegraphics[scale=.2]{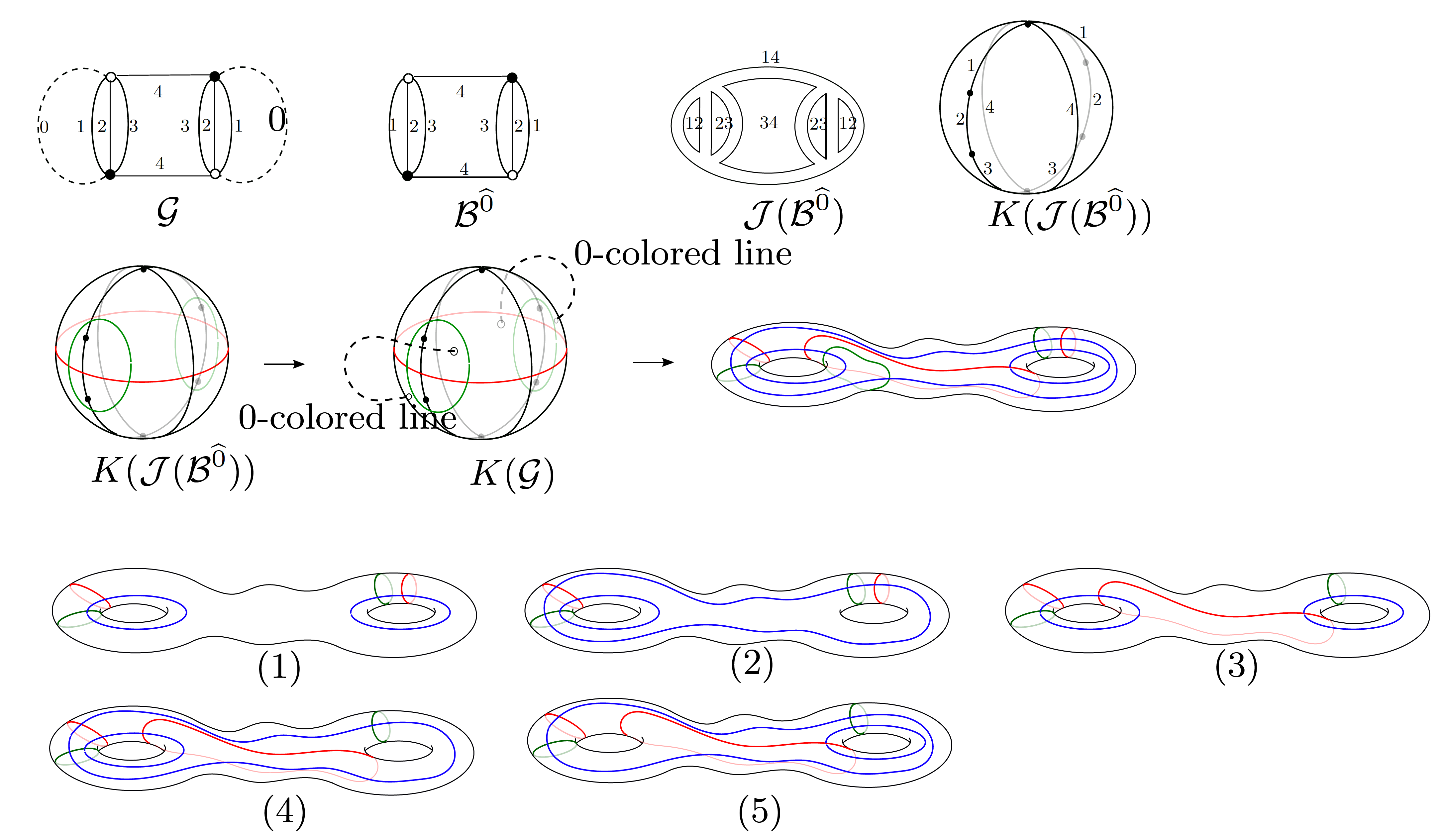}
\caption{}
\label{fig:example4}
\end{minipage}
\end{figure}

\begin{figure}[H]
    \begin{minipage}[t]{0.9\textwidth}
      \centering
\def\svgwidth{0.75\columnwidth}
\centering
\includegraphics[scale=.3]{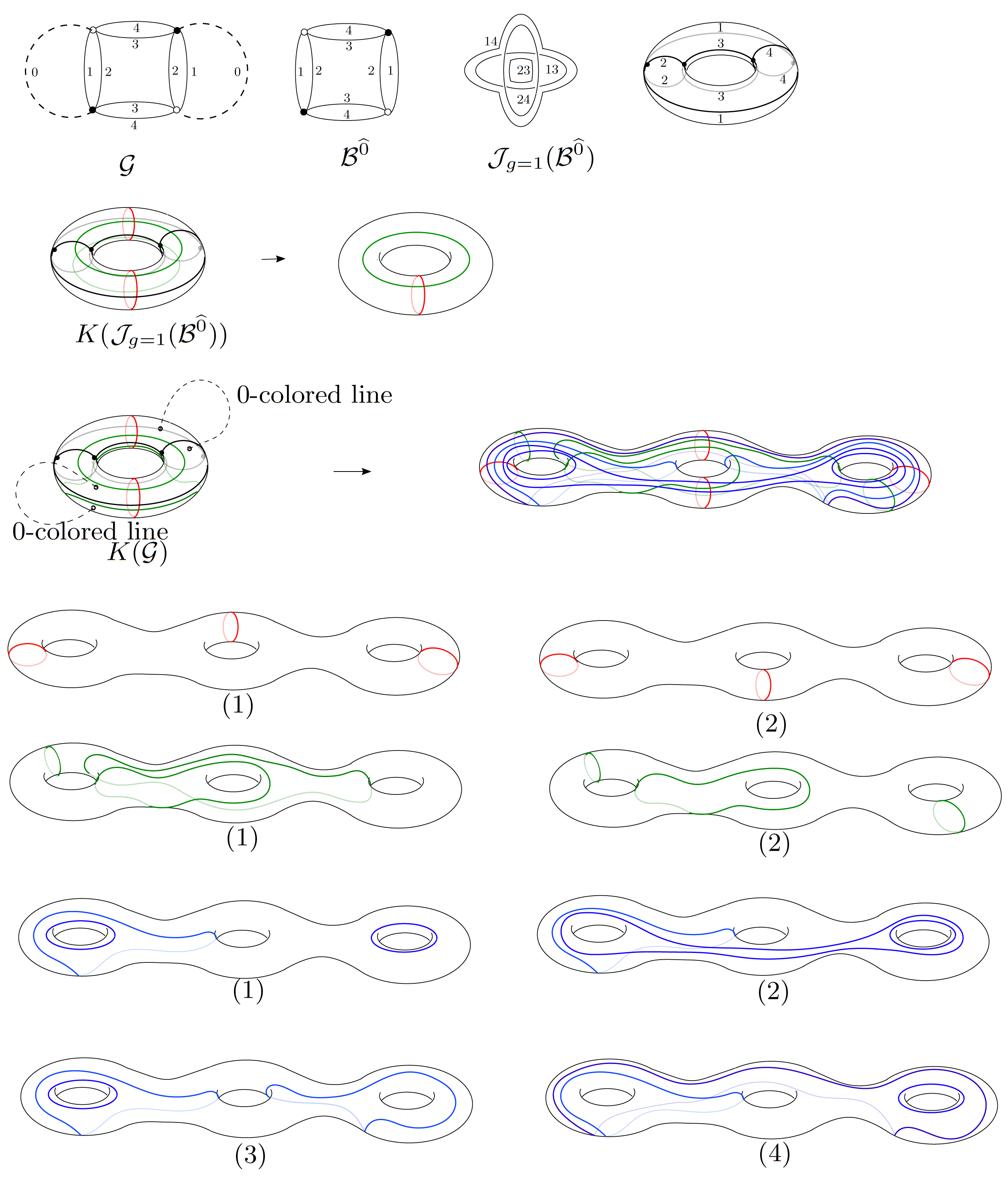}
\caption{}
\label{fig:example5}
\end{minipage}
\end{figure}

Finally, let us look at the graph shown in fig.~\ref{fig:example5}. Though not melonic, this diagram still corresponds to a sphere. out of the fifteen possible choices for constructing a central surface, we look at the less trivial one. In fact, removing either the color $3$ or the color $4$, leads again to a surface of genus $2$. On the contrary, removing the color $0$ (or equivalently the colors $1$ or $2$), leads to a necklace-like bubble. Such bubbles have one jacket which is not spherical, but rather dual to a torus. This is the choice we consider in the example shown in fig.~\ref{fig:example5}.
Let us note how, already for such a simple graph, we obtain a huge proliferation of redundant attaching curves, leading to different trisection diagrams. In particular, looking at the different ways we have to choose the attaching curves in this example, we obtain sixteen possible trisection diagrams of genus $3$ for the sphere $\Sp^4$, all out of a single combinatorial choice (out of fifteen possible choices).

%\newpage

\end{document}